\documentclass[12pt]{article}

\usepackage[english]{babel}

\usepackage{amscd,amsfonts,amssymb,amsthm,amscd,amsmath,latexsym}
\usepackage{mathrsfs}
\usepackage{cite}
\usepackage{graphicx}
\usepackage{epsfig}

\def\ie{\textit{i.e.\ }}
\def\etc{\textit{etc}}

\newtheorem{theorem}{Theorem}[section]
\newtheorem{proposition}[theorem]{Proposition}
\newtheorem{lemma}[theorem]{Lemma}

\newtheorem{definition}[theorem]{Definition}

\newtheorem{faktum}[theorem]{\it \hspace*{-5pt}}

\newtheorem{define}[]{Definition}

\def\cB{{\cal B}}
\def\cC{{\cal C}}

\def\cF{{\cal F}}

\def\cH{{\cal H}}

\def\cK{{\cal K}}

\def\cN{{\cal N}}
\def\cO{{\cal O}}
\def\cP{{\cal P}}

\def\cS{{\cal S}}
\def\cT{{\cal T}}

\newcommand{\tild}[1]{\mathbf{\Tilde{#1}}}            
\newcommand{\ttild}[1]{\mathbf{\Tilde{\Tilde{#1}}}}   

\newcommand{\tcC}{\tild{\cal C}}                      
\newcommand{\ttcC}{\ttild{\cal C}}                    

\newcommand{\tAl}{\tild{\Al}}                         
\newcommand{\ttAl}{\ttild{\Al}}                       

\newcommand{\tsigma}{\tild{\sigma}}                   
\newcommand{\ttsigma}{\ttild{\sigma}}                 

\newcommand{\quer}[1]{\overline{#1}}                  
\newcommand{\qquer}[1]{\overline{\overline{#1}}}      

\newcommand{\qsigma}{\quer{\sigma}}                   
\newcommand{\qqsigma}{\qquer{\sigma}}                 

\newcommand{\qAl}{\quer{\Al}}                         
\newcommand{\qqAl}{\qquer{\Al}}                       

\def\CC{{\mathbb C}}

\def\NN{{\mathbb N}}

\def\RR{{\mathbb R}}
\def\TT{{\mathbb T}}

\newcommand{\ba}{{\mbox{\boldmath $a$}}}

\newcommand{\bl}{{\mbox{\boldmath $l$}}}

\newcommand{\bu}{{\mbox{\boldmath $u$}}}
\newcommand{\bv}{{\mbox{\boldmath $v$}}}
\newcommand{\bx}{{\mbox{\boldmath $x$}}}

\def\rest{\upharpoonright}

\newcommand{\MI}{M} 
\newcommand{\MIO}{\cO} 
\newcommand{\MIOc}{\cO^c} 
\newcommand{\MIC}{\cC} 
\newcommand{\MICc}{\cC^c} 
\newcommand{\FL}{{V}} 
\newcommand{\FLC}{\overline{V}} 
\newcommand{\Trans}{\overline{V}_+} 
\newcommand{\Spec}{\overline{V}_+} 

\newcommand{\Lore}{{\cal L}_+^\uparrow} 
\newcommand{\Poin}{{\cal P}_+^\uparrow} 
\newcommand{\Semi}{{\cal S}_+^\uparrow} 

\newcommand{\tPoin}{\widetilde{\cal P}_+^\uparrow} 
\newcommand{\tSemi}{\widetilde{\cal S}_+^\uparrow} 

\newcommand{\tU}{\widetilde{U}} 

\newcommand{\HY}{\mathsf{H}} 
\newcommand{\HB}{\mathsf{B}} 
\newcommand{\ha}{{a}} 
\newcommand{\hb}{{b}} 

\newcommand{\HC}{\mathsf{C}} %
\newcommand{\HO}{\mathsf{O}} 
\newcommand{\HK}{\mathsf{K}}  
\newcommand{\HL}{\mathsf{L}}  

\newcommand{\zeit}{{\mbox{\small{$\theta$}}}}
\newcommand{\szeit}{{\scriptstyle \theta}}
\newcommand{\boost}{{\mbox{\small{$\beta$}}}}
\newcommand{\sboost}{{\scriptstyle \beta}}

\newcommand{\Al}{\mathfrak A}  
\newcommand{\AlV}{\mathfrak A(\FL)} 
\newcommand{\AlO}{\mathfrak A(\MIO)} 

\newcommand{\AlC}{\mathfrak A(\MIC)} 

\newcommand{\AlCc}{\mathfrak A(\MIC^{\, c})} 
\newcommand{\Ri}{\mathfrak R}  
\newcommand{\RiV}{\mathfrak R (\FL)}  
 
\newcommand{\RiC}{\mathfrak R(\MIC)} 

\newcommand{\compos}{\raisebox{1.2pt}{$\, \scriptscriptstyle \circ \, $}}
\newcommand{\product}{\raisebox{1.2pt}{$\, \scriptscriptstyle \bullet \, $}}

\begin{document} 

\title{\Large 
\vspace*{-15mm} 
\bf New Light on Infrared Problems: \\ Sectors, Statistics, 
Symmetries and Spectrum}
\author{\large { Detlev Buchholz\,$^a$ 
\ and \ John E.\ Roberts\,$^b$\thanks{Supported in part by the 
ERC Advanced Grant 227458 OACFT 
"Operator Algebras and Conformal Field Theory", 
PRIN-MIUR and GNAMPA-INDAM}}\\[4mm]
${}^a$ Institut f\"ur Theoretische Physik and 
Courant Centre \\
``Higher Order Structures in Mathematics'', 
Universit\"at G\"ottingen, \\ 37077 G\"ottingen, Germany  \\[2mm]
${}^b$ Dipartimento di Matematica, 
Universit\`a di Roma ``Tor Vergata'' \\
        00133 Roma, Italy}
\date{\large }
\maketitle

\hspace{5mm} Dedicated to Roberto Longo on the occasion of 
his sixtieth birthday\\[2mm] 
\begin{abstract}  \noindent
A new approach to the analysis of the physical state space 
of a theory is presented within the general setting of local 
quantum physics. It also covers theories with long range forces, such as 
Quantum Electrodynamics. Making use of the notion of charge class,
an extension 
of the  concept of superselection sector, 
infrared problems are avoided by restricting the states
to observables localized in a light cone. The charge structure of 
a theory can be explored in a systematic manner. 
The present analysis focuses on 
simple charges, thus including the electric charge. It is shown
that any such charge has a conjugate charge.
There is a meaningful concept of statistics:
the corresponding charge classes are either of Bose
or of Fermi type. The family of simple charge classes 
is in one--to--one correspondence with the irreducible 
unitary representations of a compact Abelian group. 
Moreover, there is a meaningful
definition of covariant charge classes.  
Any such class determines a continuous unitary
representation of the Poincar\'e group or its covering group
satisfying the relativistic spectrum condition.
The resulting particle aspects are also briefly discussed.  
\end{abstract}

\newpage

\addtolength{\textheight}{3mm}  
\addtolength{\voffset}{-3mm}    
\section{Introduction} 
\setcounter{equation}{0}
Algebraic quantum field theory \cite{Ha} has proved to 
be a powerful framework for understanding superselection 
structure in local quantum physics, \ie the possible patterns of
coherent subspaces (sectors) of the physical state space. For theories 
in Minkowskian spacetime, describing 
exclusively states of compactly localizable
charges~\cite{DoHaRo1} or massive particles~\cite{BuFr},  
this structure encoded in the local observables is fully understood: 
the sectors in any such theory correspond to the dual of some compact
group, interpreted as a global gauge group; each sector
carries a specific representation of the permutation group, 
determining its statistics; and there are charged
field operators transforming as tensors under the global
gauge group, satisfying Bose or Fermi commutation relations 
at spacelike distances and interpolating between the various 
sectors~\cite{DoRo}. These results have 
been extended successfully to theories in lower dimensions, 
where quantum groups can appear in the description of the
superselection structure and sectors can have 
braid group statistics~\cite{FrGa,FrReSch,Lo1}. Yet the physically 
important class of theories describing long range forces 
mediated by massless particles, such as Quantum Electrodynamics, 
is not covered by these results. In fact, attempts to 
clarify their superselection structure \cite{Fr,BuDoMoRoSt1}
have failed, cf. \cite{BuDoMoRoSt2}. 

The difficulties in these theories can be traced back to two 
related features, namely the long range effects of the forces 
between charged particles and the multifarious ways in 
which their interaction can form clouds of low energy massless 
particles. In Quantum Electrodynamics, for example,  
any configuration of electrically charged particles gives rise to 
a specific long range structure of the electromagnetic field,
created, both by the Coulomb fields of the moving charges
and by the infrared clouds caused by their acceleration.
The resulting long range tails of the electromagnetic  field 
can be discriminated by central sequences of local observables, 
hence giving rise to an abundance of different superselection 
sectors~\cite{Bu1,Bu2}. 

Whereas the electric charge carried by particles  
is a superselected quantity of fundamental physical
significance, the long range features of the infrared clouds 
are a theoretical concept defying experimental
verification. When applying the theory, this is usually taken 
into account by considering inclusive processes where 
expectation values are summed over the undetected low 
energy photons. This procedure effectively amounts to 
wiping out marginal features of the infrared 
sectors without destroying physically relevant 
information, such as their total charge. 
In other words, the notion of superselection sector 
provides too fine a resolution of the physical state space
leading to unnecessary theoretical complications. 
It seems desirable to use concepts closer to actual  
experimental practice and providing a coarser
resolution of the physical state space. 
 
The aim of the present investigation is to establish such 
a framework and discuss its theoretical implications.
Our approach is based on the fundamental fact that
both measurements and the preparation of specific 
states can extend well into the future, whereas
it is impossible to perform these operations in the distant past.
Thus the spacetime regions where actual experiments 
can take place are at best future directed light cones $\FL$
in Minkowski space whose apices are fixed by some largely arbitrary 
initial event. Of course, these limitations on the 
spacetime localization of concrete physical operations do 
not imply that it is impossible to obtain information about 
past properties of global states. But this 
information can only be obtained indirectly by making 
measurements in these states in some accessible light cone 
$\FL$ and trying to reconstruct  their 
past values from the data on the basis of theory. 
In this context, it has to be noted 
that there is a fundamental difference between theories
describing only massive particles and theories including  
massless particles as well. In the massive case, the 
unital C*--algebras $\Al(V)$ generated by the local observables  
in any given light cone $\FL$ are generically 
irreducible \cite{SaWo}. In principle, therefore, 
all the information on the global properties of a state can 
be derived from measurements made in any 
such region $\FL$. In the presence of 
massless particles, however, the algebra $\AlV$ 
of observables localized in $\FL$ is highly reducible and 
time translations act on it as proper endomorphisms~\cite{Bu1}.
Thus, as time proceeds,  one loses information 
about past properties of the states. This constant loss is 
inevitable since, by Huygens principle,  
outgoing massless particles created in the past of 
a given light cone $\FL$ will never enter that cone. Hence
their properties cannot be determined by later measurements in $\FL$.

It is a remarkable fact that, by the same mechanism, the various 
sectors formed by infrared clouds cannot be distinguished by 
measurements in any fixed light 
cone~$\FL$; discriminating them 
requires observations in larger regions of Minkowski space. 
On the other hand, the total charge carried by massive particles,
such as the electric charge, can be determined in any given~$\FL$. 
For these particles will either eventually enter this cone, or some
of them will, in the course of time, be annihilated or created in pairs 
of conjugate charge. But the total charge is not affected and can be 
sharply determined by measurements in the given $\FL$. These 
features are relevant in the present context. They imply that
states of equal total charge differing  
only by infrared clouds created in collisions 
are mutually normal if restricted to the 
algebra of observables $\AlV$ of any given  
light cone $\FL$. In this way, an abundance of superselection
sectors coalesce to form a \textit{charge class}~\cite{Bu1}. 
Thus restricting measurements to light cones has a similar effect  
to summing expectation values over undetected 
massless particles. 

These observations suggest basing the analysis and physical interpretation of 
Minkowski space theories entirely on the algebras $\AlV$ of observables 
localized in a given light cone~$\FL$. This should avoid infrared problems 
right from the outset. Yet implementing this idea requires 
solving several conceptual problems. 

\vspace*{1mm} 
(a) \ As already mentioned, in the presence of massless particles the 
algebras $\AlV$ are reducible in any sector. 
In fact, their weak closures are generically 
factors of type~III${}_1$, as can be inferred from  
results established in \cite{Lo,Bu}. 
There is then no intrinsic way of superimposing states 
of $\AlV$ and the usual characterization of 
sectors fails. Here we need the notion
of charge class, expressed in terms of $\AlV$. 
The states on $\AlV$, belonging to a given charge class,  
are distinguished by being primary (or even pure in the absence of
massless particles) and can be transformed 
into each other by the operational effects of 
observables localized in the light cone $\FL$. More precisely,
two relevant primary states $\omega_1$, $\omega_2$ on the 
algebra $\AlV$
belong to the same charge class if  $\omega_1$
can be mapped into any given (norm) neighbourhood
of~$\omega_2$ by the adjoint action of some  
inner automorphism of $\AlV$, and \textit{vice versa}. 
This criterion generalizes the 
notion of superselection sector in the irreducible case 
and remains meaningful in the presence of
massless particles, for the inner automorphisms of
type~III${}_1$ factors act almost transitively on 
normal states~\cite{CoSt} and can be approximated
by inner automorphisms of $\AlV$ by 
the Kaplansky density theorem. 

\vspace*{1mm} \noindent  
(b) \ In Quantum Field Theory, local operations on states can 
create charges in conjugate pairs. This allows one to
pass from states in the vacuum sector in Minkowski space 
to states in a charged sector 
by using a suitable charge transfer chain to create a   
charge in some region by shifting the 
conjugate charge along some path to spacelike infinity 
(``behind the moon'') where it evades observation. 
The energetic effects of this operation can be  
controlled by not putting a sharp restriction on 
the location of the chosen path, \ie allowing it to fluctuate 
within some spacelike cone. 
The resulting limit states are in a charged sector, but, 
by locality, coincide with the initial states   
for observations in the spacelike complement 
of the chosen cone~\cite{BuFr,Bu1}. 

One can proceed similarly in the 
light cone $\FL$. To see this, regard 
$\FL$ as a globally hyperbolic spacetime 
with a metric induced from the ambient Minkowski space. 
It is foliated by hyperboloids 
(time--shells) parametrized by the proper time of 
inertial observers passing through the apex of~$\FL$. 
In the spacetime $\FL$, Minkowskian spacelike infinity is
replaced by the asymptotic 
lightlike boundary of~$\FL$. The analogues of spacelike cones
in Minkowski space are hypercones $\MIC \subset \FL$, \ie the 
causal completions of hyperbolic cones formed by geodesics 
on a given time--shell emanating from a common apex.  
(The precise definition 
of these hypercones is given in the appendix.) 
As in Minkowski space, a limiting procedure
using local unitary operators from the algebra $\AlC$ 
of some hypercone $\MIC$ creates charges in pairs,
allowing one to pass from states in the charge class 
of the vacuum to states in another charge class. 
The resulting limit operation yields states  
in the new charge class, agreeing with the initial
states for observations in the spacelike complement
$\MICc \subset \FL$ of the chosen hypercone by locality. 

Now, whereas the electrically charged sectors in Minkowski 
space depend on the direction of the 
spacelike cone used to prepare them,
this is not so for charge classes created by the 
analogous operations in the spacetime~$\FL$: 
the disjoint infrared clouds produced using  
different cones cannot be sharply discriminated 
in $\FL$. Hence the charge class 
of the resulting states does not depend on 
the choice of hypercone $\MIC$. So 
these infrared problems disappear for observers in~$\FL$. 

As will be argued later, these heuristic considerations
are reflected in the following mathematical setting. 
Denoting 
by $\RiV$ the weak closure of the algebra 
$\AlV$ relative to the vacuum class, 
our charge classes can be reached from the 
vacuum class using a morphism $\sigma:\AlV \to \RiV$ 
localized in some hypercone $\MIC$, that is $\sigma$ acts 
trivially on $\AlCc$, whereby $\MIC \subset V$ can be chosen
arbitrarily. Our analysis shows that these morphisms 
have a rich algebraic structure, familiar from the theory of 
superselection sectors.
 
\vspace*{1mm} \noindent  
(c) \  Spacetime translations in $\FL$ do not induce 
automorphisms of  $\AlV$, 
posing the question of how to define the energy content
of states on these algebras. In a massive theory this is 
not a problem,
 since the light cone algebras are 
irreducible in all sectors and the spectral
resolutions of the global energy--momentum operators
are contained in their weak closures. It then makes   
physical sense to characterize states of $\AlV$ by 
their spectral properties. These can in principle be
checked with arbitrary precision in any light cone $\FL$. 
But in the presence of massless particles the total  
energy  of a state can no longer be sharply
determined by such measurements since the energy 
carried by outgoing massless particles created in
the past of $\FL$ fluctuates in the corresponding statistical 
ensemble and cannot be deduced from measurements in~$\FL$. 
When analyzing the energy content of states one  
can, however, exploit the fact that the semigroup of 
future-directed time translations  
acts as endomorphisms on each algebra $\AlV$. This 
allows one to introduce a notion of covariance 
for charge classes and to establish the existence of
selfadjoint generators for the semigroup action. 
Even though these generators cannot be interpreted as 
energy observables since they include gross  
fluctuations of the energy in $\FL$,  
like the Liouvillians in quantum statistical 
mechanics, they contain relevant information on 
the energy content of states. In fact, as 
massless particles in the past of $\FL$ are not  taken
into account, one expects on physical grounds that 
they provide (fuzzy) lower bounds on the total energy  
of states and are bounded from below, like the global
energy. 

\vspace*{1mm}
The relevant notions above are all that is required 
for analyzing charge classes. No global information is needed, 
only the algebras of 
local observables for fixed $\FL$ and the endomorphic action of 
the semigroup $\Semi \doteq \Trans \rtimes \Lore$ on $\AlV$, 
where   
$\Trans$ denotes the closed semigroup of future directed 
translations and $\Lore$ the group  of proper orthochronous
Lorentz transformations on $\FL$. 
This input suffices to \mbox{characterize} the vacuum state on $\AlV$ 
and to determine its properties. In particular, one can 
\mbox{establish} the existence of a continuous unitary representation 
of the full Poincar\'e group $\Poin$ in its GNS--representation 
satisfying the relativistic 
spectrum condition whose restriction to the semigroup
$\Semi$ induces the given endomorphic action on~$\AlV$. 
States in the charge class of the vacuum are induced by 
vectors in this GNS--representation. 

The analysis of all other charge classes can be based on the above 
hypercone localized morphisms $\sigma : \AlV \rightarrow \RiV$. 
In this paper we restrict our attention to simple charge classes 
and simple morphisms, to be defined later. This important special 
case (which includes the electric charge) simplifies the exposition 
whilst retaining the novel features. Equivalent simple morphisms 
carry the same (simple) charge. Using simple morphisms and an 
appropriate version of Haag duality \cite{Ha} in the spacetime $\FL$,  
one concludes that there is a composition law for the simple charges 
making the set of simple charges into an Abelian group 
where the inverse is charge conjugation.
A simple charge and its conjugate  either obey Bose or Fermi statistics. 
More precisely, exchanging such a pair of  charges is described by 
charge transport operators which are $1$ or $-1$ respectively, 
whenever the charges are localized in spacelike separated
hypercones. Finally, the group of simple charges is  
the dual of some compact 
Abelian group, the global gauge group, when all charges are simple.
Thus the general structure of these  
families resembles that of the simple sectors of localizable
charges in Minkowski space~\cite{DoHaRo2}. 

As already mentioned, there is a meaningful notion of covariant
charges. One requires 
corresponding morphisms $\sigma$ to be extendible to 
coherent families
of morphisms ${}^\lambda \sigma : \AlV \rightarrow \RiV$, $\lambda \in \Semi$, 
as explained below. It turns out that  
the set of simple covariant
charges is a subgroup of the set of simple charges.    
Moreover, each covariant simple morphism has an associated 
continuous unitary representation of the
covering group of the Poincar\'e group $\Poin$, unique up to 
equivalence, inducing the endomorphic action of $\Semi$
on $\AlV$. In accordance 
with physical expectations, these
representations are shown to satisfy the relativistic spectrum condition. 

These results do not depend on the choice of light cone $V$. In
the present approach, analyzing the state space 
of local theories provides a consistent physical 
picture even in the presence of massless particles. 
The interpretation of a theory in terms of light cone data 
reflects experimental limitations and provides a ``geometric 
regularization'' thus avoiding the spurious infrared problems 
appearing in treatments based on Minkowski space. 
This observation may also lead to a deeper
understanding of the phenomenon of quantum decoherence. 

The article is organized as follows. The   
assumptions underlying the present analysis are stated
in the next section. 
In Sect.~3 the concept of vacuum in a light cone 
is introduced and its charge class is analyzed. 
Section 4 establishes the concept of simple charge classes and 
derives their group structure and statistics.
In Sect.~5, covariant charge classes are defined; they are 
shown to have similar properties and to 
admit unitary representations of the Poincar\'e group.
Section 6 is devoted to a proof of the spectrum condition
for covariant charge classes 
and in Sect.~7 the relation between the present approach 
and the usual Minkowski space interpretation of 
quantum field theories is discussed.
The conclusions comment on the
particle aspects of the present approach and on possible future 
developments. The appendix is devoted to defining
hypercones and proving the necessary results about them.
They are referred to in the main text as \ref{A.1} to \ref{A.13}.

\section{Nets, localization and covariance} 
\setcounter{equation}{0}
\label{input}

\addtolength{\textheight}{0mm}  
\addtolength{\voffset}{5mm}    

In this brief section we introduce the basic objects treated in this
paper and the corresponding notation.

  As outlined in the introduction, 
fixing an open (forward) light cone $\FL$ in 
four--dimensional Minkowski space $\MI$ with its standard metric
$x^2 \doteq (x_0^2 - \bx^2)$, $x \in \RR^4$, we 
consider a foliation of 
$\FL$ by three--dimensional hyperboloids~$\HY$ (time--shells).
We will consider two types of causally complete regions 
in $V$, namely standard 
(relatively compact) double cones~$\MIO$ and hypercones~$\MIC$.
A hypercone $\MIC$ is the causal completion of a hyperbolic cone on some 
fixed hyperboloid~$\HY$, where a hyperbolic cone 
is a cone in the sense of hyperbolic geometry. 
The  appendix contains precise definitions and proofs of 
the principal properties of hyperbolic cones. 
If ${X}$ is a subset of $\FL$ then $X^c$, the spacelike complement of $X$
in $\FL$, denotes the interior of the set of points $y\in\FL$ such 
that $(x-y)^2 < 0$ for all $x \in X$. In particular, 
$\MIO^c$ denotes the spacelike complement of the double cone $\MIO$ 
and $\MICc$ the spacelike complement of the hypercone 
$\MIC$ respectively.  

We let $\MIO\mapsto\AlO$ be a net $\Al$ of unital 
$C^*$--algebras\footnote{We 
follow the practice of using $C^*$--algebras 
rather than von Neumann algebras. The added generality is spurious 
in that only their weak closures in the vacuum representation
will play a role.} 
on the set $\cK$ of double 
cones $\MIO \subset \FL$, ordered under inclusion. This net 
describes the observables of the underlying theory. 
The  $C^*$--inductive limit of the $\AlO$, $\MIO\in\cK$, 
is denoted by $\Al(\FL)$.
Similarly, if $\MIC$ is a hypercone, $\AlC$ 
denotes the closed subalgebra of $\AlV$ generated by the $\AlO$ 
with $\MIO\subset\MIC$ and $\AlCc$ denotes the closed subalgebra 
generated by the $\Al(\MIC_1)$ with  
$\MIC_1 \subset \MICc$.{}\footnote{In 
fact $\AlCc$ is also 
generated by the $\AlO$ with $\MIO \subset \MICc$, 
see \ref{A.4}.} 

We assume that the net $\Al$ satisfies 
locality (Einstein causality) \ie observables 
localized in spacelike separated regions commute, in short
\begin{equation} \label{locality}
[\Al(\MIO_1), \, \Al(\MIO_2)] = 0 \quad \mbox{if} \quad \MIO_1 \subset \MIOc_2 \, . 
\end{equation} 
Below, we will strengthen this to an appropriate version of 
Haag duality appropriate to the present 
geometric setting. 

The semigroup $\Semi \doteq \Trans \rtimes \Lore \subset \Poin$ 
of Poincar\'e transformations acts on $\FL$;
its general elements are denoted by 
$\lambda = (x,\Lambda)$, where $x \in \Trans$ are the translations
and $\Lambda \in \Lore$ the Lorentz transformations. This semigroup 
induces endomorphisms~$\alpha_\lambda$ of~${\Al}(\FL)$
with the obvious geometric action, 
\begin{equation} \label{covariance} 
 \alpha_\lambda (\Al(\MIO)) = \Al(\lambda \MIO ) \, , 
\quad \lambda \in \Semi \, , \ \MIO \in \cK \, .
\end{equation}

Clearly these features of the net extend canonically to the 
algebras $\AlC$ associated with hypercones $\MIC$.

\addtolength{\textheight}{5mm}  
\addtolength{\voffset}{-5mm}    

\section{Vacuum}
\setcounter{equation}{0}
\label{vacuum}

The input
specified in the preceding section suffices to identify vacuum states
on $\FL$ and to establish their characteristic properties. 

\begin{definition} A state $\omega_0$ on  ${\AlV}$ is called a
vacuum state if 
\begin{enumerate}
\item[(a)] $\omega_0 \compos \alpha_\lambda = \omega_0$ \ for  
$\lambda  \in \Semi$, 
\item[(b)]  $ \lambda \mapsto \omega_0(A^* \alpha_\lambda(B))$ is
continuous, $A,B \in   {\AlV}$,   
\item[(c)]  $x \mapsto \omega_0(A^* \alpha_x(B))$ extends 
continuously to a function on the complex domain $\Trans + i \Trans$
which is analytic in its interior and whose   
modulus is bounded {by \footnote{The explicit 
form of the bound is not needed, but it simplifies matters.}}  
$\sqrt{\omega_0(A^*A) \omega_0(B^*B)}$, \
$A,B \in {\AlV}$.
\end{enumerate} 
\end{definition}

\noindent 
The following result fully justifies this characterization of 
vacuum states in $\FL$.

\begin{proposition} Let $\omega_0$ be a vacuum state on ${\AlV}$ and
let $(\pi_0,{\cal H},\Omega)$ be its GNS representation. Then 
\vspace*{-1mm} 
\begin{enumerate}
\item[(i)] $\Omega$ is cyclic for $\pi_0(\alpha_\lambda(\AlV)$ for any  
$\lambda \in \Semi$,  
\item[(ii)] there is a continuous unitary representation 
$U_0$ of the full Poincar\'e group 
\mbox{$\Poin = \RR^4 \rtimes \Lore$} on ${\cal H}$ leaving $\Omega$
invariant and inducing the endomorphic action of the semigroup, 
$\text{\rm Ad}\,U_0(\lambda) \compos \pi_0 = \pi_0 \compos \alpha_\lambda$, 
$\lambda \in \Semi$, 
\item[(iii)] the spectrum of $U_0 \rest \RR^4$ is contained in the
closed forward light cone $\Spec$.
\end{enumerate}
\end{proposition}

\noindent \textbf{Remark.} 
This result shows that it is meaningful for an 
observer in $\FL$ to talk about the energy--momentum content of the 
states in ${\cal H}$. It should be noticed, however, that the 
unitaries $U_0 \rest \RR^4$ are not contained in the 
weak closure of $\pi_0(\AlV)$ in the presence of massless 
particles, so there is no global observable determining this
energy--momentum  content. Instead, 
$U_0 \rest \RR^4$ determines the
energy--momentum content relative to 
a state (the vacuum) whose energy--momentum 
is not precisely known as information about processes 
in the past of $\FL$ is lacking. 
This is similar to the situation with KMS--states. 

\begin{proof} (i) \ As the light cone is invariant
under Lorentz transformations, it suffices to prove the 
first statement for the semigroup of translations $x \in \Trans$. 
Property~(c) of vacuum states implies that 
the functions $x \mapsto \pi_0(\alpha_x(A)) \Omega$,
$A \in {\AlV}$, extend continuously to vector--valued
functions  which are analytic in the interior of 
the domain $\Trans + i \Trans$. Now if $\Psi \in {\cal H}$ is a vector
in the orthogonal complement of $\pi_0(\alpha_y({\AlV})) \Omega$
for some $y \in \Trans$, it follows from isotony 
and covariance of the net that 
$(\Psi, \pi_0(\alpha_x(A)) \Omega) = 0$ for 
all $x \in \Trans + y$. The edge-of-the-wedge theorem implies
$(\Psi, \pi_0(\alpha_x(A)) \Omega) = 0$ for all $x \in \Trans$
and hence $(\Psi, \pi_0(A) \Omega) = 0$, $A \in {\AlV}$.
Thus $\Psi = 0$ 
since the GNS vector $\Omega$ is cyclic for $\pi_0(\AlV)$, 
proving the first part.

\addtolength{\textheight}{-5mm}  
\addtolength{\voffset}{5mm}    

\noindent 
(ii) \ Making use of property (a) of vacuum states, one can consistently
define  
isometries $U_0(\lambda)$, $\lambda \in  \Semi$ on ${\cal H}$ by putting 
$U_0(\lambda) \pi_0(A) \Omega = \pi_0(\alpha_\lambda(A)) \Omega$,
$A \in {\Al}(\FL)$.  These isometries are unitary since their range is dense 
by the preceding step.
Moreover, by construction, they induce the endomorphic action 
of the semigroup $\Semi$ and are weakly continuous by
property (b) of vacuum states. 

To prove that  $U_0$ extends to the full Poincar\'e group $\Poin$
we note that $U_0 \rest \Lore$ already defines a representation
of the subgroup of Lorentz transformations.
So we need only consider the
subgroup of translations $\RR^4$. Given $x,y,v,w \in \Trans$ with
$x - y = v - w$ one has  $U_0(x)U_0(w) = U_0(x+w) = U_0(v)U_0(y)$.
So these unitary operators commute and $U_0(x)U_0(y)^{-1} = U_0(v) U_0(w)^{-1}$.  
Hence, as any $z \in \RR^4$ can be decomposed into $z = x - y$ with 
$x,y \in \FLC$, one can consistently extend $U_0$
to a continuous unitary representation of $\RR^4$, 
putting $U_0(z) \doteq U_0(x) U_0(y)^{-1}$. Moreover, since 
$U_0(\Lambda) U_0(y)^{-1} U_0(\Lambda)^{-1} = 
(U_0(\Lambda) U_0(y) U_0(\Lambda)^{-1})^{-1}
= U_0(\Lambda y)^{-1}$ for $y \in \Trans$, it is also clear
that the Lorentz transformations act correctly 
on the extended translations. Thus, putting 
$U_0(\lambda) \doteq U_0(z)U_0(\Lambda)$, $\lambda = (z,\Lambda) \in 
\Poin$, one obtains the desired extension of $U_0$ to $\Poin$. \\
(iii) \ Picking $A,B \in {\AlV}$,  
property (c) of vacuum states implies that the
function $z \mapsto (\pi_0(B) \Omega, U_0(z) \pi_0(A) \Omega)
= \omega_0(\alpha_y(B^*) \alpha_x(A))$, where $z=x-y$, $x,y\in\Trans$
can be continuously extended
into the tube $\RR^4 + i \Trans$ and is 
analytic in its interior. 
Moreover, the modulus
of this extension is bounded by 
$\| \pi_0(B) \Omega \| \|  \pi_0(A) \Omega \|$.  It then 
follows from standard arguments in the theory of Laplace
transforms that the spectrum of $U_0 \rest \RR^4$ is
contained in the closed forward light cone $\Spec$.
\end{proof}

\noindent \textbf{Remark.} If the eigenvalue $0$ in 
the spectrum of $U_0 \rest \RR^4$ is simple, 
the vacuum is said to be  
unique. This is the case if and only if $\omega_0$ is weakly 
clustering, \ie 
$$ \lim_{\MIO \nearrow \FL} \, 
\frac{\textstyle 1}{|{\scriptstyle \, \MIO \,}|}  \int_\MIO \! dx \, 
\omega_0(A \alpha_x(B)) =
\omega_0(A) \, \omega_0(B) \, , \quad A,B \in {\AlV} \, .$$
Thus this property of the vacuum state can also be determined in 
$\FL$. 

\medskip 
In the subsequent analysis 
we assume uniqueness of the vacuum state. 
Moreover, without loss of generality, we regard  
the corresponding vacuum representation as the defining 
representation of~$\AlV$ and therefore replace 
in the following the morphism $\pi_0$ by the identity $\iota$. 

It is worthy of note that the preceding proposition allows one to
extend the given net $\Al$ on $V$ to a local, 
Poincar\'e covariant net 
$\Al_\MI$ on Minkowski space $\MI$.
It is given by  putting for any double cone $\MIO_\MI \subset \MI$
$$ \Al_\MI(\MIO_\MI) \doteq 
U_0(x_\MI,1)^{-1} \, \Al(\MIO_\MI + x_\MI) \,
U_0(x_\MI,1) \, ,
\ x_\MI\ \in \Trans,  \ \MIO_\MI + x_\MI \subset V \, .
$$
This definition is consistent (\ie independent of the 
choice of $x_\MI \in \Trans$) because the original net 
is covariant and the unitaries 
$U_0 \rest \RR^4$ are mutually commutative. Moreover, the vacuum state 
$\omega_0$ on~$\AlV$ can be extended to a vacuum state $\omega_\MI$ 
on $\Al_\MI(\MI)$ by $\omega_\MI(A_\MI) \doteq 
(\Omega, A_\MI \Omega)$, $A_\MI \in \Al_\MI(\MI)$.
This illustrates the remark made in the introduction 
that theory together with data in~$V$ yields information 
about the past. But there is 
a caveat: the extension of $\omega_0$ to the net $\Al_M$ is not
unique in the presence of massless particles. For there are 
different extensions describing in addition outgoing massless 
particles created in 
the past of $V$ \cite[Sec.~5]{Bu1}. Nevertheless, the above      
extension is of interest here. On one hand it allows us to make use 
of results on local nets in Minkowski space established in the literature. 
On the other hand it makes contact
with the standard Minkowskian 
interpretation of quantum field theory. We shall
return to this topic in Sec.~\ref{minkowski}.  
 
In the subsequent analysis we need to consider 
the weak operator topology on the observable algebras 
in the vacuum representation. The closures of the various 
C*--algebras in this topology 
are marked by replacing the letter~$\Al$ by~$\Ri$. 
Thus $\RiV$ denotes the weak closure 
of~$\AlV$, $\RiC$ that of~$\AlC$,~\etc.
As a consequence of the preceding assumptions,
cf.~\cite[Rem.~4]{Lo}, either    
$\RiV = \cB(\cH)$ or $\RiV$ is a factor of type~III${}_1$.  
The former is the case if the spectrum of $U_0 \rest \RR^4$ 
has a mass gap \cite[Thm~2]{SaWo}, 
the latter when the spectrum has non--zero weight on the boundary 
of $\Spec \backslash \{0\}$, \ie in the presence of massless particles.
For then the algebra $\RiV$ has a non--trivial
commutant as a consequence of Huygens principle
\cite[p.~161]{Bu}. (This can be established  
in the present framework using the extended net~$\Al_\MI$.)
Our arguments apply to both cases, but as we are primarily interested
in theories with long range forces we suppose here  
that $\RiV$ is a factor of type III${}_1$. 

We supplement the preceding results assuming the
vacuum vector $\Omega$ to be cyclic and separating for the algebras 
$\AlO$ associated with any given double cone $\cO \subset V$.
This Reeh--Schlieder property of the vacuum \cite{Ha} is a 
generic feature of nets of observable 
algebras generated by quantum fields \cite{BoYn,GlJa}. 
The hypercone algebras are assumed to satisfy the  
appropriate form of Haag duality as adapted to 
the present geometrical setting: 
in analogy to sector analysis in Minkowski space theories 
\cite{DoHaRo1,BuFr} we require that there is a sufficiently large
family $\cF$ of hypercones (cf.~Definition~1 in the appendix) 
such that for each $\MIC \in \cF$
\begin{equation}  \label{duality} 
\Al(\MIC^c)^\prime \, {\textstyle \bigcap} \, \RiV = \Ri(\MIC) 
\quad \mbox{and} \quad 
\Al(\MIC)^\prime \, {\textstyle \bigcap} \, \RiV = \Ri(\MIC^c) \, ,
\end{equation}
where  a prime ${}^\prime$ on an algebra denotes its commutant 
in $\cB(\cH)$. We will refer to this condition as hypercone
duality. It expresses the idea that the hypercone  algebras are maximal
in the sense that any extension would conflict
with Einstein causality. This version of duality 
was introduced and tested for the free Maxwell field in \cite{Ca}.

\section{Charge classes and morphisms}
\label{morphisms}
\setcounter{equation}{0}

We now start to analyze charged states. 
As explained in the introduction, the concept of superselection
sector does not make sense in 
the presence of massless particles and 
has to be replaced by the notion of charge class. 

By definition \cite{Bu1}, the charge class ${\mathscr C}_0$ 
of the vacuum $\omega_0$ consists of the 
set of normal states on $\AlV$ relative to the defining  
vacuum representation~$\iota$. Thus all of these 
states extend to normal 
states on the type III${}_1$ factor~$\RiV$. As already mentioned, 
the group of inner automorphisms of a type III${}_1$ factor acts 
almost transitively on its normal states \cite{CoSt}. Hence, as the
group of unitaries in $\AlV$ is strongly *--dense in the group of 
unitaries in $\RiV$ \cite[Thm.~4.11]{Ta}, given a
state in the charge class of the vacuum, there is a sequence of 
inner automorphisms  
$\{ \gamma_n \in \mbox{In} \, \AlV \}_{n \in \NN}$ with 
$\omega_0 \compos \gamma_n$ converging in norm to that state.
Conversely, any state that can be approximated in this
way belongs to the charge class of the vacuum.

Ensembles carrying a definite global charge  
that can be precisely determined in~$V$ are described 
by primary states inducing  factorial representations of~$\AlV$.
We consider here only states carrying simple charges
(defined below) where the weak closures of the algebras in the 
corresponding representations  
are factors of type III${}_1$. The charge classes of these states 
have a characterization analogous to those in the vacuum class.

\begin{definition} 
A state $\omega$ on $\AlV$ is said to 
be elemental if the weak closure of its 
GNS-representation is a factor of type~III${}_1$.
The charge
class $\mathscr{C}$   of an elemental
state~$\omega$ is the norm closure of the set of states 
\ $\{ \omega \compos \gamma : \gamma \in  \mbox{\rm In} \, \AlV \}$
and coincides with the set of normal states in
its GNS--representation. (Note that any other 
state in $\mathscr{C}$ is elemental and belongs to the same charge class.)
\end{definition} 

\noindent {\bf Remark.} 
The notion of a charge class of elemental states 
is a physically meaningful generalization of 
the concept of superselection sector of pure states in 
massive theories. In fact, by the Kadison transitivity 
theorem \cite[Thm.~4.18(iii)]{Ta} any state in the sector of a given
pure state $\omega$ on~$\AlV$  is an element of the set 
$\{ \omega \compos \gamma : \gamma \in  \mbox{\rm In} \, \AlV \}$,
a closure in norm is not needed. 
Hence both sectors and charge classes consist of
just those states that can be reached by exploiting the 
quantum effects of physical operations
starting from a given 
pure or elemental state, respectively. 

\vspace*{2mm}
As explained in the introduction in heuristic terms,
the charge classes of interest here are obtained from 
the states in the charge class of the vacuum by 
composing with suitable sequences of inner automorphisms of $\AlV$.
We assume that for given target charge class $\mathscr{C}$  and any  
hypercone $\MIC \subset V$ there is a sequence of inner automorphisms
$\{ \gamma_n \in \mbox{\rm In} \, \AlC \}_{n \in\NN} \subset 
\mbox{\rm In} \, \AlV $ such that 
$\{ \omega_0 \compos \gamma_n \}_{n \in \NN}$ converges pointwise
on $\AlV$ to some state  $\omega \in \mathscr{C}$,
\begin{equation} \label{con1}
\lim_n \, \omega_0 \compos \gamma_n (A) = \omega (A) \, , \ A \in \AlV \, .
\end{equation}

\noindent Thus the condition of norm convergence within 
charge classes is relaxed to weak--*--convergence, 
thereby allowing the limit states $\omega$ to have a different charge. 
We supplement this assumption by a condition expressing the
heuristic idea that the process of charge creation in a 
given hypercone~$\MIC$ and operations performed in its 
distant spacelike complement are only weakly correlated in 
the vacuum state.  
The precise form of this ``independence relation'' is
\begin{equation} \label{con2} 
\lim_n \, \sup_{B} \, | \omega_0 \compos \gamma_n(B^*A) | = 
\sup_B \, \lim_n \, | \omega_0 \compos \gamma_n(B^*A) | \, ,
\end{equation} 
where $A \in \AlV$ is any fixed operator 
localized in some double cone $\MIO$, the
supremum being taken over all operators
$B \in  \Al(\MIO_d)$, satisfying the normalization condition
$\omega_0 \compos \gamma_n(B^*B) = \omega_0(B^*B) = 1$,
where $\MIO_d \subset \MIOc \cap \MICc$ 
is any other distant double cone.  
Note that localization implies $\gamma_n \rest \Al(\MIO_d) = \iota$.
The following lemma shows that
the limits of such sequences of inner 
automorphisms exist. 

\begin{lemma} Let $\{ \gamma_n \in \mbox{\rm In} \, \AlC \}_{n
    \in\NN}$ be a sequence of inner automorphisms 
satisfying conditions (\ref{con1}) and (\ref{con2}). The limit 
$\sigma_\MIC \doteq \lim_n  \gamma_n$ exists pointwise on $\AlV$
in the strong operator topology. Moreover,
\begin{enumerate}
\item[(i)]  $\sigma_\MIC : \AlV \rightarrow \RiV$ 
is a homomorphism (briefly: morphism) of algebras. 
\item[(ii)] $\sigma_\MIC \rest \AlCc = \iota$. 
\item[(iii)]$\sigma_\MIC(\Al(\MIC_1)) \subset \Ri(\MIC_1)$   
for any hypercone $\MIC_1 \supseteq \MIC$. 
\end{enumerate}
In virtue of the last two properties we say that $\sigma_\MIC$ is 
localized in the hypercone $\MIC$. 
\end{lemma}

\begin{proof} 
In the proof we make use of the Reeh--Schlieder
property of the vacuum, \ie the fact that $\Omega$ is cyclic
for $\Al(\MIO_d)$ and hence 
separating for $\Ri(\MIO_d^{\, c})$, where $\MIO_d$
is any double cone. Let $A \in \Al(\MIO)$ for any given 
$\MIO$ and let 
$B_1,B_2\in \Al(\MIO_d)$, where 
\mbox{$\MIO_d\subset\MIO^c\cap\MIC^c$}.
Then $\gamma_n(B_1^*AB_2)=B_1^*\gamma_n(A)B_2$. Acting with $\omega_0$, 
the limit over $n$ exists by the first condition. Since 
$\Omega$ is cyclic for $\Al(\MIO_d)$
and $\gamma_n(A)$ is uniformly bounded 
$\sigma_\MIC\doteq \lim_n\gamma_n$ 
exists in the pointwise weak operator topology and 
defines a linear and symmetric map 
\mbox{$\sigma_\MIC : \Al(V) \mapsto \Ri(V)$}.  For it to exist 
in the pointwise strong operator topology, we only need
to show that $\gamma_n(A) \, \Omega$ converges strongly since
$\gamma_n(A) \in \Al(\MIO \bigcup \MIC) \subset \Ri(\MIO_d^{\, c})$
and $\Omega$ is separating for $\Ri(\MIO_d^{\, c})$.
Now according to the second condition one has for 
$B \in \Al(\MIO_d)$ with $(\Omega, B^* B \Omega) = 1$ 
$$
\sup_B \, \lim_n | (\Omega, B^* \gamma_n(A) \Omega) | = 
\lim_n \, \sup_{B} \, |(\Omega, B^* \gamma_n(A) \Omega)| \, .
$$
This implies $\| \lim_n \gamma_n(A) \Omega \| = 
\lim_n \| \gamma_n(A) \Omega \|$ by 
the Reeh--Schlieder property of $\Omega$, proving the strong 
convergence of $\gamma_n(A) \, \Omega$. 
Hence $\lim_n \gamma_n \to \sigma_\MIC$ in the pointwise strong 
operator topology, proving that $\sigma_\MIC$ is also 
multiplicative, \ie a morphism. 
As $\gamma_n\rest\Al(\MIC^c)=\iota$, $n\in \mathbb N$, property (ii) 
of $\sigma_\MIC$ is evident and (iii) follows from the 
inclusion $\gamma_n(\Al(\MIC_1))\subset\Ri(\MIC_1)$ for $\MIC_1\supset\MIC$ 
and~\mbox{$n \in \NN$}.
\end{proof}

\vspace*{-1mm} 
An immediate consequence of this lemma 
and the Reeh--Schlieder property of the vacuum is that the 
GNS--representation of the (charged) limit 
state $\omega_0 \compos \sigma_\MIC$ 
is 
\mbox{$\sigma_\MIC : \AlV \rightarrow \RiV \subset \cB(\cH)$.}
Thus the charged states
of interest here are vector states in the 
defining Hilbert space of the theory $\cH$ for
representations of the observable algebra
derived from hypercone localized morphisms. 
Furthermore, if the state $\omega_0 \compos \sigma_\MIC$ 
is elemental, composing it with the elements of $\mbox{In} \, \AlV$ 
generates a norm dense subset of states in its charge class.
Since $\sigma_\MIC \compos \mbox{In} \, \AlV 
\subset \mbox{In} \, \RiV \compos \sigma_\MIC$ by the 
first part of the preceding lemma,  
the corresponding GNS--representations of $\AlV$ act on 
$\cH$ too and are given by (not necessarily localized)
morphisms with range in~$\RiV$. Moreover, all such
representations are  
equivalent to the initial representation~$\sigma_\MIC$ 
via unitary intertwiners in $\RiV$. Thus almost all states 
in the charge class of  $\omega_0 \compos \sigma_\MIC$ 
induce representations of $\AlV$, mutually 
equivalent to each other in this strong sense. 

The starting point of sector analysis in massive theories 
is to consider morphisms of the (irreducible) algebra of 
observables whose ranges and whose intertwiners are 
contained in its weak closure \cite{BuFr}.
These observations motivate our assumption that given  
a relevant charge class there are corresponding 
morphisms localized in given  
hypercones and mutually equivalent via 
unitary intertwiners in $\RiV$.  
In the subsequent analysis we restrict attention to the 
physically significant case of ``simple charges'',
characterized as follows. 

\vspace*{2mm} \noindent
{\bf Criterion:} Let $\mathscr{C}$ be a charge class of 
elemental states on $\AlV$. The states and their charge class are said to be 
{simple} if given a hypercone $\MIC \in \cF$ \footnote{$\cF$
is the family of hypercones appearing in condition 
(\ref{duality}) of hypercone duality.}
there is a morphism $\sigma_\MIC : \AlV \rightarrow \RiV$ \
with \
$\omega_0 \compos \sigma_\MIC \in \mathscr{C}$ and 

\vspace*{-3.5mm}
\begin{enumerate}
\item[(a)] $\sigma_{\, \MIC} \rest \AlCc = \iota$, 
\item[(b)] $\sigma_{\, \MIC} (\Al(\MIC_1))^- = \Ri(\MIC_1)$ \ for 
any $\MIC_1 \supseteq \MIC$, 
\item[(c)] $\mbox{In} \, \RiV \compos \sigma_{\, \MIC_1} \,
\mbox{\large $\bigcap$} \, 
\mbox{In} \, \RiV \compos \sigma_{\, \MIC_2} \neq \emptyset$
\ for any pair of hypercones $\MIC_1, \MIC_2 \in \cF$.  
\end{enumerate} 

In the sequel, it will be convenient to distinguish the morphisms 
$\sigma_{\, \MIC}$ not only by 
their action on $\AlV$ but also by a hypercone $\MIC$ of 
localization in the sense of conditions (a) and (b). 
Thus two morphisms \mbox{$\sigma_{\, \MIC_1}$, $\sigma_{\, \MIC_2}$} 
with different hypercones $\MIC_1, \MIC_2$ 
of localization may act  on $\AlV$ in exactly the same way.   
Condition (b), where the bar ${}^-$ denotes the 
weak closure of the respective algebras, encodes the 
decisive information that the charge class is
simple. In particular, 
\mbox{$\sigma_{\, \MIC}(\AlV)^- = \RiV$} follows. This generalizes how simple 
sectors are characterized in Minkowski space theories, where an  
analogous equality of algebras is implied, 
cf.~\cite[Lem.~2.2]{DoHaRo1}. Condition (c) says that
for any pair of morphisms $\sigma_{\, \MIC_1}, \sigma_{\, \MIC_2}$
associated with $\mathscr C$ there are unitary 
intertwiners in $\RiV$, unique up to a phase by Condition~(b). 

The collection of hypercone localized morphisms corresponding to 
simple charge classes in the theory is denoted by $\Sigma(V)$ and  
$\Sigma(\MIC) \subset \Sigma(V)$ is the subset of morphisms 
localized in a given hypercone; throughout the 
subsequent discussion we shall implicitly assume that  $\MIC \in \cF$. 
We shall see that there is a composition law for the elements
of $\Sigma(V)$, reflecting the composition of charges, 
every element of $\Sigma(V)$ has an inverse in  
the conjugate charge class, and
the morphisms in every charge class have definite 
(Bose or Fermi) statistics. 

After fixing the framework we begin
to analyze the simple charge classes  
described by the morphisms in~$\Sigma(V)$. 
The proofs are similar to those adopted in 
\cite{BuFr} for massive theories. Yet geometric complications 
mean that some arguments have to be modified.  
Since we shall have to deal with various extensions
of the morphisms we introduce the following convenient 
notation.  

\vspace*{2mm} 
\noindent {\bf Notation:}
Let $\tAl(V)$ be any given C*--algebra, 
$\AlV \subseteq \tAl(V) \subseteq \RiV$, and  
let $\tild{\sigma}_i : \tAl(V) \rightarrow
\RiV$ be morphisms, $i = 1,2$. We write 
$\tild{\sigma}_1 \simeq \tild{\sigma}_2$ if there is 
a unitary intertwiner $W \in
\RiV$ between $\tild{\sigma}_1$ and $\tild{\sigma}_2$, \ie 
$$  \mbox{Ad} \, W \compos  \tild{\sigma}_{1} (A)  = 
W  \tild{\sigma}_{1} (A)   W^{-1} =
\tild{\sigma}_{2} (A) \, , \quad A \in \tAl(V) \, . $$
The set of all such unitary intertwiners in $\RiV$ is denoted by 
$(\tild{\sigma}_{1}, \tild{\sigma}_{2})$.

\vspace*{2mm}

\vspace*{1mm} 
The first step is to show 
that any morphism $\sigma_{\, \MIC} \in \Sigma(V)$ can be extended (as a morphism)
from $\AlV$ to certain larger domains. These domains are fixed by 
specifying some funnel of hypercones $\{ \tcC_n \}_{n \in \NN}$
where 
\mbox{${\tcC}_1 \supset \dots 
{\tcC}_n \supset \dots$ and ${\tcC}_n{}^{\! c}
\nearrow V$}, cf.~\ref{A.1}, and are defined as C*--inductive limits   
\mbox{$\tAl(V) \doteq \varinjlim \ \Ri({\tcC}_n{}^{\! c})  
\subset \Ri(V)$} of the net 
\mbox{$\Ri({\tcC}_1{}^{\! c}) \subset 
\dots \, \Ri({\tcC}_n{}^{\! c}) \subset \dots $} 
of von Neumann algebras associated with the spacelike complements 
of the hypercones in the given funnel. Since 
$\MIO \subset {\tcC}_n{}^{\! c}$,  cf.~\ref{A.1},  
for any given double cone $\MIO$ and sufficiently large $n$, 
$\tAl(V) \supset \AlV$. 
The following lemma establishes the existence of 
such extensions and proves some of 
their properties.

\begin{lemma} \label{extensions}
Let $\sigma_{\, \MIC} \in \Sigma(V)$ 
and let $\{{\tcC}_n \}_{n \in \NN}$ be any funnel of hypercones.
\begin{enumerate}
\item[(i)]  $\sigma_{\, \MIC}$ extends to a morphism  $\tsigma_{\, \MIC} : 
\tAl(V) \rightarrow \RiV$
normal on each von Neumann algebra $\Ri({\tcC}_n^{\ c})$, $n \in \NN$.
If $\MIC \subset {\tcC}_m^{\ c}$, for some $m \in \NN$, the domain 
$ \tAl(V) $ is stable under  $\tsigma_{\, \MIC}$.   
\item[(ii)] Let $\{ {\ttcC}_n \}_{n \in \NN}$ be another 
funnel of hypercones and $\ttsigma_{\, \MIC} : \, 
\ttAl(V) \rightarrow \RiV$
be the corresponding extension of $\sigma_{\, \MIC}$. If 
$\MIC_0$ is a hypercone with 
$\MIC_0 \subset  {\tcC}_m{}^{\! c} \bigcap {\ttcC}_m{}^{\! c}$
for some $m \in \NN$, 
then $\ttsigma_{\, \MIC} \rest \Ri(\MIC_0) 
= \tsigma_{\, \MIC} \rest \Ri(\MIC_0) $.
\item[(iii)] Let $\sigma_{\, \MIC_1} \simeq \sigma_{\, \MIC_2}$ and 
let  $\tsigma_{\, \MIC_1}$, $\tsigma_{\, \MIC_2}$ be their respective extensions
to a common domain $\widetilde\Al(V)$. 
Then $\tsigma_{\, \MIC_1} \simeq \tsigma_{\, \MIC_2}$ 
and $(\tsigma_{\, \MIC_1}, \tsigma_{\, \MIC_2}) = 
(\sigma_{\, \MIC_1}, \sigma_{\, \MIC_2})$, \ie the associated
sets of unitary intertwiners in $\RiV$ coincide. 
\end{enumerate}
\end{lemma}
\begin{proof} 
As morphisms can be localized arbitrarily within any given charge class 
there are morphisms $\sigma_{\, \tcC_n} \simeq
  \sigma_{\, \MIC}^{}$ and intertwiners $W_n \in (\sigma_{\, \tcC_n},
  \sigma_{\, \MIC}^{})$ with 
$$ \sigma_{\, \MIC}^{} \rest \Al({\tcC}_n{}^{\! c})
= \mbox{Ad} \, W_n \compos \sigma_{\, \tcC_n} \rest 
\Al({\tcC}_n{}^{\! c})
=  \mbox{Ad} \, W_n \rest \Al({\tcC}_n{}^{\! c}) \, .
$$
Hence $\sigma_{\, \MIC}$ extends by weak continuity to 
$\Ri({\tcC}_n^{\, c})$, $n \in \NN$, and thus to the C*--inductive
limit $\tAl(V)$ 
of these algebras. The algebraic and normality properties of the
resulting extension $\tsigma_{\, \MIC}$ are apparent from this 
construction. If $\MIC \subset {\tcC}_m^{\, c}$ for some $m \in \NN$, 
$\sigma_{\, \MIC} \rest \Al({\tcC}_n) = \iota$,  
for $n \geq m$, then 
$\sigma_{\, \MIC} (\Al({\tcC}_n^{\, c})) \subset \Al({\tcC}_n)^\prime =
\Ri({\tcC}_n^{\, c})$ by hypercone duality. But 
$\sigma_{\, \MIC} (\Al({\tcC}_n^{\, c}))^- = \tsigma_{\, \MIC}(\Al({\tcC}_n^{\,
  c})^-) = \tsigma_{\, \MIC}(\Ri({\tcC}_n^{\, c}))$, 
completing the proof of the first part of the statement. 
The second part  follows from
$$ \tsigma_{\, \MIC} \rest \Al(\MIC_0)
= {\sigma}_{\, \MIC} \rest \Al(\MIC_0)
=  \ttsigma_{\, \MIC} \rest \Al(\MIC_0)
$$
and the normality of $\tsigma_{\, \MIC}$ and $\ttsigma_{\, \MIC}$ on 
$\Ri(\MIC_0) \subset  \Ri(\tcC_m^{\, c}) \bigcap \Ri(\ttcC_m{}^{\! \!
  \! c})$. 
The third part follows similarly   
from the normality properties of the extensions.
\end{proof}

Lemma \ref{extensions}(i) implies that 
morphisms $\sigma_{\, \MIC_1} , \dots , \sigma_{\, \MIC_n} \in \Sigma(V)$ 
can be extended to a common stable domain 
whenever there is some auxiliary hypercone 
$\MIC_0 \subset \bigcap_{k=1}^n \MIC_k^{\, c}$  
containing a funnel, cf.~\ref{A.1}. The (associative) product of
the extended morphisms is then well defined. However, 
for certain pairs of hypercones there is no such  
auxiliary hypercone $\MIC_0$, cf.~the appendix. 
The following result ensures that this geometrical obstacle does 
not cause major problems.

\begin{proposition} \label{prop1}
Let $\sigma_{\, \MIC_a}, \tau_{\, \MIC_b} \in \Sigma(V)$.  
\begin{enumerate}
\item[(i)] If $\tsigma_{\, \MIC_a}$ is an extension of
$\sigma_{\, \MIC_a}$ containing $\Ri(\MIC_b)$ in its domain, the composed map 
$\tsigma_{\, \MIC_a} \compos \tau_{\, \MIC_b} : \AlV \rightarrow \RiV$ 
is a well defined morphism independent of the 
chosen extension of $\sigma_{\, \MIC_a}$.
\item[(ii)] Let $\sigma_{\, \MIC_A} \simeq \sigma_{\, \MIC_a}$ and
$\tau_{\, \MIC_B} \simeq \tau_{\, \MIC_b}$. Then 
$\ttsigma_{\, \MIC_A} \compos {\tau}_{\, \MIC_B} 
\simeq  \tsigma_{\, \MIC_a} \compos \tau_{\,
  \MIC_b}$, where $\ttsigma_{\, \MIC_A}$ is any 
extension of $\sigma_{\, \MIC_A}$ containing $\Ri(\MIC_B)$ in its 
domain.
\item[(iii)] Given any hypercone $\MIC$ there is a unitary 
intertwiner in $\RiV$ taking a morphism 
$\varrho_{\, \MIC} \in \Sigma(\MIC)$
to $\tsigma_{\, \MIC_a} \compos \tau_{\, \MIC_b}$. 
Moreover, if $\MIC_a, \MIC_b \subset \MIC$, 
$\tsigma_{\, \MIC_a} \compos \tau_{\, \MIC_b} \rest \AlV
\in \Sigma(\MIC)$. 
\end{enumerate}
\end{proposition}
\begin{proof} Given any double cone
$\cO \subset V$ there is a
hypercone $\MIC_B \subset \cO^c \bigcap \MIC_b$, cf.~\ref{A.3}, and a morphism 
$\tau_{\, \MIC_B} \simeq \tau_{\, \MIC_b}$ with 
$\tau_{\, \MIC_B} \rest \Al(\cO) = \iota$  
so that if $W \in (\tau_{\, \MIC_B},  \tau_{\, \MIC_b})$ 
$$ \tau_{\, \MIC_b}(A) = \mbox{Ad} \, W \compos \tau_{\, \MIC_B} (A)
= \mbox{Ad} \, W (A) \, , \quad A \in \Al(\cO) \, .
$$
Similarly, $\tau_{\, \MIC_b} \rest \Al(\MIC_b^{\, c}) = \iota 
= \tau_{\, \MIC_B} \rest \Al(\MIC_b^{\, c}) $ implying   
$(\tau_{\, \MIC_B}, \tau_{\, \MIC_b}) \subset \Ri(\MIC_b)$ 
by hypercone duality. 
Hence $W$ is in the domain of $\tsigma_{\, \MIC_a}$ and  
$$ 
\tsigma_{\, \MIC_a}(W) \sigma_{\, \MIC_a}(A) \tsigma_{\, \MIC_a}(W^{-1}) 
= \tsigma_{\, \MIC_a} (\mbox{Ad} \, W (A)) 
= \tsigma_{\, \MIC_a} (\tau_{\, \MIC_b}(A)) \, , \quad A \in
\Al(\cO) \, . 
$$
As $\cO$ was arbitrary,  
$\tsigma_{\, \MIC_a} \compos \tau_{\, \MIC_b} : \AlV \rightarrow \RiV$
is a well defined morphism. Moreover, since  
$\tsigma_{\, \MIC_a}(W)$ does not depend on the chosen 
extension of $\sigma_{\, \MIC_a}$ by 
Lemma~\ref{extensions}(ii) neither does
$\tsigma_{\, \MIC_a} \compos \tau_{\, \MIC_b}$. 

To prove the second part of the proposition we first keep 
$\tau_{\, \MIC_b}$ fixed. Let 
$\tsigma_{\, \MIC_a}, \tsigma_{\, \MIC_A}$
be extensions of $\sigma_{\, \MIC_a}, \sigma_{\, \MIC_A}$, respectively,
both having $\Ri(\MIC_b)$ in their domain and let 
\mbox{$W \in (\sigma_{\, \MIC_A}, \sigma_{\, \MIC_a})$}. 
The normality
properties of extensions, established in Lem\-ma~\ref{extensions}(i), 
imply
$\mbox{Ad} \, W \compos \tsigma_{\, \MIC_A} \rest
\Ri(\MIC) = \tsigma_{\, \MIC_a} \rest
\Ri(\MIC)$, whenever $\Ri(\MIC)$ is in the domain 
of both extensions. Hence 
$\tsigma_{\, \MIC_A} \compos {\tau}_{\, \MIC_b} 
\simeq  \tsigma_{\, \MIC_a} \compos \tau_{\,
  \MIC_b}$. Next, keeping $\sigma_{\, \MIC_A}$ fixed we 
vary $\tau_{\, \MIC_b}$ and pick any $\tau_{\, \MIC_B} \simeq \tau_{\, \MIC_b}$
such that there is a larger 
hypercone $\MIC_0 \supset \MIC_b \bigcup \MIC_B$. Then 
$\tau_{\, \MIC_B} \rest \Al(\MIC_0^{\, c}) = \iota = 
\tau_{\, \MIC_b} \rest \Al(\MIC_0^{\, c})$, hence 
$(\tau_{\, \MIC_B}, \tau_{\, \MIC_b}) \subset \Ri(\MIC_0)$ 
by hypercone duality. Choosing an extension 
$\ttsigma_{\, \MIC_A}$ of $\sigma_{\, \MIC_A}$
with $\Ri(\MIC_0)$ in its domain,   
we obtain 
$\tsigma_{\, \MIC_A} \compos \tau_{\, \MIC_b} 
= \ttsigma_{\, \MIC_A} \compos \tau_{\, \MIC_b}
= \ttsigma_{\, \MIC_A} \compos \mbox{Ad} \, W \compos \tau_{\, \MIC_B}
= \mbox{Ad} \, \ttsigma_{\, \MIC_A}(W) \compos
\ttsigma_{\, \MIC_A} \compos  \tau_{\, \MIC_B} $
for $W \in (\tau_{\, \MIC_B}, \tau_{\, \MIC_b})$.  
Hence \ $\ttsigma_{\, \MIC_A} \compos {\tau}_{\, \MIC_B} 
\simeq  \tsigma_{\, \MIC_a} \compos \tau_{\,
  \MIC_b}$ for the restricted set of regions 
$\MIC_b, \MIC_B$. The result for pairs of
hypercones $\MIC_b, \MIC_B \in \cF$ in general position 
then follows by a standard interpolation argument as 
the family $\cF$ is pathwise connected, cf.~\ref{A.5}.

To prove 
the third part of the proposition, we pick morphisms
$\sigma_{\, \MIC_A} \simeq \sigma_{\, \MIC_a}$,  
\mbox{$\tau_{\, \MIC_B} \simeq \tau_{\, \MIC_b}$} 
with $\MIC_A, \MIC_B \subset \MIC$ and define
$\varrho_{\, \MIC} \doteq \tsigma_{\, \MIC_A} \compos \tau_{\,
  \MIC_B}$. The localization properties of  
$\sigma_{\, \MIC_A}$, $\tau_{\, \MIC_B}$ imply 
$\varrho_{\, \MIC} \rest \Al(\MIC^{\, c}) = \iota$,
so $\varrho_{\, \MIC}$ satisfies point (a) of the 
criterion. To establish~(b)
we use Lemma \ref{extensions}(i)   
choosing, for given $\MIC_1 \supset \MIC$, 
an extension $\ttsigma_{\, \MIC_A}$ of $\sigma_{\, \MIC_A}$
normal on $\Ri(\MIC_1)$. But, as shown in the first step,
$\tsigma_{\, \MIC_A} \compos \tau_{\, \MIC_B} = \ttsigma_{\, \MIC_A} 
\compos \tau_{\,  \MIC_B}$, and 
$\ttsigma_{\, \MIC_A} (\Al(\MIC_1)^-) =
\sigma_{\, \MIC_A}(\Al(\MIC_1))^- = \Ri(\MIC_1) = 
\tau_{\,  \MIC_B}(\Al(\MIC_1))^-$. Hence 
$$\varrho_{\, \MIC}(\Al(\MIC_1))^-
= \ttsigma_{\, \MIC_A} \compos \tau_{\, \MIC_B} (\Al(\MIC_1))^-
= \ttsigma_{\, \MIC_A}  (\tau_{\, \MIC_B}(\Al(\MIC_1))^-)
=  \Ri(\MIC_1) \, .
$$ 
Thus $\varrho_{\, \MIC}$ also satisfies (b).
But the hypercone $\MIC$ was arbitrary so 
the results established in (ii) imply that the morphism 
$\varrho_{\, \MIC}$ satisfies (c), too. 
Hence $\varrho_{\, \MIC} \in
\Sigma(\MIC)$ and, in particular, 
$\tsigma_{\, \MIC_a} \compos \tau_{\, \MIC_b}  \in \Sigma(\MIC)$
whenever $\MIC_a, \MIC_b \subset \MIC$. 
\end{proof}

The preceding proposition shows how pairs of morphisms
$\sigma_{\, \MIC_a},\tau_{\, \MIC_b} \in \Sigma(V)$ can be composed
inducing a composition of the corresponding
simple charge classes $\mathscr{C}_a, \mathscr{C}_b$
by applying the composed morphisms to the
vacuum state. Since the composition of morphisms 
$\tsigma_{\, \MIC_a} \compos \tau_{\, \MIC_b}$
does not depend on the chosen extension $\tsigma_{\, \MIC_a}$
of $\sigma_{\, \MIC_a}$
(up to limitations depending on the given localization 
hypercone $\MIC_b$ of $\tau_{\, \MIC_b}$) we will omit the 
tilde in the following  and 
simply write $\sigma_{\, \MIC_a} \product \tau_{\, \MIC_b}$
for the composed morphisms. Yet it must be remembered 
that care must be taken with domains in this product and that 
$\sigma_{\, \MIC_a} \product \tau_{\, \MIC_b} \not\in \Sigma(V)$ unless 
$\MIC_a, \MIC_b \subset \MIC$ 
for some hypercone $\MIC$.  

\vspace*{2mm} 
We show next that each simple charge class $\mathscr{C}$  has a 
simple ``conjugate'' charge class $\quer{\mathscr{C}}$. The following
proposition says that the corresponding charges  compensate 
(neutralize) one another by composition. 
\begin{proposition} \label{prop2}
Given a simple charge class there is a 
simple conjugate  charge class in the following sense: 
given a morphism $\sigma_{\, \MIC} \in \Sigma(V)$ there is a morphism
$\qsigma_{\, \MIC} \in \Sigma(V)$ with 
$\qsigma_{\, \MIC} \product \sigma_{\, \MIC} 
= \sigma_{\, \MIC} \product \qsigma_{\, \MIC} = \iota$.         
\end{proposition}
\begin{proof}
Given $\sigma_{\, \MIC}$ we 
pick an increasing sequence of 
hypercones $\MIC_n \supset \MIC$, $n \in \NN$, 
with $\MIC_n \nearrow V$; the corresponding opposite 
cones ${\tcC}_n \subset \MIC_n^{\, c}$, $n \in \NN$,  
then form a funnel of hypercones, cf.~\ref{A.2}. 
Next, we choose morphisms 
$\sigma_{\, \tcC_n} \simeq \sigma_{\, \MIC}$ and  
unitary intertwiners $W_n \in (\sigma_{\, \tcC_n}, \sigma_{\,
  \MIC})$, $n \in \NN$. Thus for any given $k \in \NN$
$$ \mbox{Ad} \, W_n^{-1} \compos \sigma_{\, \MIC} \rest 
\Al(\MIC_k) =  \sigma_{\, \tcC_n} \rest 
\Al(\MIC_k) = \iota \, , \quad n \geq k \, .
$$
According to (b) of the criterion,  
$\, \sigma_{\, \MIC}(\Al(\MIC_k))^- = \Ri(\MIC_k)$, $k \in \NN$, hence 
$$ \mbox{Ad} \, W_m^{-1}  \rest \Ri(\MIC_k) =  
\mbox{Ad} \, W_n^{-1} \rest \Ri(\MIC_k) \, , \quad m,n \geq k \, .
$$
So the pointwise norm limit 
$$
\qsigma_{\, \MIC} \doteq \lim_n \, \mbox{Ad} \, W_n^{-1} 
$$ 
exists on the C*--inductive limit 
$\qAl(V) \doteq \varinjlim \, \Ri(\MIC_k) \supset \AlV$
of $\Ri(\MIC_k), \,  k \in \NN$ 
and defines a morphism 
$\qsigma_{\, \MIC} : \qAl(V) \rightarrow \RiV$ 
which is 
normal on each algebra $\ \Ri(\MIC_k)$, \mbox{$k \in \NN$}. Moreover, it 
is a left and right inverse of 
the (suitably extended) morphism $\sigma_{\, \MIC}$ as we will show
next. Let $\tsigma_{\, \MIC}$ be an extension of $\sigma_{\, \MIC}$
based on the funnel $\{ {\tcC}_n \}_{n \in \NN}$. Then,
for any $n \geq k$, 
$\, \tsigma_{\, \MIC} \rest {\Ri}(\MIC_k) = \mbox{Ad} \, W_n \rest
{\Ri}(\MIC_k) $ and, by construction, 
\mbox{$\qsigma_{\, \MIC} \rest {\Ri}(\MIC_k) = \mbox{Ad} \, W_n^{-1} \rest
{\Ri}(\MIC_k)$}. As $\MIC \subset \MIC_k$,  
\mbox{${\Ri}(\MIC_k) = \sigma_{\, \MIC}({\Al}(\MIC_k))^- =
\mbox{Ad} W_n \, ({\Ri}(\MIC_k)) $}, implying  
$\qsigma_{\, \MIC}({\Al}(\MIC_k))^- = {\Ri}(\MIC_k)$, whence 
\begin{equation} \label{compensation}
\begin{split}
& \qsigma_{\, \MIC}(\sigma_{\, \MIC}(A)) 
= \mbox{Ad} \, W_k^{-1} ( \mbox{Ad} \, W_k (A)) = A \, \\
& \tsigma_{\, \MIC}(\qsigma_{\, \MIC}(A)) 
= \mbox{Ad} \, W_k ( \mbox{Ad} \, W_k^{-1} (A)) = A  \, ,
\end{split}
\end{equation}
for $A \in \Al(\MIC_k)$, $k \in \NN$, 
and these equalities extend 
by continuity to $\AlV$.

To proceed we need to show that the restriction 
$\qsigma_{\, \MIC} \rest \AlV$
does not depend on the initial choice of a sequence of hypercones.
Another admissible sequence yields another morphism 
$\qqsigma_{\, \MIC} : \qqAl(V) \rightarrow \RiV$ 
with the properties established above. In particular, 
it is a left inverse of $\sigma_{\, \MIC}$, hence 
$\qqsigma_{\, \MIC} (\sigma_{\, \MIC} (A)) 
= A = \qsigma_{\, {\MIC}}  (\sigma_{\, \MIC} (A)) $ for 
$A \in \AlV$.
Both $\qqsigma_{\, \MIC}$ and $\qsigma_{\, \MIC}$ 
have $\Ri(\MIC)$ in their domains and are normal on this 
algebra. It therefore follows from this equality and
$\sigma_{\, \MIC}(\Al(\MIC))^- = \Ri(\MIC)$ that \ 
\mbox{$\qqsigma_{\, \MIC} \rest \Ri(\MIC) 
= \qsigma_{\, {\MIC}} \rest \Ri(\MIC)$}. Moreover,   
given any double cone $\cO \subset V$
there is a unitary $W \in \Ri(\MIC)$ with
$\sigma_{\, \MIC} \rest \Al(\cO) = \mbox{Ad} W \rest \Al(\cO)$. 
Thus using the above equality once more   
$$
\qqsigma_{\, \MIC}(W) \qqsigma_{\, \MIC}(A) 
\qqsigma_{\, \MIC}(W^{-1}) = A = 
\qsigma_{\, {\MIC}}(W) \qsigma_{\, {\MIC}}(A) 
\qsigma_{\, {\MIC}}(W^{-1}) \, , \quad 
A \in  \Al(\cO) \, .
$$
But $\qqsigma_{\, \MIC} (W) =  \qsigma_{\, {\MIC}} (W)$,
hence $\qqsigma_{\, \MIC}(A) =  \qsigma_{\, {\MIC}}(A)$, 
$A \in \Al(\cO)$. Since $\cO$ was arbitrary this shows
$\qqsigma_{\, \MIC} \rest \AlV
= \qsigma_{\, {\MIC}} \rest \AlV$.

We are now in a position to prove that the morphisms 
$\qsigma_{\, {\MIC}} : \AlV \rightarrow \RiV$  satisfy
the criterion. Choosing morphisms and intertwiners as in 
the first step of the proof we have 
$$ \mbox{Ad} \, W_n^{-1}  \rest \Al(\MIC^c) =  
\mbox{Ad} \, W_n^{-1} \compos \sigma_{\, \MIC} \rest \Al(\MIC^c) = 
\sigma_{\, \tcC_n} \rest \Al(\MIC^c) \, .
$$
Since $\sigma_{\, \tcC_n} \rightarrow \iota$ 
pointwise in norm on $\AlV$ as
$n \rightarrow \infty$ it follows that 
$\qsigma_{\, \MIC} \rest \Al(\MIC^c) = \iota$
proving (a). 
Next, given a hypercone $\MIC_0 \supset \MIC$, we 
choose an increasing sequence of hypercones 
$\MIC_n$, $n \in \NN$, with $\MIC_1 \doteq \MIC_0$, 
yielding an extension $\qqsigma_{\, \MIC}$
of $\qsigma_{\, \MIC} \rest \AlV$ normal on~$\Ri(\MIC_0)$. 
Thus, bearing in mind that
$\sigma_{\, \MIC}(\Al(\MIC_0))^- = \Ri(\MIC_0) = \Al(\MIC_0)^-$, we get
\begin{equation*}
\begin{split}
\qsigma_{\, \MIC}(\Al(\MIC_0))^- & = \qqsigma_{\, \MIC}(\Al(\MIC_0))^- 
=  \qqsigma_{\, \MIC}(\Al(\MIC_0)^-) \\
& =  \qqsigma_{\, \MIC}({\sigma}_{\, \MIC}(\Al(\MIC_0))^-) 
=  \qqsigma_{\, \MIC}({\sigma}_{\, \MIC}(\Al(\MIC_0)))^-
=  \Al(\MIC_0)^- \, ,
\end{split} 
\end{equation*}
where in the last equality we used (\ref{compensation}). This
proves (b). Finally, 
if $\MIC_a$, $\MIC_b$ and \mbox{$\MIC_0 \supset \MIC_a, \MIC_b$}  
 are hypercones and 
$\sigma_{\, \MIC_a} \simeq \sigma_{\, \MIC_b}$ 
with conjugates 
$\qsigma_{\, \MIC_a}$ and $\qsigma_{\, \MIC_b}$,
respectively, we choose an 
increasing sequence of hypercones 
$\MIC_n$, $n \in \NN$, with $\MIC_1 \doteq \MIC_0$, 
as above.
As has been shown, the corresponding extensions 
$\qqsigma_{\, \MIC_a}$, 
$\qqsigma_{\, \MIC_b}$ of the conjugate morphisms
to the domain 
$\qAl(V) \doteq \varinjlim \, \Ri(\MIC_k)  
\supset \Al$ are normal on each member of the net 
$\Ri(\MIC_k)$, $k \in \NN$. Now 
the unitary intertwiners $W_0 \in ({\sigma}_{\, \MIC_b}, {\sigma}_{\, \MIC_a})$
are elements of $\Ri(\MIC_0)$ by hypercone duality. So like the elements of  
${\sigma}_{\, \MIC_b}(\AlV)$,  ${\sigma}_{\, \MIC_a}(\AlV)$,
they are in the domain of $\qqsigma_{\, \MIC_a}$  and we can compute 
\begin{gather*}
\qqsigma_{\, \MIC_a}(W_0) \, 
\qqsigma_{\, \MIC_a}({\sigma}_{\, \MIC_b} (A)) \, 
\qqsigma_{\, \MIC_a}(W_0^{-1}) =
\qqsigma_{\, \MIC_a}(W_0 \, {\sigma}_{\, \MIC_b} (A) W_0^{-1})  \\
= \qqsigma_{\, \MIC_a}({\sigma}_{\, \MIC_a} (A)) = A
= \qqsigma_{\, \MIC_b} ({\sigma}_{\, \MIC_b} (A)) \, , 
\quad A \in \Al(V) \, .
\end{gather*}
But  
${\sigma}_{\, \MIC_b}(\Al(\MIC_k))^- = \Ri(\MIC_k)$, $k \in \NN$, 
so normality implies 
$\mbox{Ad} \, \qqsigma_{\, \MIC_a}(W_0) \compos 
\qqsigma_{\, \MIC_a} = \qqsigma_{\, \MIC_b}$
on $\qAl(V)$. Restricting this equality to
$\AlV$ we conclude that the unitary operator 
$\qqsigma_{\, \MIC_a}(W_0) \in \RiV$  intertwines 
$\qsigma_{\, \MIC_a}$ and $\qsigma_{\, \MIC_b}$ for
the restricted pairs of hypercones. By a standard 
interpolation argument, cf.~\ref{A.5},  
this equivalence extends to arbitrary 
pairs of morphisms $\qsigma_{\, \MIC_a}, \qsigma_{\, \MIC_b}$,
so they satisfy (c), too. 
Thus $\qsigma_{\, \MIC} \in \Sigma(V)$ for any 
choice of hypercone $\MIC$.
Relation (\ref{compensation}) implies  
$\qsigma_{\, \MIC} \product \sigma_{\, \MIC} =
\sigma_{\, \MIC} \product  \qsigma_{\, \MIC} = \iota$,  
completing the proof of the proposition. 
\end{proof}

\vspace*{2mm}
We now analyze the statistics of charge classes as 
encoded in the structure of the intertwiners of composed morphisms. 
Since we are just dealing with simple charges we do not 
need the full arsenal of categorical methods developed in \cite{DoHaRo1,BuFr}
and can rely on strategies established in \cite{DoHaRo3}. 
But, again, geometric problems mean that some  arguments 
have to be modified. 

\begin{lemma} \label{exchange} 
Let $\sigma_{\, \MIC_a}$, $\tau_{\, \MIC_b}$ be morphisms.
\begin{enumerate}
\item[(i)] $\sigma_{\, \MIC_a} \product \tau_{\, \MIC_b} =
\tau_{\, \MIC_b} \product \sigma_{\, \MIC_a}$ if $\MIC_a$ and $\MIC_b$
are spacelike separated.
\item[(ii)] $\sigma_{\, \MIC_a} \product \tau_{\, \MIC_b} \simeq 
\tau_{\, \MIC_b} \product \sigma_{\, \MIC_a}$ if $\MIC_a$ and $\MIC_b$ are
in arbitrary position.
\end{enumerate}
\end{lemma}
\begin{proof}
Given a double cone $\cO \subset V$, we choose (cf.~\ref{A.3}) 
hypercones 
$\MIC_A \subset \cO^c \bigcap \MIC_a$, 
\mbox{$\MIC_B \subset \cO^c \bigcap \MIC_b$}, morphisms
$\sigma_{\, \MIC_A} \simeq \sigma_{\, \MIC_a}$, 
$\tau_{\, \MIC_B}  \simeq   \tau_{\, \MIC_b}$ and intertwiners 
\mbox{$W_a \in ( \sigma_{\, \MIC_A}, \sigma_{\, \MIC_a})$}, 
\mbox{$W_b \in (\tau_{\, \MIC_B}, \tau_{\, \MIC_b})$}.
If  $\MIC_a$ and $\MIC_b$ are spacelike separated,
the localization of the
morphisms $\sigma_{\, \MIC_A}$, $\tau_{\, \MIC_B}$ 
imply that the intertwiners
$W_a \in \Ri(\MIC_a)$, $W_b \in \Ri(\MIC_b)$ commute, so
\begin{equation*}
\begin{split}
&  \sigma_{\, \MIC_a} \product \tau_{\, \MIC_b} (A) =
\tsigma_{\, \MIC_a}(\tau_{\, \MIC_b}(A)) =  
\tsigma_{\, \MIC_a}(\mbox{Ad} \, W_b (A)) =
\mbox{Ad} \, W_a ( \mbox{Ad} \, W_b (A))  \\
& = \mbox{Ad} \, W_b ( \mbox{Ad} \, W_a (A))  
= \tilde{\tau}_{\, \MIC_b}(\mbox{Ad} \, W_a (A)) =
\tilde{\tau}_{\, \MIC_b}(\sigma_{\, \MIC_a}(A)) = 
\tau_{\, \MIC_b} \product \sigma_{\, \MIC_a} (A) 
\end{split}
\end{equation*}
for any $A \in \Al(\cO)$.

Since $\cO$ was arbitrary, (i) follows.
We complete the proof by choosing  spacelike separated hypercones 
$\MIC_A$, $\MIC_B$ and morphisms 
$\sigma_{\, \MIC_A} \simeq \sigma_{\, \MIC_a}$, 
$\tau_{\, \MIC_B} \simeq \tau_{\, \MIC_b}$. (ii) 
then
follows from Proposition \ref{prop1}(iii) and the preceding result. 
\end{proof} 

We now consider equivalent morphisms 
$\sigma_{\, \MIC_a} \simeq \sigma_{\, \MIC_b}$ associated with
a given charge class. When discussing statistics 
it suffices to look at pairs 
$\MIC_a$, $\MIC_b$ having some hypercone 
$\MIC \subset \MIC_a^c \bigcap \MIC_b^c$, as is the case 
if $\MIC_a$ and $\MIC_b$ are spacelike separated, 
cf.~\ref{A.8}.  Choosing a funnel of hypercones 
$\{ \tcC_n \subset \MIC \}_{n \in \NN}$ yields
extensions  of ${\sigma}_{\, \MIC_a}$, ${\sigma}_{\, \MIC_b}$
to morphisms $\tsigma_{\, \MIC_a}, \tsigma_{\, \MIC_b}$ 
both acting on the common domain 
 $\tAl(V) \doteq \,\varinjlim \Al({\tcC}_n{}^{\! c})^-$,
\mbox{cf.\ Lemma \ref{extensions}.} 
According to part (iii) of this lemma
the spaces of intertwiners $({\sigma}_{\, \MIC_b}, {\sigma}_{\, \MIC_a})$
and  $(\tsigma_{\, \MIC_b}, \tsigma_{\, \MIC_a})$
coincide. Moreover, 
$({\sigma}_{\, \MIC_b}, {\sigma}_{\, \MIC_a}) \subset \Ri(\MIC^c) \subset
\tAl(V)$ using the localization properties of 
the morphisms and hypercone duality. Thus the unitary intertwiner 
$W \in ({\sigma}_{\, \MIC_b}, {\sigma}_{\, \MIC_a})$ 
is unique up to a phase, and 

$$
{\varepsilon} ({\sigma}_{\, \MIC_a}, {\sigma}_{\, \MIC_b}) 
\doteq W^{-1}  \tsigma_{\, \MIC_a}(W) =
\tsigma_{\, \MIC_b} (W) W^{-1} 
$$
is well defined.
By construction, ${\varepsilon} ({\sigma}_{\, \MIC_a}, {\sigma}_{\,
  \MIC_b}) $ is an intertwiner in 
$(\sigma_{\, \MIC_a} \product \sigma_{\, \MIC_b}, \sigma_{\, \MIC_b} 
\product \sigma_{\, \MIC_a})$.
The following lemma shows that it 
is an intrinsic quantity depending only on the given morphisms.

\begin{lemma}  \label{independence}
The intertwiner ${\varepsilon} ({\sigma}_{\, \MIC_a}, {\sigma}_{\,
  \MIC_b}) $ is independent of the 
extensions of the given morphisms when
chosen as above.
\end{lemma}
\begin{proof}
The operator 
 $ \tsigma_{\, \MIC_a}(W)$ is independent of the choice of funnel
contained in a given hypercone $\MIC \subset \MIC_a^c \bigcap \MIC_b^c$ 
as the corresponding extensions of $\sigma_{\, \MIC_a}$ 
coincide on $\Ri(\MIC^c)$ by Lemma \ref{extensions}(i). 
Next, if $\MIC_0 \subset \MIC_a^c \bigcap \MIC_b^c$ is another hypercone,
there is a 
hypercone $\ttcC_1 \subset \MIC_0$ making 
$\ttcC_1{}^{c} \bigcap \MIC^c$ hypercone connected, 
\ie this region contains with any pair of hypercones 
a path of hypercones interpolating 
between them, cf.~\ref{A.7}. 
We take $\ttcC_1$ as initial member of a funnel 
$\{ \ttcC_n \subset \MIC_0 \}_{n \in \NN}$ and
consider the corresponding extension $\ttsigma_{\, \MIC_a}$
of $\sigma_{\, \MIC_a}$. 
By hypercone connectivity, there is
an interpolating path of hypercones 
$\MIC_k \subset \ttcC_1{}^c \bigcap \MIC^c $, $k = 1, \dots , m$,
with $\MIC_1 = \MIC_a$ and $\MIC_m = \MIC_b$. Consequently
we can write the intertwiner $W$ as product 
$W = W_1 \cdots W_m$ with $W_k \in \Ri(\MIC_k) \subset 
\Ri(\ttcC_1{}^{c}) \bigcap \Ri(\MIC^c)$, $k = 1, \dots , m$.
Since both $\ttsigma_{\, \MIC_a}$ and
$\tsigma_{\, \MIC_a}$ have the algebra 
$\Ri(\ttcC_1{}^{c}) \bigcap \Ri(\MIC^c)$ in their 
respective domains,  
it follows from Lemma \ref{extensions}(ii) that 
$\ttsigma_{\, \MIC_a}(W_k) = \tsigma_{\, \MIC_a}(W_k)$,
$k = 1, \dots , m$. Hence  
$W^{-1} \ttsigma_{\, \MIC_a}(W) = W^{-1} \tsigma_{\, \MIC_a}(W)$,
as claimed. 
\end{proof} 

With this information we can establish that each simple charge 
class has a definite
``statistics parameter''.

\begin{proposition} 
Given a simple charge class $\mathscr{C}$ 
and the corresponding family of
morphisms $\sigma_{\, \MIC} : \AlV \rightarrow \RiV$. 
\begin{enumerate}
\item[(a)] There is a statistics parameter
$\varepsilon_{\raisebox{-1.3pt}{$\scriptstyle \mathscr{C}$}} \in \{ \pm 1 \}$, 
depending only on the charge class, such that 
${\varepsilon} ({\sigma}_{\, \MIC_a}, {\sigma}_{\, \MIC_b}) = 
\varepsilon_{\raisebox{-1.3pt}{$\scriptstyle \mathscr{C}$}} $ 
for any pair of morphisms ${\sigma}_{\, \MIC_a}$,  ${\sigma}_{\, \MIC_b}$
localized in spacelike separated hypercones $\MIC_a$, $\MIC_b$. 
\item[(b)] The statistics parameter $\varepsilon_{\quer{\mathscr{C}}}$ 
of the corresponding 
conjugate charge class $\quer{\mathscr{C}}$ 
has the same value, $\varepsilon_{\quer{\mathscr{C}}} = 
\varepsilon_{\raisebox{-1.3pt}{$\scriptstyle \mathscr{C}$}}$.
\end{enumerate} 
\end{proposition}
\begin{proof}
Let $\MIC_a$, $\MIC_b$ be spacelike separated hypercones. 
Then ${\sigma}_{\, \MIC_a} \product {\sigma}_{\, \MIC_b} = 
{\sigma}_{\, \MIC_b} \product {\sigma}_{\, \MIC_a}$
by Lemma~\ref{exchange}(i). Since these composed morphisms are  
members of a simple charge class 
by Proposition \ref{prop1}(iii), the corresponding unitary 
self--intertwiners are multiples of the identity. Hence  
${\varepsilon} ({\sigma}_{\, \MIC_a}, {\sigma}_{\, \MIC_b}) = \varepsilon \,
1$ for some phase factor $\varepsilon \in \TT$. Choosing a
hypercone $\MIC \subset \MIC_a^{\, c} \bigcap \MIC_b^{\, c}$,
cf.~\ref{A.8}, and 
extensions $\tsigma_{\, \MIC_a}$, $\tsigma_{\, \MIC_b}$
based on a funnel contained in $\MIC$, 
Lemma \ref{independence} gives
${\varepsilon} ({\sigma}_{\, \MIC_a}, {\sigma}_{\, \MIC_b}) = 
W^{-1} \tsigma_{\, \MIC_a}(W) = \tsigma_{\, \MIC_b} (W) W^{-1}$, 
where $W \in ({\sigma}_{\, \MIC_b}, {\sigma}_{\, \MIC_a})$. 
Now given a hypercone $\MIC_A \subset \MIC_a$ and a morphism
$\sigma_{\, \MIC_A} \simeq  {\sigma}_{\, \MIC_a}$ there is a 
unitary intertwiner 
$W_A \in (\sigma_{\, \MIC_a}, {\sigma}_{\, \MIC_A}) \subset \Ri(\MIC_a)$.
Hence $W_A W \in (\sigma_{\, \MIC_b}, \sigma_{\, \MIC_A})$ and  
computing gives
\begin{equation*}
{\varepsilon} ({\sigma}_{\, \MIC_A}, {\sigma}_{\, \MIC_b}) 
= \tsigma_{\, \MIC_b} (W_A W) W^{-1} W_A^{-1} 
= W_A \, \tsigma_{\, \MIC_b} (W) W^{-1} W_A^{-1} 
= {\varepsilon} ({\sigma}_{\, \MIC_a}, {\sigma}_{\, \MIC_b}) \, ,
\end{equation*}
where we used the localization properties of $\sigma_{\, \MIC_b}$
and the fact that 
${\varepsilon} ({\sigma}_{\, \MIC_a}, {\sigma}_{\,  \MIC_b})$ 
is a multiple of the identity. Hence 
${\varepsilon} ({\sigma}_{\, \MIC_a}, {\sigma}_{\,  \MIC_b})$
is independent of the choice of both ${\sigma}_{\, \MIC_a}$ and 
$\sigma_{\, \MIC_b}$ within
their respective localization cones $\MIC_a$ and $\MIC_b$. 
The spacelike complement of
a hypercone is hypercone path connected, 
cf.~\ref{A.6}, so it follows from Lemma~\ref{independence} and a 
standard \mbox{interpolation} argument that 
$ \varepsilon_{\raisebox{-1.3pt}{$\scriptstyle \mathscr{C}$}} \doteq 
{\varepsilon} ({\sigma}_{\, \MIC_a}, {\sigma}_{\,  \MIC_b})$
is independent of the  
choice, both of the spacelike separated hypercones  $\MIC_a$, $\MIC_b$
and of the morphisms within the given simple 
charge class~$\mathscr{C}$. Thus,  
${\varepsilon} ({\sigma}_{\, \MIC_a}, {\sigma}_{\,  \MIC_b})  =
{\varepsilon} ({\sigma}_{\, \MIC_b}, {\sigma}_{\,  \MIC_a})$ 
 and,  consequently, 
$$
\varepsilon_{\raisebox{-1.3pt}{$\scriptstyle \mathscr{C}$}}^2 = 
{\varepsilon} ({\sigma}_{\, \MIC_a}, {\sigma}_{\,  \MIC_b})^2 =
{\varepsilon} ({\sigma}_{\, \MIC_a}, {\sigma}_{\,  \MIC_b}) \,  
{\varepsilon} ({\sigma}_{\, \MIC_b}, {\sigma}_{\,  \MIC_a}) =
W^{-1} \tsigma_{\, \MIC_a}(W) \, \tsigma_{\, \MIC_a} (W^{-1}) W
= 1 \, ,
$$
proving the first part of the proposition.

To prove the second, we pick a hypercone $\MIC$
and spacelike separated hypercones $\MIC_a, \MIC_b \subset \MIC$.
There is then a hypercone $\tcC \subset \MIC^{\, c}$
and a corresponding morphism $\sigma_{\, \tcC}$ in the given 
charge class. Let $W_a \in (\sigma_{\, \tcC}, {\sigma}_{\, \MIC_a})$
and $W_b \in (\sigma_{\, \tcC}, {\sigma}_{\, \MIC_b})$
be unitary intertwiners. Then
$\sigma_{\, \MIC_a} \rest \Ri(\MIC) =
\mbox{Ad} \, W_a \rest \Ri(\MIC)$, \
$\sigma_{\, \MIC_b} \rest \Ri(\MIC) =
\mbox{Ad} \, W_b \rest \Ri(\MIC)$ \ and 
\mbox{$W \doteq W_a W_b^{-1} \in (\sigma_{\, {\MIC}_b}, {\sigma}_{\, \MIC_a})
\subset \Ri(\MIC)$}. So  
$$
\varepsilon_{\raisebox{-1.3pt}{$\scriptstyle \mathscr{C}$}} = 
{\varepsilon} ({\sigma}_{\, \MIC_a}, {\sigma}_{\,  \MIC_b})  
= W^{-1} \tsigma_{\, \MIC_a}(W) = 
W_b W_a^{-1} W_a(W_a W_b^{-1}) W_a^{-1} = W_b W_a  W_b^{-1} W_a^{-1}
\, .
$$
As far as the conjugate goes,    
the argument used in the first part of the proof of
Proposition \ref{prop2}, gives  \
\mbox{$\qsigma_{\, \MIC_a} \rest \Ri(\MIC) =
\mbox{Ad} \, W_a^{-1} \rest  \Ri(\MIC)$ and
$\qsigma_{\, \MIC_b} \rest \Ri(\MIC) =
\mbox{Ad} \, W_b^{-1} \rest  \Ri(\MIC)$}. 
Moreover, the 
last part of that same proof implies 
$\qsigma_{\, \MIC_a}(W^{-1}) \in 
(\qsigma_{\, \MIC_b}, \qsigma_{\, \MIC_a})$ for \  
{$W \in ({\sigma}_{\, \MIC_b}, {\sigma}_{\, \MIC_a}) \subset \Ri(\MIC)$}.
Hence 
$$ \qsigma_{\, \MIC_a}(W_b W_a^{-1}) =
W_a^{-1}(W_b W_a^{-1}) W_a = W_a^{-1}  W_b \in
(\qsigma_{\, \MIC_b}, \qsigma_{\, \MIC_b}) \, .
$$ 
A similar computation for the corresponding 
conjugate charge class~$\quer{\mathscr{C}}$ yields, 
$$
\varepsilon_{\quer{\mathscr{C}}} = 
\varepsilon(\qsigma_{\, \MIC_a},  \qsigma_{\, \MIC_b}) 
= W_b^{-1} W_a \, \qsigma_{\, \MIC_a} (W_a^{-1}  W_b)
= W_b^{-1} W_a^{-1}  W_b W_a \, . 
$$ 
But, by the above equality,    
$W_b W_a = \varepsilon_{\raisebox{-1.3pt}{$\scriptstyle \mathscr{C}$}}
\, W_a W_b$,  
hence 
$\varepsilon_{\quer{\mathscr{C}}} =  
\varepsilon_{\raisebox{-1.3pt}{$\scriptstyle \mathscr{C}$}}$,
completing the proof.  
\end{proof}

As expressed in the criterion, the simple charge 
classes $\mathscr{C}$ of a theory are 
in one--to--one correspondence with  
equivalence classes of morphisms in $\Sigma(V)$, modulo the 
equivalence relation $\simeq$ introduced above.
The preceding analysis shows that the 
structure of simple charge classes 
is analogous to that of simple sectors in 
superselection theory \cite{DoHaRo3}. We summarize these results.

\begin{theorem}
Let $\Sigma(V)$ be the family of all hypercone localized morphisms 
satisfying the criterion.
\begin{enumerate}
\item[(i)] For any given pair of morphisms 
$\sigma_{\, \MIC_a}, \tau_{\, \MIC_b} \in \Sigma(V)$ 
there exists a composed morphism
$\sigma_{\, \MIC_a} \product \tau_{\, \MIC_b} : \AlV \rightarrow
\RiV$. By composing on the left with 
$\omega_0$  it determines 
a  charge class of states, the simple composite class, 
depending only on
the charge classes of the given morphisms.
\item[(ii)] The composition of charge classes is commutative.
Given any two classes, 
morphisms $\sigma_{\, \MIC_a}$ and $\tau_{\, \MIC_b}$ 
can be picked, one from each of the classes, such that 
$\sigma_{\, \MIC_a} \product \tau_{\, \MIC_b} =
\tau_{\, \MIC_b} \product \sigma_{\, \MIC_a} $
when $\MIC_a$ and $\MIC_b$ are spacelike separated. 
\item[(iii)]  A simple charge class has a simple conjugate
charge class: for any morphism $\sigma_{\, \MIC} \in \Sigma(V)$ 
in the given class there is a 
$\qsigma_{\, \MIC} \in \Sigma(V)$ in the conjugate class with 
\mbox{$\qsigma_{\, \MIC} \product \sigma_{\, \MIC} = \sigma_{\, \MIC} \product 
\qsigma_{\, \MIC} = \iota$}.
\item[(iv)] To any simple charge class there corresponds a 
statistics parameter 
$\varepsilon \in \{\pm 1\}$,
characteristic of Bose and Fermi statistics, respectively. 
The conjugate charge class has the same statistics.
\end{enumerate}
\end{theorem}

\noindent {\bf Remark.} 
The first three parts of this theorem imply that 
$\Sigma(V)/\!\simeq$ is an Abelian group whose product is implemented 
by the composition of morphisms.
This group is to be interpreted as the dual of the global gauge group
deduced from the intrinsic structure of the 
charge classes, cf.\ the analogous result for 
superselection sectors in \cite{DoHaRo3}.

\vspace*{2mm} 
We conclude by pointing out 
that the preceding results on  the structure of simple charge 
classes are independent of our {\it ad hoc}
choice of hypercones. We have selected a family
$\cF$ of hypercones, based on a given hyperboloid 
$\HY$. Selecting another hyperboloid~$\HY^\prime$, there is
another family $\cF^{\, \prime}$ of hypercones based 
on it. Now, as shown in the 
appendix, cf.~\ref{A.9} and \ref{A.10}, given a 
hypercone $\MIC \in \cF$, there are hypercones 
$\mathbf{\check{\MIC}}, \mathbf{\hat{\MIC}} \in \cF^{\, \prime}$ with
$\mathbf{\check{\MIC}} \subset \MIC \subset 
\mathbf{\hat{\MIC}}$ and 
{\it vice versa}. Since the preceding arguments involve only the 
partial ordering and the causal relations between hypercones 
our structural results on
simple charge classes do not change 
if the family $\cF$  
is replaced by any other family $\cF^\prime$.

The results of this section, show that it suffices  
in the following to denote the morphisms by $\sigma$ rather 
than $\sigma_{\MIC}$,  \ie without singling out a choice $\MIC$ of 
localization hypercone. If localization matters we 
write $\sigma\in\Sigma(\MIC)$ to indicate 
that $\sigma$ is localized in $\MIC$.

\section{Covariant morphisms}
\label{covariantmorphisms}
\setcounter{equation}{0}

Since 
the semigroup $\Semi$ of spacetime transformations only acts as 
endomorphisms on the observables in $\AlV$, the usual way of 
describing the transport of states and morphisms makes no sense 
here. Hence it is not obvious how covariant morphisms and their 
charge classes are to be defined. 

\begin{definition}
A morphism $\sigma \in \Sigma(V)$ 
is covariant if first, for some neighbourhood of the 
identity $\cN_\cS \subset \Semi$,  there are morphisms  
\mbox{${}^\lambda \! \sigma : \AlV \rightarrow \RiV$}, $\lambda \in \cN_\cS$,    
looking like the original morphism on the transformed algebra,~\ie 
\begin{equation}  \label{transport}
{}^\lambda \! \sigma \compos \alpha_\lambda = \alpha_\lambda \compos
\sigma \, , \quad \lambda \in  \cN_\cS \, .
\end{equation} 
Secondly, there are unitary intertwiners 
$\Gamma_\lambda \in ({}^\lambda \sigma, \sigma)$,
$\lambda \in \cN_\cS$, such that 
\begin{equation}  \label{covtransport}
\alpha_\lambda(\Gamma_\mu) \in ({}^{\lambda \mu} \! \sigma ,
{}^\lambda \! \sigma) \, , \quad \lambda, \mu,  \lambda \mu \in \cN_\cS \, .  
\end{equation} 
This condition expresses the idea that the morphisms ${}^{\lambda} \!
\sigma $ all carry the same charge and are transported 
covariantly by physical operations.
Finally, there is a strong operator continuous section
\begin{equation}  \label{contsections} 
\mbox{$\lambda \mapsto \Gamma_\lambda \in ({}^\lambda\sigma, \sigma)$} 
\end{equation} 
of unitary intertwiners over $\cN_\cS$.
(At the expense of additional technical complications, 
continuity can be relaxed to measurability.)
\end{definition} 

\vspace*{2mm}
\noindent {\bf Remark.} Relations
(\ref{transport}) and (\ref{covtransport}) imply 
${}^{\lambda \mu} \! \sigma \compos \alpha_\lambda 
= \alpha_\lambda \compos {}^{\mu} \! \sigma$ for
$\lambda, \mu, \lambda \mu \in \cN_\cS$. 

\vspace*{2mm}
A covariant morphism $\sigma$ determines 
a continuous unitary local projective (ray) representation of 
$\Semi$ on $\cN_\cS$ by putting 
\begin{equation} \label{projectiverep}
U_\sigma(\lambda) \doteq \Gamma_\lambda \, U_0(\lambda) \, , 
\quad \lambda \in \cN_\cS \, , 
\end{equation}
where $U_0$ is the continuous unitary representation 
of $\Poin$ in the vacuum representation. Relation (\ref{covtransport}) 
implies that
$\Gamma_\lambda \alpha_\lambda(\Gamma_\mu) \in 
({}^{\lambda \mu} \sigma,\sigma)$  and 
$ \Gamma_{\lambda \mu}
\in ({}^{\lambda \mu} \sigma,\sigma)$ differ  
at most by a phase. Hence for $\lambda, \mu, \lambda \mu \in
\cN_\cS$ 
there is a 
$\zeta(\lambda, \mu) \in \TT$ such that 
$$ 
U_\sigma(\lambda)  U_\sigma(\mu) = 
\Gamma_\lambda \alpha_\lambda(\Gamma_\mu) \, U_0(\lambda \mu)
= \zeta(\lambda, \mu) \, \Gamma_{\lambda \mu} U_0(\lambda \mu)  = 
\zeta(\lambda, \mu) \,  U_\sigma(\lambda \mu) 
\, .
$$
Moreover, relation (\ref{transport}) gives
$$ 
\text{\rm Ad} \,  U_\sigma(\lambda) \compos \sigma =
\text{\rm Ad} \, \Gamma_\lambda \compos  \alpha_\lambda \compos \sigma  =
\text{\rm Ad} \, \Gamma_\lambda \compos {}^\lambda \! \sigma 
\compos \alpha_\lambda = \sigma \compos \alpha_\lambda \, , \quad 
\lambda \in \cN_\cS \, .
$$
Thus $U_\sigma$ is a local 
projective representation of $\Semi$ inducing
the corresponding local action 
on the observables in the representation $\sigma$.
Its continuity follows from that of $\lambda \mapsto
\Gamma_\lambda$. Applying  a well known result of
Bargmann \cite{Ba} we can establish the following. 

\begin{proposition}  \label{bargmann}
Let  $\sigma \in \Sigma(V)$ be a covariant morphism. There is 
a continuous unitary representation 
$\tU_\sigma$ of the covering group 
$\tPoin \doteq \RR^4 \rtimes \text{SL}(2,\CC) $
of  $\Poin$ such that 
$ \text{\rm Ad} \,  {\widetilde U}_\sigma(\widetilde{\lambda}) 
\compos \sigma = \sigma \compos \alpha_\lambda$ 
and $\tU_\sigma(\widetilde{\lambda})  
U_0(\lambda)^{-1} \in ({}^\lambda \sigma, \sigma)$
for $\widetilde{\lambda} \in 
\tSemi \doteq \Trans \rtimes \text{SL}(2,\CC)$. 
Here $\widetilde{\lambda} \mapsto \lambda$ is the 
canonical covering map  from the covering group 
to the Poincar\'e group.

\end{proposition} 

\begin{proof} The crucial step is to show  
that the local projective representation~$U_\sigma$ 
on $\cN_\cS \subset \Semi$, defined above, can 
be extended to a local projective 
representation on some
neighbourhood of the identity $\cN_\cP \subset \Poin$. 
Without loss of generality we assume that 
$\cN_\cS = \cN^T_+ \times \cN^L$ where 
$\cN^T_+ = \{ x \in \RR^4 : 2 |\bx| \leq x_0 + 
|\bx| < 2 \varepsilon \} \subset \Trans$ 
is a double cone for any given $\varepsilon > 0$
and  $\cN^L \subset \Lore$ is some neighbourhood of $1$. 
Then $\cN^T \doteq 
\{ x \in \RR^4 : |\bx| < \varepsilon, 
|x_0| + |\bx| < 2 \varepsilon \} \supset \cN^T_+$ 
is a neighbourhood of $0 \in \RR^4$ and we can put 
$\cN_\cP \doteq \cN^T \times \cN^L$. 

The desired extension of $U_\sigma$ requires 
several steps. First, we note that definition (\ref{projectiverep}) 
implies $U_\sigma(0,1) \in (\sigma, \sigma)$, hence by adjusting
phases we may assume \mbox{$U_\sigma(0,1) = 1$}. In a second step we
extend $U_\sigma$ to the translations $x \in \cN^T$. 
Given any such $x$ we write 
$x = (x - \kappa(x) e) + \kappa(x) e$, where 
$\kappa(x) \doteq (x_0 - |\bx|)$ and $e = (1, \mathbf{0})$ denotes
the time direction in the chosen coordinate system. Note that
both $(x - \kappa(x) e), \, |\kappa(x)| e \in \cN^T_+$ so 
definition 
$$
\mathbf{\breve{\mathnormal{U}}}_\sigma (x,1) \doteq 
 U_\sigma(x - \kappa(x) e, 1) \cdot 
\begin{cases}
U_\sigma(\kappa(x) e, 1) & \text{if} \quad \kappa(x) \geq 0 \\
U_\sigma(|\kappa(x)| e, 1)^{-1} & \text{if} \quad \kappa(x) \leq 0 
\end{cases}
$$
is consistent. As $U_\sigma \rest \cN^T_+$ is a local  
projective representation, the group theoretic commutators
of the corresponding unitaries are multiples of the identity,
and it is easy to verify that 
$\mathbf{\breve{\mathnormal{U}}}_\sigma$ yields a local projective
representation of $\cN^T$. Moreover 
$\mathbf{\breve{\mathnormal{U}}}_\sigma \rest \cN^T_+ $ coincides
with  $U_\sigma \rest \cN^T_+$ up to a phase. Lastly, 
for $\lambda = (x,\Lambda) \in \cN_\cP$, we put    
$\mathbf{\breve{\mathnormal{U}}}_\sigma(\lambda)
\doteq \mathbf{\breve{\mathnormal{U}}}_\sigma (x,1) \, 
U_\sigma(0,\Lambda)$. Since $U_\sigma$ is a local projective
representation of $\cN_\cS$ one has 
$U_\sigma(0,\Lambda) U_\sigma(y,1) = \zeta \,  U_\sigma(\Lambda y,1)
 U_\sigma(0,\Lambda) $ for 
$y, \Lambda y \in \cN^T_+$, $\Lambda \in \cN^L$ and
some phase factor  $\zeta$, hence 
$U_\sigma(0,\Lambda) U_\sigma(y,1)^{-1} = \overline{\zeta} \,  
U_\sigma(\Lambda y,1)^{-1} U_\sigma(0,\Lambda)$. 
Using these equalities, another 
easy computation shows that 
$\mathbf{\breve{\mathnormal{U}}}_\sigma$ is a local 
projective representation of $\cN_\cP$, \ie
$\mathbf{\breve{\mathnormal{U}}}_\sigma (\lambda)
\mathbf{\breve{\mathnormal{U}}}_\sigma (\mu) =
\xi(\lambda, \mu) \, \mathbf{\breve{\mathnormal{U}}}_\sigma (\lambda
\mu)$ for $\lambda, \mu, \lambda \mu \in \cN_\cP$ and
phase factors \mbox{$\xi(\lambda, \mu) \in \TT$}. It is 
continuous on $\cN_\cP$ because of the continuity  
inherited from $U_\sigma$ and, by construction,  
$\mathbf{\breve{\mathnormal{U}}}_\sigma \rest \cN_\cS$ coincides
with $U_\sigma$ modulo some phase factors. 

Now by the results of Bargmann \cite{Ba}, exploiting 
the phase freedom in the definition of 
$\mathbf{\breve{\mathnormal{U}}}_\sigma$ in some 
neighbourhood of the identity $\cN_\cP \subset \Poin$ leads to
a true continuous unitary representation still 
denoted by $\mathbf{\breve{\mathnormal{U}}}_\sigma$. 
As the covering group is locally isomorphic to $\Poin$, its 
local representation induces a local 
continuous unitary representation 
$\tU_\sigma$ of $\tPoin$, 
given by $\tU_\sigma(\widetilde{\lambda}) \doteq
\mathbf{\breve{\mathnormal{U}}}_\sigma(\lambda)$, 
$\widetilde{\lambda} \in
\widetilde{\cN}_\cP$. The covering group being simply 
connected, there is a unique extension of $\tU_\sigma$
to a strongly continuous unitary representation of
$\tPoin$ got by representing its elements as
finite products of elements close to the identity (monodromy theorem). 
This establishes the existence of  
$\tU_\sigma$. Furthermore, for any $A \in \AlV$ one has 
$$ 
\mbox{Ad} \, \tU_\sigma (\widetilde{\lambda}) \compos \sigma(A)
= \mbox{Ad} \, \mathbf{\breve{\mathnormal{U}}}_\sigma (\lambda) 
\compos \sigma(A)
= \sigma \compos \alpha_\lambda (A) \, , \quad 
\widetilde{\lambda} \in \widetilde{\cN}_\cS \, .
$$ 
Thus iterating, the first and last members of this
equality are equal for all 
$\widetilde{\lambda} \in \tSemi$. 
Finally, 
$\tU_\sigma (\widetilde{\lambda}) U_0(\lambda)^{-1} 
\in ({}^\lambda \sigma, \sigma) \subset \RiV$
for $\widetilde{\lambda} \in \widetilde{\cN}_\cS$. Hence as for  
$\widetilde{\lambda}_1, \dots , \widetilde{\lambda}_n \in 
\widetilde{\cN}_\cS$
\begin{equation*}
\begin{split}
\tU_\sigma (\widetilde{\lambda}_n \cdots \widetilde{\lambda}_1) 
& U_0(\lambda_n \cdots \lambda_1)^{-1} \\
& =
\big( \tU_\sigma(\widetilde{\lambda}_{n})  U_0(\lambda_n)^{-1}
\big) \, U_0(\lambda_n) \big(
\tU_\sigma (\widetilde{\lambda}_{n-1} \cdots \widetilde{\lambda}_1) 
U_0(\lambda_{n-1} \cdots \lambda_1)^{-1}  \big) \, 
 U_0(\lambda_n)^{-1}
\end{split}
\end{equation*}
and $U_0(\lambda_n) \RiV U_0(\lambda_n)^{-1} \subset \RiV$
it follows by induction that 
\mbox{$\tU_\sigma (\widetilde{\lambda}) U_0(\lambda)^{-1} \in
\RiV$} for any $\widetilde{\lambda} \in \tSemi$, completing the proof. 
\end{proof}

Our previous, operationally inspired characterization of covariant 
morphisms \mbox{$\sigma \in \Sigma(V)$} involved an associated 
covariant family of morphisms 
$({}^\lambda \sigma, \sigma$), $\lambda \in \cN_\cS$.
This raises the question of whether the resulting 
unitary representation of the Poincar\'e group depends on the 
choice of such a family. We will answer this question in the 
subsequent lemma, where we show that this representation
is uniquely fixed by $\sigma \in \Sigma(V)$. 

\begin{lemma}
Let $\sigma \in \Sigma(V)$ be a covariant morphism, then 
the associated unitary representation $\tU_\sigma$ of $\tPoin$  
given in the preceding proposition is unique.
\end{lemma}

\begin{proof}
Let $\tU_j$,  $j = 1,2$, be unitary representations
of $ \tPoin$ as in the preceding proposition. Then
$ \text{\rm Ad} \,  \tU_1(\widetilde{\lambda}) 
\compos \sigma =\sigma \compos \alpha_\lambda = 
\text{\rm Ad} \, \tU_2 (\widetilde{\lambda}) \compos \sigma$ 
and hence \ 
$ \text{\rm Ad} \,  \tU_2(\widetilde{\lambda})^{-1}  
\tU_1(\widetilde{\lambda}) \compos \sigma = \sigma$ 
for $\widetilde{\lambda} \in \tSemi$.
Recalling that $\sigma(\AlV)^{-} =  \RiV$ this implies 
$$ 
\tU_2(\widetilde{\lambda})^{-1}  
\tU_1(\widetilde{\lambda}) \in  \RiV^{\prime} 
 \, , \quad \widetilde{\lambda} \in 
\tSemi \, .
$$ 
Moreover, for such $\widetilde{\lambda}$, 
$\tU_j(\widetilde{\lambda}) U_0(\lambda)^{-1} 
\in \RiV$, $j = 1,2$, and consequently   
$$
\tU_2(\widetilde{\lambda})^{-1}  
\tU_1(\widetilde{\lambda}) 
\in  U_0(\lambda)^{-1} \RiV \, U_0(\lambda) \, , \quad 
\widetilde{\lambda} \in \tSemi \, .
$$
Restricting $\widetilde{\lambda}$ in the
preceding two relations to the subgroup 
$\text{SL}(2,\CC)$ and bearing in mind that in this case  
$U_0(\lambda)^{-1} \RiV \, U_0(\lambda) =  \RiV $ 
and that $\RiV$ is a factor it follows that 
{$\tU_2(\widetilde{\lambda})^{-1}  
\tU_1(\widetilde{\lambda}) \in \TT \, 1$} 
for  $\widetilde{\lambda} \in \text{SL}(2,\CC)$.
Since there are no non--trivial one--dimensional 
representations of the Lorentz group, 
the restrictions of $\tU_1$, 
$\tU_2$ to $\text{SL}(2,\CC) $ 
coincide. 

Turning to the translations, let $x \in \Trans$ and let
$\Delta(x) \doteq \tU_2(x,1)^{-1}  
\tU_1(x,1)$. Then, by the preceding step,
$\Delta(x) \in U_0(x,1)^{-1} \, \RiV \, U_0(x,1) \, 
\bigcap \, \RiV^{\prime}$. Hence, again using 
$\tU_1(\widetilde{\lambda}) U_0(\lambda)^{-1} 
\in \RiV$
for $\widetilde{\lambda} \in \tSemi$, gives  
$$ \tU_1(\widetilde{\lambda})^{-1} \Delta(x) \,  
\tU_1(\widetilde{\lambda}) = 
U_0(\lambda)^{-1}  \Delta(x) \,  U_0(\lambda)   \, , \quad 
 \widetilde{\lambda} \in \tSemi \, .
\quad  $$
On the other hand, for $\widetilde{\lambda} \in \text{SL}(2,\CC)$  
\begin{equation*}
\begin{split}
\tU_1(\widetilde{\lambda})^{-1} \Delta(x)  
\tU_1(\widetilde{\lambda}) & = 
\tU_2(\widetilde{\lambda})^{-1} 
\tU_2 (x)^{-1} \tU_2(\widetilde{\lambda}) \, 
\tU_1(\widetilde{\lambda})^{-1} 
\tU_1(x) U_1(\widetilde{\lambda})  \\ 
& =  \tU_2({\lambda}^{-1} x)^{-1}  
\tU_1({\lambda}^{-1} x)
=  \Delta({\lambda}^{-1} x)  \, ,
\end{split}
\end{equation*}
where the first equality follows since $\tU_1$, 
$\tU_2$ coincide on $\text{SL}(2,\CC)$ 
and the second since $\tU_1$, $\tU_2$
are unitary representations of $\tPoin$. 
Combining the preceding two relations yields 
$$ 
U_0(\Lambda)  \Delta(x)  U_0(\Lambda)^{-1} =  \Delta(\Lambda x) \, , 
\quad x \in \Trans \, , \ \Lambda \in \Lore \, . 
$$
Now, given any lightlike translation $l \in  \Trans$, 
there is a corresponding family of boosts 
$\{ \Lambda_s \in \Lore \}_{s \in \RR}$ scaling $l$,
\ie $ \Lambda_s \, l = e^{-s} \, l $, $s \in \RR$. 
The preceding equality and the 
continuity of $\tU_1$, $\tU_2$
involved in the definition of $\Delta$ show that 
$$ 
\lim_{s \rightarrow \infty} U_0(\Lambda_s) \Delta(l)  U_0(\Lambda_s)^{-1}
=  \lim_{s \rightarrow \infty} \Delta(e^{-s} l) = 1 \, ,
$$
in the strong operator topology.
Taking matrix elements of this equation in the vector state given by 
$\Omega$, bearing in mind that $\Omega$ is
invariant under the action of $U_0(\lambda)$, leads to
$(\Omega, \Delta(l) \Omega) =1 $. Hence 
$\Delta(l) \Omega = \Omega$ since $\Delta(l)$ 
is unitary. But $\Omega$ is separating for $ \RiV^{\prime}$, 
so $\Delta(l) = 1$ and consequently
$\tU_1(l) = \tU_2(l)$ for lightlike 
translations $l \in \Trans$. As the linear span
of lightlike translations generates the subgroup of all translations
and $\tU_1$, $\tU_2$ are 
representations of $\tPoin$,  
they coincide on $\RR^4$,
and hence on the whole group. \end{proof}

Let $\sigma \in \Sigma(V)$ be a covariant morphism and 
let $\tU_\sigma$ be the associated representation
of $\tPoin$. Then any other 
equivalent morphism $\sigma^\prime \simeq \sigma$ is also
covariant and the corresponding representation is given by 
$\tU_{\sigma^\prime}(\widetilde{\lambda}) =
W \tU_\sigma (\widetilde{\lambda}) W^{-1}$, 
$\widetilde{\lambda} \in \tPoin$, 
where $W \in (\sigma , \sigma^\prime)$. Up till now 
the specific localization of the covariant morphisms
did not matter, but to proceed further we need to have a 
closer look at them. 

\begin{lemma} \label{cocycleloc}
Let $\MIC$ be any given hypercone 
and let $\sigma \in \Sigma(\MIC)$ be a covariant morphism
with associated representation $\tU_\sigma$ of 
$\tPoin$. There is a hypercone 
$\MIC_0 \supset \MIC$ (depending only on $\MIC$) 
and a neighbourhood of the identity 
$\widetilde{\cN}_\cS \subset \tSemi$ with
\mbox{$\tU_\sigma(\widetilde{\lambda}) U_0(\lambda)^{-1}
\in \Ri(\MIC_0)$} for $\widetilde{\lambda} \in \widetilde{\cN}_\cS $.
\end{lemma}

\begin{proof}
We put $\widetilde{\Gamma}_{\widetilde{\lambda}} 
\doteq \tU_\sigma(\widetilde{\lambda}) \,
U_0(\lambda)^{-1}$, $\widetilde{\lambda} \in \tPoin$. 
These unitaries satisfy the cocycle equation 
$\widetilde{\Gamma}_{\widetilde{\lambda}} 
\alpha_\lambda(\widetilde{\Gamma}_{\widetilde{\mu}})
= \widetilde{\Gamma}_{\widetilde{\lambda} \widetilde{\mu}}$. 
Moreover, 
\begin{equation} \label{locinfo}
\mbox{Ad} \,  \Gamma_{{\lambda}} \compos \alpha_\lambda \compos \sigma 
= \sigma \compos \alpha_\lambda \, ,
\quad  {\lambda} \in \Semi \, , 
\end{equation}
where we have set $\Gamma_\lambda \doteq \pm \, 
\widetilde{\Gamma}_{\widetilde{\lambda}}$
since phases drop out in the adjoint action. In the 
subsequent argument we anticipate the existence of  
hyperbolic cones with certain 
specific geometric properties. This will be justified
at the end of the proof.
Evaluating the preceding equality, the localization of $\sigma$, 
for Lorentz transformations,
${\Lambda} \in  \Lore$, gives 
$\mbox{Ad} \, \Gamma_{{\Lambda}} \rest \Al(\MICc_1) = \iota$ 
for any hypercone \mbox{$\MIC_1 \supset \MIC \bigcup \Lambda \MIC$}.
Thus $\pm \, \Gamma_{{\Lambda}} \in \Ri(\MIC_1)$ by
hypercone duality and making $\MIC_1$ sufficiently big,
this inclusion holds for all $\Lambda$ in some neighbourhood
of the identity $\cN^L \subset \Lore$.

As the family of hypercones based on a given
hyperboloid is not stable under translations, analyzing 
the localization of the 
corresponding cocycles requires more work. Since $\MIC_1\supset\MIC$,  
$\sigma(\Al(\MIC_1))^- =
\Ri(\MIC_1)$ 
equation (\ref{locinfo}) allows us to conclude that 
$ \mbox{Ad} \, \Gamma_{{\lambda}} \compos \alpha_\lambda
(\Ri(\MIC_1))  \subset  \Ri(\MIC_2) $
for $\lambda \in \Semi$ and any 
hypercone $\MIC_2 \supset \MIC \bigcup \lambda \MIC_1$.
We now use the cocycle equation.
Since $(0,\Lambda) \, (x,1) = (\Lambda x, 1) (0,\Lambda)$,  
$\Gamma_\Lambda \, \alpha_\Lambda (\Gamma_x) = 
\pm \, \Gamma_{\Lambda x} \, \alpha_{\Lambda x} (\Gamma_\Lambda) $
for $x \in \Trans$, $\Lambda \in \Lore$ and consequently 
$$
\alpha_\Lambda(\Gamma_x) \, \Gamma_{\Lambda x}^{-1}
= \pm \, \Gamma_\Lambda^{-1} \ 
\Gamma_{\Lambda x} \, \alpha_{\Lambda x}(\Gamma_\Lambda) \, 
\Gamma_{\Lambda x}^{-1} \, \in \, \Ri(\MIC_2) \, ,
$$ 
provided $\MIC_1 \supset \MIC  \bigcup \Lambda \MIC$
and $\MIC_2 \supset  \MIC_1 \bigcup \, (\MIC_1 + \Lambda x)$. 
We exploit this information choosing sequences of boosts 
and translations $\Lambda_n \in \Lore$, $l_n \in \Trans$ 
where $\Lambda_n l_n = l$ is a fixed (lightlike) vector, 
$\MIC_1 \supset \Lambda_n \MIC$  and $l_n$ tends to $0$. Thus
$\alpha_{\Lambda_n}(\Gamma_{l_n}) \, \Gamma_l^{-1} \in \Ri(\MIC_2)$,
$n \in \NN$, where $\MIC_2 \supset \MIC_1 \bigcup (\MIC_1 + l)$. Now 
$\alpha_{\Lambda_n}(\Gamma_{l_n}) \, \Omega \rightarrow \Omega$ 
since $\Gamma_{l_n} \rightarrow 1$ in the strong operator
topology  and $U_0(\Lambda_n) \, \Omega = \Omega$. 
Moreover, if $l \in \Trans$ is sufficiently close to $0$ there
is a double cone $\MIO \subset \MICc_1 \bigcap (\MICc_1 + l)$,
and, after a moment's reflection, relation~(\ref{locinfo}) shows that all
operators $\alpha_{\Lambda_n}(\Gamma_{l_n})$, $n \in \NN$, commute
with the elements of~$\Al(\MIO)$. Hence, by the Reeh--Schlieder
property of the vacuum, 
$\alpha_{\Lambda_n}(\Gamma_{l_n}) \rightarrow 1$ in the strong
operator topology and consequently $\Gamma_l \in \Ri(\MIC_2)$. 
Varying the direction of the chosen sequence of boosts slightly,  
the convex hull $\cK^T \subset \Trans$ of the resulting lightlike 
vectors $l$ is the tip of a convex cone 
with open interior. Using the 
cocycle equation once more \mbox{$\Gamma_{l_1} 
\alpha_{l_1}(\Gamma_{l_2}) = \pm \, \Gamma_{l_1 + l_2}$}, 
we see that  $\Gamma_{k} \in \Ri(\MIC_2)$ if $k \in \cK^T$.  
Given any neighbourhood $\cN^T_+ \subset \Trans$ of $0$,
$\alpha_x(\Gamma_k) \in \Ri(\MIC_0)$, $x \in \cN^T_+$, 
for any hypercone $\MIC_0 \supset \MIC_2 + \cN^T_+$.
Now if $\cN^T_+$ is small enough, there is a 
$k \in \cK^T$ such that $k + \cN^T_+ \subset \cK^T$. 
Since $\Gamma_x =  \pm \Gamma_{k +x} \alpha_x(\Gamma_k)^{-1}$, 
$\Gamma_x \in \Ri(\MIC_0)$, $x \in \cN^T_+$. 
Choosing $\cN_\cS \doteq \cN^T_+ \times \cN^L \subset \Semi$,
a final application of the cocycle equation yields 
${\Gamma}_\lambda \in \Ri(\MIC_0)$, $\lambda \in \cN_\cS$, 
and identifying $\widetilde{\cN}_\cS$ with $\cN_\cS$
by the covering map the result follows. 

It remains to summarize the geometrical assumptions on hypercones
made in the preceding argument. Given $\MIC$ we anticipated 
that there is a hypercone \mbox{$\MIC_1 \supset \Lambda \MIC$} 
for $\Lambda$ varying in some neighbourhood \mbox{$\cN^L \subset
  \Lore$} of the identity. In addition, we required that there were 
sequences of boosts $\Lambda_n$ acting in a fixed direction
(within some open set of directions), such that 
$\Lambda_n \, \MIC \subset \MIC_1$, $n \in \NN$ and 
lightlike vectors $l_n \rightarrow 0$
scaling as $\Lambda_n l_n = l$, $n \in \NN$
for fixed $l$. 
These properties characterize the hypercone $\MIC_1$. Moreover, we needed that, 
given a bounded set of translations $\cB^T \subset \Trans$ there
is a hypercone $\MIC_2 \supset \MIC_1 + \cB^T$. 
The result then followed. The existence of such
hypercones is 
established  in the appendix, cf.~\ref{A.11}, \ref{A.12} and  \ref{A.13}. 
\end{proof}

This information allows us to establish the basic   
properties of covariant morphisms and their charge classes  
under composition and conjugation. 

\begin{definition}
Let $\mathscr{C}$ be a simple charge class. The class is said
to be covariant if there is a covariant morphism 
$\sigma \in \Sigma(V)$ with 
$\omega_0 \compos \sigma \in \mathscr{C}$. The family of  
covariant morphisms is denoted by $\Sigma_c(V) \subset \Sigma(V)$
and the subset of covariant morphisms localized in a given
hypercone $\MIC$ is denoted by $\Sigma_c(\MIC)$.  
\end{definition} 

\vspace*{0.1mm} 
\begin{theorem} \label{covariancethm}
The family  $\Sigma_c(V)$ of covariant morphisms
is stable under composition and conjugation. More explicitly, for
any pair $\sigma_1, \sigma_2 \in \Sigma_c(\MIC)$ one has
$\sigma_1 \product \sigma_2 \in \Sigma_c(\MIC)$ and for any 
$\sigma \in \Sigma_c(\MIC)$ there is a $\qsigma \in \Sigma_c(\MIC)$ 
such  that $\qsigma \product \sigma = \sigma \product \qsigma = \iota$. 
\end{theorem}

\begin{proof} Since all morphisms in the equivalence class of 
a covariant morphism are covariant there is no loss of 
generality in picking any hypercone 
$\MIC$ and morphisms $\sigma_1, \sigma_2 \in \Sigma_c(\MIC)$.
Let $\tU_j$ be the associated unitary 
representations of $\tPoin$ and let 
$\widetilde{\Gamma}_j(\widetilde{\lambda}) = 
\tU_j(\widetilde{\lambda}) U_0(\lambda)^{-1}$, 
$\widetilde{\lambda} \in \tPoin$
be the corresponding cocycles,  $j = 1,2$. By the
preceding lemma there is a hypercone $\MIC_0 \supset \MIC$
and a neighbourhood $\cN_\cS \subset \Semi$ 
of the identity (the image of $\widetilde{\cN}_\cS$
under the covering map) such that 
$\Gamma_{j \lambda} \in \Ri(\MIC_0)$ for $\lambda \in \cN_\cS$, 
$j = 1,2$. We choose an extension $\tild{\sigma}_1$ of 
$\sigma_1$ normal on 
$\Ri(\MIC_0)$  and having the range of $\sigma_2$ in its 
domain, cf.~Lemma \ref{extensions}(i). Then 
$$
\sigma_1 \product \sigma_2 \compos \alpha_\lambda =
\tild{\sigma}_1 \compos \mbox{Ad} \, \Gamma_{2 \lambda} 
\compos \alpha_\lambda \compos \sigma_2  =
\mbox{Ad} \, \tild{\sigma}_1 (\Gamma_{2 \lambda}) \, \Gamma_{1 \lambda}
\compos \alpha_\lambda \compos \sigma_1 \product \sigma_2 \, ,
\quad \lambda \in \cN_\cS \, .
$$
Putting 
$\Gamma_{1 2 \, \lambda} \doteq \tild{\sigma}_1(\Gamma_{2 \, \lambda})
\,  \Gamma_{1 \, \lambda}$, \ 
$\sigma_{1 2} = \sigma_1 \product \sigma_2$ and  \ 
${}^\lambda \sigma_{1 2} \doteq 
\mbox{Ad} \, \Gamma_{12 \, \lambda}^{-1} \compos \sigma_{1 2}$, 
the preceding equality reads  
${}^\lambda \sigma_{1 2} \compos \alpha_\lambda 
= \alpha_\lambda \compos \sigma_{1 2}$ and 
$\Gamma_{1 2 \, \lambda} \in ({}^\lambda \sigma_{1 2},  \sigma_{1
  2})$,  $\lambda \in \cN_\cS$.
Moreover, if 
$\lambda, \mu, \lambda \mu \in  \cN_\cS$, 
$\Gamma_{12 \, \lambda}^{-1} \, \Gamma_{12 \, \lambda \mu} \in
({}^{\lambda \mu} \sigma_{1 2} , {}^\lambda \sigma_{1 2} )$ and,  
on the algebra $\AlV$, 
$$
\Gamma_{12 \, \lambda}^{-1} \, \Gamma_{12 \, \lambda \mu} 
= \Gamma_{1 \lambda}^{-1} \, \tild{\sigma}_1(\Gamma_{2 \lambda}^{-1}
\,  \Gamma_{2 \lambda \mu}) \, \Gamma_{1 \lambda \mu}
=  \Gamma_{1 \lambda}^{-1} 
\tild{\sigma}_1(\alpha_\lambda(\Gamma_{2 \mu})) 
\Gamma_{1 \lambda} \alpha_\lambda(\Gamma_{1 \mu})
= \alpha_\lambda(\Gamma_{1 2 \, \mu}) \, ,
$$
where we used the cocycle equations for 
$\Gamma_1, \Gamma_2$ (second equality) and the
covariance of $\sigma_1$ (third equality).
Thus $ \alpha_\lambda(\Gamma_{1 2 \, \mu}) 
\in ({}^{\lambda \mu} \sigma_{1 2} , {}^\lambda \sigma_{1 2} )$. 
Since $\lambda \mapsto \Gamma_{12 \, \lambda}$ is continuous
on $\cN_\cS$ in the strong operator topology by the 
continuity of the cocycles $\Gamma_1$, $\Gamma_2$
and the normality of $\tild{\sigma}_1$ on $\Ri(\MIC_0)$, 
$\sigma_{1 2}$ is covariant, \ie 
$\sigma_{1} \product \sigma_2 \in \Sigma_c(\MIC)$.  

Turning to conjugation, let $\sigma \in \Sigma_c(\cC)$, $\tU(\lambda)$
its associated unitary representation and 
cocycle $\Gamma_\lambda$, $\lambda \in \cN_\cS$,  
and let $\qsigma \in \Sigma(\cC)$ be the conjugate morphism which exists
by Proposition \ref{prop2}. We choose extensions
$\tild{\sigma}$, $\tild{\qsigma}$ of $\sigma$, $\qsigma$
with a common stable domain and normal on $\Ri(\MIC_0)$,  
cf.\ Lemma \ref{extensions}(i). Note that 
by continuity $\tild{\qsigma} \compos \tild{\sigma} =
\tild{\sigma} \compos \tild{\qsigma} = \iota$ on this domain
and $\tild{\qsigma}(\Ri(\MIC_0)) \subset \Ri(\MIC_0)$ since 
$\MIC_0 \supset \MIC$.  Putting
$\quer{\Gamma}_\lambda \doteq \tild{\qsigma}(\Gamma_\lambda^{-1}) $ and 
${}^\lambda \qsigma \doteq \mbox{Ad} \, \quer{\Gamma}_\lambda^{-1} \compos \qsigma
= \tild{\qsigma} \compos \mbox{Ad} \, \Gamma_\lambda$, $\lambda \in
\cN_\cS$, we can compute on~$\AlV$ 
$$
\tsigma \compos {}^\lambda \qsigma \compos \alpha_\lambda =
\tsigma \compos  \tild{\qsigma} \compos \mbox{Ad} \, \Gamma_\lambda
\compos \alpha_\lambda =
\mbox{Ad} \, \Gamma_\lambda \compos \alpha_\lambda =
\mbox{Ad} \, \Gamma_\lambda \compos \alpha_\lambda \compos \tsigma
\compos \tild{\qsigma} =
\tsigma \compos \alpha_\lambda  \compos \qsigma \, . 
$$
Composing this equality on the left with 
$\tild{\qsigma}$ gives 
${}^\lambda \qsigma \compos \alpha_\lambda =
\alpha_\lambda  \compos \qsigma$ 
and, by construction, 
{$\quer{\Gamma}_\lambda \in ({}^\lambda \qsigma, \qsigma) $},
$\lambda \in \cN_\cS$.  
Moreover, if $\lambda, \mu, \lambda \mu \in \cN_\cS$, 
$$
\tsigma(\quer{\Gamma}_\lambda^{-1} \, \quer{\Gamma}_{\lambda \mu}) =
\Gamma_\lambda \, \alpha_\lambda(\Gamma_\mu^{-1}) \, \Gamma_\lambda^{-1} =
\mbox{Ad} \, \Gamma_\lambda \compos \alpha_\lambda \compos
\tsigma \compos \tild{\qsigma} (\Gamma_\mu^{-1}) =
\tsigma \compos \alpha_\lambda \, 
 (\tild{\qsigma} \, (\Gamma_\mu^{-1})) \, ,
$$
on $\AlV$, where the first equality uses the cocycle equation. 
Composing with $\tild{\qsigma}$ gives 
$ \alpha_\lambda(\quer{\Gamma}_\mu) = 
\alpha_\lambda(\tild{\qsigma} \, (\Gamma_\mu^{-1})) 
= \quer{\Gamma}_\lambda^{-1} \, \quer{\Gamma}_{\lambda \mu} \in 
({}^{\lambda \mu} \qsigma, \qsigma)$. The 
normality of $\tild{\qsigma}$ implies
that $\lambda \mapsto \quer{\Gamma}_\lambda$ is continuous on 
$\cN_\cS$. Thus $\qsigma$ is covariant, \ie 
$\qsigma \in \Sigma_c(\MIC)$,  completing the proof.
\end{proof}

The preceding results show that, even in the presence of charges, 
the energy and the angular momentum of the partial states in 
the light cone $V$ can be defined in a mathematically precise 
and physically meaningful way. 
Given a simple charge class, this
information is encoded in the  spectral properties of the 
generators of the unitary representations $\tU_\sigma$ 
of $\tPoin$, where $\sigma$ is any one of the equivalent morphisms 
associated with the class. Yet, even though these generators
are uniquely fixed and, as we have seen, can be reconstructed from 
data in $V$, they should not be interpreted as 
genuine quantum observables since they are not affiliated with  
the algebra $\sigma(\AlV)^-$. They contain not only
pertinent information about the
states  in $V$ but also, in a consistent, though 
hypothetical way, some information on outgoing massless 
particles (radiation) created in the past that evades direct
observations in $V$. This hypothetical input enters with the condition  
of covariance expressing, within the mathematical setting, 
the postulate that measurements and operations can be repeated 
at any time using exactly the same procedures. The generators 
incorporate the implicit assumption that this postulate  
applies to the distant past, too, involving, as it does the 
chosen extension of time translations from the semigroup action 
to the full group. Although this hypothesis seems plausible and is 
fully consistent with physical predictions, it cannot 
be verified experimentally. This  explains 
why the  generators are not described by 
observables in the light cone $V$.

\section{Spectral properties}
\label{spectrum} 
\setcounter{equation}{0}

We now analyze the spectral 
properties of covariant charge classes and
their associated covariant morphisms. As we shall see, the energy
momentum of the corresponding states is bounded below, 
in accordance with the physical idea that these states describe stable 
elementary systems. However,
the standard ``additivity of the energy''
argument used in sector analysis 
can not be applied in the present setting,  
because the necessary asymptotic commutativity of the
transported morphisms is lacking. Yet, using the asymptotic 
commutation of hypercone localized operators 
under the action of suitable Lorentz boosts, we can  
establish a somewhat weaker spectral result, 
still enabling us to prove  that all covariant 
charge classes satisfy the spectrum condition. 
We begin by stating the main technical 
result of this section.

\begin{lemma} \label{technical}
Let 
$\sigma_1, \, \sigma_2, \, \sigma_1  \product \sigma_2 \in \Sigma(\MIC)$
be covariant morphisms 
and let $\tU_{\sigma_1}, \tU_{\sigma_2}$ and $\tU_{\sigma_1 \product
  \sigma_2}$ 
be the corresponding unitary
representations of the (covering group of the) Poincar\'e group. 
If $\tU_{\sigma_1} \rest \RR^4$ or 
$\tU_{\sigma_2} \rest \RR^4$ (or both) violate the 
relativistic spectrum condition,
then $\tU_{\sigma_1 \product \sigma_2} \rest \RR^4$ also 
violates the spectrum condition.
\end{lemma}

\noindent Since the spectrum of translations in the vacuum class is
contained in $\Spec$, applying this lemma to covariant morphisms 
and their conjugates immediately gives the following basic result. 

\vspace*{2mm}
\begin{theorem} Let 
$\sigma \in \Sigma_c(V)$ be a covariant morphism then $\tU_\sigma$,  
the corresponding unitary
representation of the (covering group of the) Poincar\'e group
$\Poin$, satisfies the relativistic spectrum 
condition, \ie $\text{\rm sp} \, \tU_\sigma \rest \RR^4  \subset \Spec$.
\end{theorem} 

The proof of the preceding lemma is rather technical. 
We therefore first outline  
the idea of the argument and subsequently 
explain the details. \smallskip 

If a representation $\tU$ of the (covering group of the)
Poincar\'e  group violates the spectrum condition,
then, for every positive lightlike vector $l$,
the unitary one--parameter group 
$\zeit \mapsto \tU(\zeit l,1)$ 
 has spectrum on the negative real axis. For, otherwise, 
as the spectrum is Lorentz invariant,
none of these groups would have spectrum on the negative real axis.
Hence the spectrum of $\tU \rest \RR^4$ would be contained in 
the intersection \ $\bigcap_{\, l} \, \{ p : pl \geq 0 \} 
= \Spec$, a contradiction. To explore the spectrum of
the one parameter groups, we fix some 
positive lightlike vector $l$ and a corresponding 
one parameter group \ $\Lambda$ 
of Lorentz boosts such that $\Lambda(\beta) \, l = \beta^{-1} \, l$,
$\beta \geq 1$. Note that there is always an opposite 
positive lightlike
vector $l^{\prime}$ scaling under the action of these 
boosts as
$\Lambda(\boost) \, l^\prime = \boost \, l^\prime$,
$\boost \geq 1$. To simplify notation in what follows,
we set \mbox{$T(\zeit) \doteq \tU(\zeit l, 1)$}, \ 
$B(\boost) \doteq \tU(0, \widetilde{\Lambda} (\boost))$, noting that
\mbox{$B(\boost) T(\zeit) = T(\boost^{-1} \zeit)  B(\boost)$}.      
Whereas these groups are globally defined, care is needed 
if one wants to determine their action
on morphisms, as they just act as endomorphisms. 
Since the spectral properties of the unitary representations
$\tU$ affiliated with a charge class do not depend on 
the particular choice of morphism we can
choose  the localization properties of the morphisms at will 
and adjust them to the geometric action
of the given boosts and translations. \smallskip

Let  $T_0$, $T_1$, $T_2$, $T_{1 \, 2}$  and 
 $B_0$,  $B_1$, $B_2$, $B_{1 \, 2}$ be the 
unitary groups corresponding to the given lightlike 
translations and boosts in the 
vacuum representation and in the representations induced
by the three given covariant families of morphisms.
If the spectrum condition is violated
in the charge class of $\sigma_2$ we choose its 
localization cone $\MIC_2$ to point 
asymptotically in the   
direction of the lightlike vector $l^\prime$ opposite to 
the given $l$. To analyze the spectral properties of 
$T_{1 \, 2}$ relative to those 
of $T_2$  we consider the sequence of operators 
\begin{equation} \label{testsequence} 
A_n \doteq B_2(\boost_n) \int \! d \zeit \, f(\zeit) \, 
T_2(\zeit + \zeit_n) \, A \, T_0(\zeit + \zeit_n)^{-1} \, B_0(\boost_n)^{-1} \, ,
\quad n \in \NN \, ,
\end{equation}
where $A \in \Al(\MIC_2)$ is any local operator, $f$ 
any test function and the integral is defined in the strong 
operator topology. These operators are designed to 
exploit the spectral properties of $T_2$ and
$T_0$ by choosing    
the support of the Fourier transform of $f$ appropriately. 
The usual strategy in sector analysis is  
to look at these operators
in the representation induced by $\sigma_1$ where they 
ought to give information on the spectral properties of $T_{1 2}$ 
relative to $T_1$. But there is a problem. The 
test functions $f$ are analytic, the integral extends over 
all of~$\RR$ and the resulting operators are not localized 
in the light cone $V$, \ie they do not lie in the domain 
of~$\sigma_1$. The sequence of 
shifts $\zeit_n$ serves to move them asymptotically into~$V$, 
without obliterating the information on 
the spectral properties of the operators. But another  
problem arises: they transport the operators 
into the future of the localization region 
of any given morphism $\sigma_1$ so that they can interfere 
with the associated charges. The resulting effects are difficult 
to control. 
The solution is to let  boosts in the direction of $l^\prime$ 
with an increasing sequence 
of rapidities act on the operators. As a result, 
the operators are contracted  
towards the boundary 
of~$V$ and pushed simultaneously
towards spacelike infinity in the direction of~$l^\prime$.  
Even though these boosts blur the information on the spectral 
properties of the operators, they do not mix
the positive and negative spectrum allowing us 
sufficient control on the spectrum of $T_{1 2}$.  
Choosing a morphism $\sigma_1$ whose localization cone 
avoids the asymptotic localization region of the 
resulting sequence of
operators, we can show that the spectrum condition must 
be violated in the charge class of~$\sigma_1 \product \sigma_2$
if it is violated in the charge class of $\sigma_2$, or interchanging
the roles of $\sigma_1$ and $\sigma_2$, in that of $\sigma_1$. 
We come now to the actual proof of the lemma. 

\begin{proof}[Proof of Lemma \ref{technical}:] \ 
If the 
unitary representation $\widetilde{U}_2$ of $\tPoin$
violates the spectrum condition, as outlined above, 
we fix a positive lightlike 
vector $l$ and choose a morphism $\sigma_2$ localized
in a hypercone  $\MIC_2$ pointing in the opposite 
positive lightlike direction~$l^\prime$. 
We will take advantage of  the localization 
of the cocycles corresponding to the representation $\tU_2$.
With the above notation,  
Lemma~\ref{cocycleloc} implies that there is a hypercone 
$\MIC_0 \supset \MIC_2$ such that 
$ T_2(\zeit) \, T_0(\zeit)^{-1},  B_2(\boost) \, B_0(\boost)^{-1} 
\in \Ri(\MIC_0)$ for sufficiently small $\zeit \geq 0$ and 
$\boost \geq 1$, respectively. As a matter of fact, since
$ \MIC_2$ points towards~$l^\prime$ and
$\Lambda(\boost) \, l^\prime = \boost \, l^\prime$, 
$\boost \geq 1$, we may in addition suppose that
$\Lambda(\boost) \, \MIC_0 \subset \MIC_0$,
$\boost \geq 1$, and $\Lambda(\boost) \, \MIC_0 \subset \MIOc$  
for any given compact region 
$\MIO \subset \FL$ and  sufficiently large $\boost$,
cf.~\ref{A.11}.  The cocycle equation then implies 
$T_2(\zeit) T_0(\zeit)^{ -1} \in \Ri(\cT_{\, \szeit})$
for any $\zeit \geq 0$, where 
$\cT_{\, \szeit} \doteq \bigcup_{0 \leq \vartheta \leq 
\szeit} (\MIC_0 + \vartheta l)$. Similarly, 
$B_2(\boost) B_0(\boost)^{ -1} \in 
\Ri(\MIC_0)$ for any $\boost \geq 1$, taking account of
the endomorphic action of the boosts $\Lambda(\boost)$, 
$\boost \geq 1$, on $\MIC_0$. \smallskip

As explained above, $T_2$ has  spectrum on the negative real 
axis $\RR_-$, so there is a compact set $K_2 \subset \RR_-$ such that 
$E_2(K_2) \neq 0$, where $E_2$ is the spectral resolution 
of $T_2$. We pick a  
test function $f : \RR \rightarrow \CC $ whose Fourier transform 
is equal to $1$ on $K_2$ and vanishes 
on the ray $({\RR}_+ - \kappa_2)$ for some $\kappa_2 > 0$.
We also choose a local operator $A \in \Al(\MIC_2)$
with $E_2(K_2) A \Omega \neq 0$ by invoking  
the Reeh--Schlieder property of~$\Omega$.
Inserting into the expression for the operators~$A_n$ 
of relation~(\ref{testsequence}), where  
$\zeit_n > 0$, $\boost_n \geq 1$, $n\in \NN$, gives 
suitable sequences to be adjusted in what follows. 
Since the support of $f$ is
all of $\RR$, the resulting operators are not
localized in $\FL$. We therefore pick
a test function $\chi : \RR \rightarrow \RR$
with support in the interval $[-1, 1]$ and equal to
$1$ in some neighbourhood of~$0$ and introduce  
approximating functions $ f_n$, putting  
$f_n(\zeit) \doteq \chi(\zeit / \zeit_n)  f(\zeit)$,
$n \in \NN$. Fixing a sequence 
$\zeit_n$, $n \in \NN$, tending to infinity, 
$f_n \rightarrow f$ in the Schwartz space topology.  
The operators 
\begin{equation} \label{approx}
A_{V,n} \doteq B_2(\boost_n) \,
\int \! d\zeit \, f_{n}(\zeit) \,
T_2(\zeit + \zeit_n) \, A  \, T_0(\zeit + \zeit_n)^{-1} 
B_0(\boost_n)^{-1} \, , \quad n \in \NN \,
\end{equation}
are elements of $\Ri(\cT_{2\szeit_n/\sboost_n})$,  bearing 
$\MIC_0 \, \bigcup \,  
\Lambda(\boost_n)  \cT_{2 \theta_n} \subset \cT_{2 \theta_n/
  \sboost_n}$ and
the localization of the cocycles
$B_2(\boost) B_0(\boost)^{-1}$, $T_2(\zeit) T_0(\zeit)^{-1}$ 
in mind as well as the way boosts act on the hypercone
$\MIC_0 \supset \MIC_2$ and the lightlike
vector~$l$. 
Moreover, $\| A_n - A_{V,n} \| \rightarrow 0$ as $n$ tends to 
infinity. 

We first analyze  the operators $A_{V,n}^* A_{V,n}$ in the
limit of large $n$. Since the boosts~$B_2$ 
cancel in these operators 
we have better control on their localization 
properties: they are localized in the region 
$\Lambda(\boost_n)  \, \cT_{2 \theta_n} =
\bigcup_{\, 0 \leq \vartheta \leq 
2 \szeit_n / \sboost_n} (\Lambda(\boost_n) \, \MIC_0 + \vartheta l)$. 
Fixing a sequence $\boost_n$ with 
$\zeit_n/\boost_n \rightarrow 0$ in the limit of large 
$n$ and bearing in mind that  
$\Lambda(\boost_n) \, \MIC_0 \subset \MIOc$  
for any given compact region 
$\MIO \subset \FL$ and  sufficiently large~$\boost_n$, 
locality implies that  $A_{V,n}^* A_{V,n}$, $n \in \NN$, 
is a central sequence in $\RiV$. Since the sequence is 
uniformly bounded and $\RiV$ is a factor, its weak limit points
are multiples of the identity.  In fact they all coincide 
and can be evaluated in the vacuum state, where
$
(\Omega, A_{V,n}^* A_{V,n} \Omega) =
\| \int \! d \zeit \, f_n(\zeit) \, T_2(\zeit) A \Omega \|^2 \, .
$
As a result we have 
\mbox{$
 \lim_n  A_{V,n}^* A_{V,n} = 
\| \int \! d \zeit \, f(\zeit) \, T_2(\zeit) A \Omega \|^2  \cdot 1
$}
in the weak operator topology.

Next, we pick a morphism $\sigma_1$ localized in a hypercone
$\MIC_1$ in the spacelike complement of $\cT_{\szeit}$ for some 
$\zeit > 0$. (Note that according to \ref{A.13} of the appendix
there is a hypercone $\cC_0 \supset \cT_{\szeit}$, whose opposite
cone can be taken as $\MIC_1$.)  
Choosing an extension $\tsigma_1$ normal on
$\Ri(\cT_{\szeit})$ a routine
computation gives
\begin{equation*}
\begin{split}
& \tsigma_1(A_{V,n}) =
\int \! d \zeit \, \boost_n f_n(\boost_n \zeit) \ 
\tsigma_1\big(T_2(\zeit + \zeit_n/\boost_n) \, 
B_2(\boost_n) A B_0(\boost_n)^{-1} \, T_0(\zeit +
\zeit_n/\boost_n)^{-1} \big) \\ 
& = \int \! d \zeit \, \boost_n f_n(\boost_n \zeit) \, 
T_{1 2}(\zeit) \ \tsigma_1 \big(T_{2}(\zeit_n/\boost_n) \, 
B_2(\boost_n) A B_0(\boost_n)^{-1} \, T_0(\zeit_n/\boost_n)^{-1} \big) \, 
T_1(\zeit)^{-1} \, ,
\end{split}
\end{equation*}
for sufficiently large $n$.
Here the first equality uses the commutation properties of 
boosts and lightlike translations, given above, and the second
equality the expression for the translations
in the composed representation $\sigma_1 \product \sigma_2$,
$$T_{1 2} (\zeit^\prime) = 
\tsigma_1(T_2(\zeit^\prime) T_0(\zeit^\prime)^{-1})
\, T_1(\zeit^\prime), \quad \zeit^\prime \geq 0,$$  
established in the proof of Theorem \ref{covariancethm}.
Note that the above integral extends over the region 
$(\zeit + \zeit_n/\boost_n) \geq 0$ 
by the support properties of $f_n$, so this
expression may be used here.  
Now let $E_1$ be the spectral resolution of $T_1$, let 
$K_1 \subset \RR$ be any compact set in its spectrum 
and let $\Phi_1 \in E(K_1) \cH$ be any non--zero vector. 
Furthermore, let $E_{1 2}$ be the spectral resolution 
of $T_{1 2}$. Then, for any bounded operator 
$B \in \cB(\cH)$,  the Fourier transform (in the sense of
distributions) 
of \mbox{$\zeit \mapsto E_{1 2}(\overline{\RR}_+) T_{1 2}(\zeit)
B T_1(\zeit)^{-1} \Phi_1$} has support in the 
region $(\RR_+ - \kappa_1)$ for some $\kappa_1 \in \RR$
(depending on the choice of~$K_1$). On the other hand,
with the above choice of the test function $f$, the Fourier transform of 
$\zeit \mapsto \boost_n f(\boost_n \zeit)$ vanishes 
in the region $(\RR_+ - \boost_n \kappa_2)$, where 
$\kappa_2 > 0$. Hence
$\int \! d \zeit \, \boost_n f(\boost_n \zeit)
 E_{1 2}(\overline{\RR}_+) T_{1 2}(\zeit) B T_1(\zeit)^{-1} \Phi_1 = 0 $ 
for sufficiently large~$\boost_n$. Now 
$(\boost_n f(\boost_n \zeit) -  \boost_n f_n(\boost_n \zeit))  
\rightarrow 0$ in $L^1(\RR)$ in the limit of large~$n$
and the sequence $B_n \doteq 
\tsigma_1 (T_{2}(\zeit_n/\boost_n) \, 
B_2(\boost_n) A B_0(\boost_n)^{-1} \, T_0(\zeit_n/\boost_n)^{-1} ) $
is uniformly bounded in this limit. Hence in 
the above expression for $\tsigma_1(A_{V,n})$ in the second integral 
we can replace
the function $\zeit \mapsto \boost_n f_n(\boost_n \zeit)$
by $\zeit \mapsto \boost_n f(\boost_n \zeit)$ since 
the difference tends to $0$ in norm. Taking account of 
$\tsigma_1(A_{V,n}) = A_{V,n}$  and the localization of $A_{V,n}$,
it is then clear that 
$\lim_n \|  E_{1 2}(\overline{\RR}_+)  A_{V,n} \Phi_1 \| = 0$.

Let us summarize the facts established so far. In the first step we  
have shown 
$$\lim_n \|  A_{V,n} \Phi_1 \|^2 
= \lim_n \, (\Phi_1,  A_{V,n}^*  A_{V,n} \Phi_1) =
\| \int \! d \zeit \, f(\zeit) \, T_2(\zeit) A \Omega \|^2 \, 
\| \Phi_1 \|^2.$$ 
The next step gives  
$$\lim_n \| (1 -  E_{1 2}(\overline{\RR}_+))  A_{V,n} \Phi_1 \|^2
=  \| \int \! d \zeit \, f(\zeit) \, T_2(\zeit) A \Omega \|^2 \, 
\| \Phi_1 \|^2 .$$ 
The support properties of 
the Fourier transform of $f$ and the choice of~$A$ yield 
$\| \int \! d \zeit \, f(\zeit) \, T_2(\zeit) A \Omega \|^2 
\geq \| E_2(K_2) A \Omega \|^2 \neq 0$. Hence
\mbox{$(1 -  E_{1 2}(\overline{\RR}_+)) \neq 0$}, proving that
$T_{1 2}$ has spectrum on the negative real axis. So
$\tU_{1 2}$ violates the spectrum condition if $\tU_2$
does. Since $\sigma_1 \product \sigma_2 \simeq 
\sigma_2 \product \sigma_1$, the proof of the lemma 
is completed interchanging 
$\sigma_1$ and $\sigma_2$. 
\end{proof}

\vspace*{-2mm} 
The preceding results accurately  
describe the energetic properties of the states of 
interest here. They 
also throw new light on the appearance of superselection rules 
in quantum field theory. From the present 
theoretical point of view, based on    
observations and operations performed in a
light cone~$V$, the creation of a charged state in~$V$ 
is achieved by creating a pair and pushing the 
opposite charge to lightlike infinity. In practice, however, 
this would require an unlimited amount of energy since
the opposite charge would have to be accelerated to the 
speed of light. The experimental creation of a 
charged state can therefore only be accomplished locally by
moving the opposite charge sufficiently far away
(``behind the moon''). 
In other words, superselection rules appear when  
the total charge in~$V$ cannot be changed by 
realistic physical operations due to an infinite energy barrier.
Nevertheless, the theoretical limit states on $V$  are 
meaningful idealizations allowing us to analyze the
properties of charges. As we have seen, relative to 
the vacuum, these states 
have finite energy bounded from below. The infinite energy 
needed for their creation from the vacuum is carried 
away by the opposite charge and not visible anymore in $V$ in 
the limit. 
 
\vspace*{-2mm} 
\section{The Minkowskian picture}
\label{minkowski}
\setcounter{equation}{0} 

\vspace*{-1mm} 
Throughout the preceding discussion, we have restricted 
attention to the algebra 
of observables in a given light cone $V$. Whether this algebra 
is part of a larger algebra 
in Minkowski space or not did not matter. Yet interestingly 
enough the framework 
established here can be used to construct an extension 
of the theory to a theory 
on Minkowski space $M$. Whereas this canonical extension may seem somewhat 
arbitrary from an observational point of view, 
it enables us to establish in a physically 
meaningful manner that our results do not depend on the 
choice of light cone $V$.

The basic ingredient is the relation between our semigroup 
$\Semi$ and the 
group~$\Poin$ it generates. $\Semi$ is a semidirect product of its 
subsemigroup of future directed (time) translations 
and the Lorentz group. The time translations are a normal 
subsemigroup generating the 
group of spacetime translations and $\Poin$ is the 
semidirect product of the group of spacetime translations 
and the Lorentz group. This tight relationship between $\Semi$ 
and $\Poin$ allows 
us to induce up structures relating to $\Semi$ to a 
corresponding structure relating to 
$\Poin$. Topological aspects of this principle can be
incorporated efficiently by observing that 
$\Semi$ contains an interior point of $\Poin$.\footnote{The 
explicit construction of $\cN_\cP$ 
from $\cN_\cS$ in Proposition~\ref{bargmann}
can be seen as a simple application of this general principle}

The first observation is that $V$ is just the quotient of 
$\Semi$ by the Lorentz 
group. The induced structure for $\Poin$ is Minkowski space $\MI$, 
the quotient of $\Poin$ 
by the Lorentz group. We leave it to the reader to formalize 
the construction of the 
affine space $\MI$ together with its Lorentzian metric 
as an extension of $V$ equipped 
with the analogous structure. Thus an observer in $V$, 
aware of the way $\Semi$ acts, 
can talk about Minkowski space.
The second observation is that the continuous unitary 
representation $U_0$ of the 
semigroup $\Semi$ determined by the vacuum state extends 
canonically to a 
continuous unitary representation of $\Poin$, as outlined in Sect.~3.
Moreover, the net $\Al$ of 
observable algebras in the given light cone $V$ extends to 
a net on Minkowski space. 

We recall that the extended net 
is fixed by assigning to each double cone 
\mbox{$\MIO_\MI \subset \MI$} 
the algebra \ $\Al_\MI(\MIO_\MI) \doteq 
U_0(x_\MI)^{-1} \, \Al(\MIO_\MI + x_\MI) \, U_0(x_\MI)$, 
where 
$x_\MI \in \Trans$ is such that $\MIO_\MI + x_\MI \subset V$.
This assignment defines an (isotonous) net $\Al_\MI$ 
on~$\MI$ which is local and covariant with regard to the 
adjoint action of $U_0(\lambda)$, $\lambda \in \Poin$. 
Moreover, the given vacuum state $\omega_0$ on $\AlV$ can be extended to a 
(pure) vacuum state $\omega_\MI$  on~$\Al_\MI(\MI) \doteq 
\bigcup_{\MIO_\MI \subset \MI}   \, \Al_\MI(\MIO_\MI) $, putting 
$$
\omega_\MI(A_\MI) \doteq (\Omega, A_\MI \Omega) \, ,
\quad A_\MI \in \Al_\MI(\MI)
\, ,
$$
where $\Omega$ is the cyclic vector in 
the GNS--representation of $\AlV$ induced by $\omega_0$.
These extensions of the net $\Al$ and
state $\omega_0$ are uniquely fixed by the hypothesis of a global 
(Minkowskian) vacuum state entering into 
the chosen extension of $U_0$ to~$\Poin$.  
Yet, as already mentioned, in the
presence of massless particles, other extensions 
consistently describe different prehistories of $\omega_0$. 

Nevertheless this canonical extension is useful since   
observations made in different light cones 
$V_1, V_2$ may be compared, provided 
the respective observers use 
their (partial) vacua as reference states  
and reconstruct the representation $U_0$
of $\Poin$ from their respective observations.
Every pair of light cones contains some common light cone
$V \subset V_1 \bigcap V_2$, so both observers have access
to $V$ and would, in principle, be able to reconstruct  the 
same global net $\Al_\MI$ and 
vacuum  state~$\omega_\MI$ from their partial information. 
If they conventionally regard their partial vacua as restrictions 
of the global vacuum state, 
$\omega_{j \, 0} = \omega_\MI \rest \Al_\MI(V_j)$, $j = 1,2$,  
it would make sense for the observers to compare 
their past data in their common future. 

The properties of simple covariant
charges as seen by observers in different light cones 
can likewise be compared on this same basis since the corresponding
\mbox{morphisms} can be extended to the Minkowskian net.
Let \mbox{$\sigma : \AlV \rightarrow \RiV$} be any covariant morphism and let 
$\tU_\sigma$ be the associated representation of $\tPoin$.    
The extension of $\sigma$ to the operators $A_\MI \in \Al_\MI(\MI)$ is given by
$$
\sigma_\MI(A_\MI) \doteq
\mbox{Ad} \, \tU_\sigma(x_\MI)^{-1} \compos
\sigma \compos \mbox{Ad} \, U_0(x_\MI) \, (A_\MI) \, ,
$$
where $x_\MI \in \Trans$ is so large that 
$\mbox{Ad} \, U_0(x_\MI) (A_\MI) \in \AlV$. This assignment yields  
a well defined morphism 
$\sigma_\MI : \Al_\MI(\MI) \rightarrow \cB(\cH)$ transforming
covariantly under the adjoint action of $\tU_\sigma$, as is 
easily verified. Note that this extension is completely fixed
by data available in $V$. 

Now given $V_0 \supset V$  
put $\sigma_{V_0} \doteq \sigma_\MI \rest \Al_\MI(V_0)$. 
If $\sigma$ is a simple hypercone localized 
morphism in $V$, then $\sigma_{V_0}$ is a simple hypercone 
localized morphism in $V_0$, \ie its composition with 
the vacuum $\omega_\MI$ leads to a charge class with all 
properties specified in the defining criterion. 
We refrain from giving a  detailed proof here but   
just indicate the main points. 
\medskip 

Let $t_0 \in \Trans$ with
$V_0 + t_0 = V$.  If the morphism 
$\sigma$ is localized in the hypercone $\MIC \subset V$, 
Lemma~\ref{cocycleloc} and the cocycle 
equation imply $\tU_\sigma(t_0) U_0(t_0)^{-1} \in \Ri(\cT_{t_0})$,
where $\cT_{t_0} \doteq \bigcup_{\, 0 \leq \varsigma \leq 1} 
\{ \MIC_0 + \varsigma t_0 \}$ and \mbox{$\MIC_0 \supset \MIC$} is some
hypercone in $V$. As shown in~\ref{A.13}, there is 
a hypercone \mbox{$\cT_{t_0} \subset \MIC_{t_0} \subset V$}.  
Hence, denoting the weak closures of 
algebras in the Minkowskian net by the symbol~$\Ri_\MI$, one has 
\begin{equation*}
\begin{split}
& \tU_\sigma(t_0)^{-1} U_0(t_0) \\
& =  U_0(t_0)^{-1} (\tU_\sigma(t_0) U_0(t_0)^{-1})^{-1}   U_0(t_0)
\in U_0(t_0)^{-1} \, \Ri(\MIC_{t_0}) \, U_0(t_0)
= \Ri_\MI(\MIC_{t_0} - t_0) \, ,
\end{split}
\end{equation*}
and $(\MIC_{t_0} - t_0)$ is a hypercone in $V_0$, based
on the shifted hyperboloid $(\HY - t_0)$. But 
$$
\sigma_{V_0} = 
\mbox{Ad} \, \tU_\sigma(t_0)^{-1} \compos
\sigma \compos \mbox{Ad} \, U_0(t_0) =
\mbox{Ad} \, \tU_\sigma(t_0)^{-1} U_0(t_0) \compos 
\mbox{Ad}  \,  U_0(t_0)^{-1} \compos 
\sigma \compos \mbox{Ad} \, U_0(t_0) \, ,
$$ 
so $\sigma_{V_0} : \Al_\MI(V_0) \rightarrow
\Ri_\MI(V_0)$ is localized in $(\MIC_{t_0} - t_0) \supset \MIC$. 
Moreover, if $\sigma, \tau$ are equivalent morphisms 
on $\AlV$ with intertwiners $W \in (\sigma, \tau) \subset \RiV$, then 
$\sigma_{V_0}, \tau_{V_0}$  are equivalent on $\Al_\MI(V_0)$ 
with intertwiners 
 $W_{V_0} \in (\sigma_{V_0}, \tau_{V_0}) \subset \Ri_\MI(V_0)$ given by 
$$
W_{V_0} = \tU_\tau(t_0)^{-1} W \tU_\sigma(t_0) =
(\tU_\tau(t_0)^{-1} U_0(t_0)) \, 
 (U_0(t_0)^{-1} W U_0(t_0)) \, (U_0(t_0)^{-1} \tU_\sigma(t_0)) \, .
$$
Letting the morphisms and their localization vary using  
Lorentz covariance one finds that 
the families of equivalent hypercone localized 
morphisms in $V$ are restrictions of families of 
equivalent hypercone localized morphisms in $V_0$. Thus the  
analysis of the preceding sections and the corresponding results 
apply to the latter families, too. 

Even though the morphisms and their 
intertwiners depend, in general, on the choice of light cone, 
important physical data, such as their statistics parameters 
$\varepsilon \in \{\pm 1 \}$ do not since the intertwiners 
$W_{V_0}$ depend continuously on $t_0$ as can be 
seen from the expression given above. Hence 
the statistics parameters stay constant under 
changes of the light cone. Moreover,
the spectral properties of the charge classes do not
depend on the choice of light cone, either. Since 
for any pair of light cones $V_1, V_2$ there is 
a light cone $V_0 \supset V_1 \bigcup V_2$,  
the respective observers will agree on 
the intrinsic properties of the charges. 

Despite the fact that the   
interpretation of charged states does not depend on the choice of
light cone, it would not make sense, in general, to take a global
view. For example, the equivalence of morphisms 
\mbox{$\sigma_\MI \rest \Al_\MI(V_0) \simeq \tau_\MI \rest \Al_\MI(V_0)$} 
for any light cone $V_0$ does \textit{not} imply that 
\mbox{$\sigma_\MI  \simeq \tau_\MI $ on $\Al_\MI(\MI)$}, the two
representations may differ by infrared clouds of massless
particles superselected in~$\MI$. Moreover,
from the Minkowskian point of view, the morphisms~$\sigma_\MI$
in general do not have localization properties allowing
composition and an analysis of statistics. Thus, for the 
interpretation of theories with 
long range forces, the restriction to light cones 
is not only physical meaningful but also solves these 
infrared problems.

Another aspect of the preceding analysis is 
worth mentioning. The species of localizable charges, 
exhaustively treated from the Minkowskian point of view 
in \cite{DoHaRo1,DoHaRo2,DoHaRo3}, also fit into the
present light cone setting. It is a distinctive feature of this type of
charge that the intertwiners between morphisms localized in 
a light cone~$V$ do not change if one proceeds to larger
cones $V_0$. In contrast, for charges 
related to a local gauge group, such as the electric 
charge, the intertwiners depend on $V_0$. Thus,
from the present point of view, there 
is a clear cut distinction between these two types of 
charges. This prompts us to ask the intriguing question 
of whether this feature can be used to recover both 
the global gauge group and the structure of the 
underlying local gauge group
from the structure of the charged states. 
This would solve a longstanding problem in the algebraic  
approach to local quantum physics \cite{Ha2}.

\section{Concluding remarks}
The present investigation establishes a new framework
for interpreting physical states in 
relativistic quantum field theory. 
Instead of implicitly supposing that the 
properties  of physical states can be controlled 
in all of Minkowski space, we have from the outset 
taken the arrow of time into account: missing measurements 
and operations in the past in general mean missing 
information in the future. 
From an experimental point of view, the best 
physicists can hope for is to explore the properties 
of (partial) states in future  
light cones. Theory only needs 
to interpret and explain such data. 

In a completely massive world, our present approach would 
not make a difference as the observables in any forward
light cone would be irreducible. So one would not loose information
about states by belated experiments. The situation is 
different, however, in the presence of massless particles. 
Since, as a consequence of Huygens principle and Einstein causality,  
there is no way of acquiring information on 
outgoing massless particles (radiation)
created in the past of a light cone, 
the family of observables in a 
light cone is highly reducible. 
Although there are numerous \mbox{\textit {a priori}} 
possibilities for the type 
of the algebra generated by these observables, they are 
generically of type~III${}_1$ according to the classification of Connes. 
 
From a physical point of view, this type comes closest to the familiar 
irreducible case as it still allows to describe in a comprehensive 
way the transitive quantum effects 
of physical operations on states carrying the same charge.
So although the superposition principle no longer holds
for states on these algebras, concepts familiar 
from the analysis of superselection sectors such as the 
creation and transport of charges, their statistics and their 
conjugation still make good sense. 
Moreover, their energetic aspects can be consistently described. 

The advantage of the present approach is 
that the infrared problems caused by infinite clouds of 
low energy massless particles disappear. In the traditional treatment 
one tries to solve these problems by splitting the massless particle
content into an energetically soft and 
therefore unobservable part and a hard 
part. This method, however,  breaks  
Lorentz invariance and is incompatible with the strict 
locality of the observables.
In the present approach the massless particle content is split into
a marginal part defying observation since it is already 
in the spacelike complement of the observer and 
an elemental part accessible to observations and
operations within his  lightcone. This splitting is both Lorentz invariant 
and compatible with Einstein causality. It thereby permits a 
consistent description and analysis of elementary systems 
even in the presence of long range forces.

Since the present approach does not aim to treat the marginal part of the 
outgoing massless particles it even admits unitary implementations of 
Lorentz transformations for charged 
light cone morphisms. Thus it is  worthwhile posing the 
question of whether electrically charged states can 
be described as Wigner particles 
(irreducible representations of the Poincar\'e group)
in the present setting.
Such a result seems not to conflict 
with previous insights into Minkowski theories \cite{FrMoSt,Bu1,Bu2}
where spacelike asymptotic properties do matter  
but have no counterpart in the present approach. Since in a 
light cone the observed 
massless contributions can be described by 
Fock states with a finite  particle number there may be
partial states, describing a single electrically 
charged particle where the (globally inevitable) accompanying
radiation field has no observable effects
in the cone. Such states
could well contribute to an atomic part 
in the mass spectrum. Thus the infraparticle problem, too,
may disappear if one just allows 
observations in a light cone.

One of the most intriguing aspects of the present approach
is the endomorphic nature of time evolution, entering as it does 
in the interpretation of the microscopic theory. This feature, 
combined with the perpetual loss of control
of outgoing massless particles might be relevant to a 
better understanding of the classical aspects of our quantum world.
Observations on outgoing radiation
cannot be affected anymore 
by later quantum experiments (unless the 
observer has taken timely precautions 
to reflect it back into his 
light cone). Hence those results can be taken as facts in the 
sense of classical physics without conflicting with 
the principles of quantum theory. This aspect seems to 
warrant a more detailed study.

In the present analysis, we have restricted attention to 
simple charges satisfying Bose or Fermi statistics.
Just as in the case of localizable charges in Minkowski space, 
substantial parts of our analysis can be extended to the more
general case of charges of arbitrary finite statistics
with one notable exception: we have not been able to 
show that the property of covariance of composite morphisms is 
stable under forming subrepresentations and to establish 
the spectrum 
condition. Since the physically successful theories coupling matter 
to the electromagnetic field do not lead to para\-statistics
there might be a deeper reason for our failure.
But caution is needed in drawing such a conclusion and we 
will pursue this question elsewhere.

\begin{appendix}
\section{Appendix} 

In this appendix we establish the geometric facts about 
hypercones used in the preceding analysis. 
Let $\MI$ be Minkowski space equipped 
with the metric $(+,-,-,-)$ and coordinates 
$x = (x_0,\bx)$. 
We fix the forward lightcone 
\mbox{$\FL \doteq \{ x \in \RR^4 : x_0 > |\bx| \}$}  
and regard it as a globally hyperbolic spacetime
with metric inherited from Minkowski space. Recall that 
$X^c \subset V$ denotes the spacelike complement of any subset
$X \subset V$.
The lightcone $V$ is foliated by the hyperboloids  
$\HY_\tau\doteq \{ x \in \FL : x_0 = \sqrt{\bx^2 + \tau^2} \}$ 
(time shells) for~$\tau > 0$. 
Since the hyperboloids $\HY_\tau$ form Cauchy surfaces of $V$, the
causal completions of disjoint sets on a given $\HY_\tau$ are spacelike 
separated regions in~$V$. 

Fixing $\tau$ and abbreviating $\HY = \HY_\tau$
we consider specific subsets of the 
corresponding hyperboloid as bases of causally complete 
regions in $\FL$. To have a Lorentz invariant description
of these sets, we equip $\HY$ with the metric induced from
the ambient space.
Given two points $\ha, \hb \in \HY$, the geodesic connecting them
is the segment of the ``great hyperbola'' got by intersecting
$\HY$ with the $2$--plane fixed by $\ha,\hb,0$ in the ambient space. 
Its length is 
$d(\ha,\hb) = \tau \cosh^{-1}(\ha \hb/\tau^2)$, where $\ha \hb$ denotes the 
Lorentz scalar product of $\ha$ and $\hb$. 
Great hyperbolae on $\HY$ thus correspond to lines and will be called
hyperbolic lines.  

To further geometric 
intuition we project the hyperboloid $\HY$ 
through the origin onto the plane~$x_0 = 1$ and thereby
identify it with the open unit ball $\HB \subset \RR^3$ about the origin.
The corresponding invertible map $\bv : \HY \rightarrow \HB$
is given by 
$$
\bv(a)  \doteq \ba / a_0 =  \ba /\sqrt{\ba^2 + \tau^2} \, ,
$$
inducing on $\HB$ the metric  
$d(\bu, \bv) = \tau 
\cosh^{-1}((1 - \bu \bv) / \sqrt{(1 - \bu^2)(1 - \bv^2)})$,
where $\bu \bv$ denotes the Euclidean scalar product.    
The result of this mapping  
is the Beltrami--Klein model of hyperbolic geometry.
Its simplifying feature is that the hyperbolic lines 
on $\HY$ are mapped to chords 
of~$\HB$, \ie straight lines connecting boundary points of $\HB$. 
Note that the spherical boundary $S^2$ of~$\HB$ corresponds to spacelike 
infinity on the hyperboloid. 

Hyperbolic rays on $\HY$ are fixed by specifying their apex 
$\ha$ and their asymptotic lightlike direction $l = (1,\bl)$,
where $\bl \in S^2$. A union of hyperbolic rays emanating from a common 
apex~$\ha$ is the analogue of a cone and is called a hyperbolic cone. The 
opposite hyperbolic cone results by taking the union of the 
corresponding opposite rays emanating from $\ha$. 
A hyperbolic cone is said to be pointed if its closure
and the opposite closed cone have only the apex in common.  
A hyperbolic cone $\HC \subset \HY$ is said to be convex if, for any
two points $\ha, \hb \in \HC$, the geodesic connecting them 
is contained in $\HC$. 

Proceeding to the 
Beltrami--Klein model, a hyperbolic ray on $\HY$, fixed by its
apex $\ha$ and asymptotic lightlike direction $l = (1,\bl)$, corresponds to 
the straight line between the apex $\bv(a) \in \HB$ and the 
boundary point $\bl \in S^2$. Thus a hyperbolic cone 
on~$\HY$ corresponds to an ordinary (truncated) Euclidean cone 
$\HK \subset \HB$ and the concepts 
of opposite, pointed and convex hyperbolic cone  
coincide in the Beltrami--Klein model with those familiar  
from Euclidean geometry. We therefore characterize the 
hyperbolic cones  $\HC \subset \HY$  
by their images $\HK \subset \HB$ in the Beltrami--Klein model,
\mbox{$\HC = \HC(\HK)$}, where we restrict attention to pointed open convex 
cones $\HK \subset \HB$ with elliptical base. These form a 
Lorentz invariant family. We will also consider \mbox{hyperbolic} balls 
$\HO \subset  \HY$, \ie open balls with arbitrary 
hyperbolic diameter and apex. Their (ellipsoidal) 
images in $\HB$ will be denoted by the same symbol.

\begin{figure}[h] 
\centering 
\epsfig{file=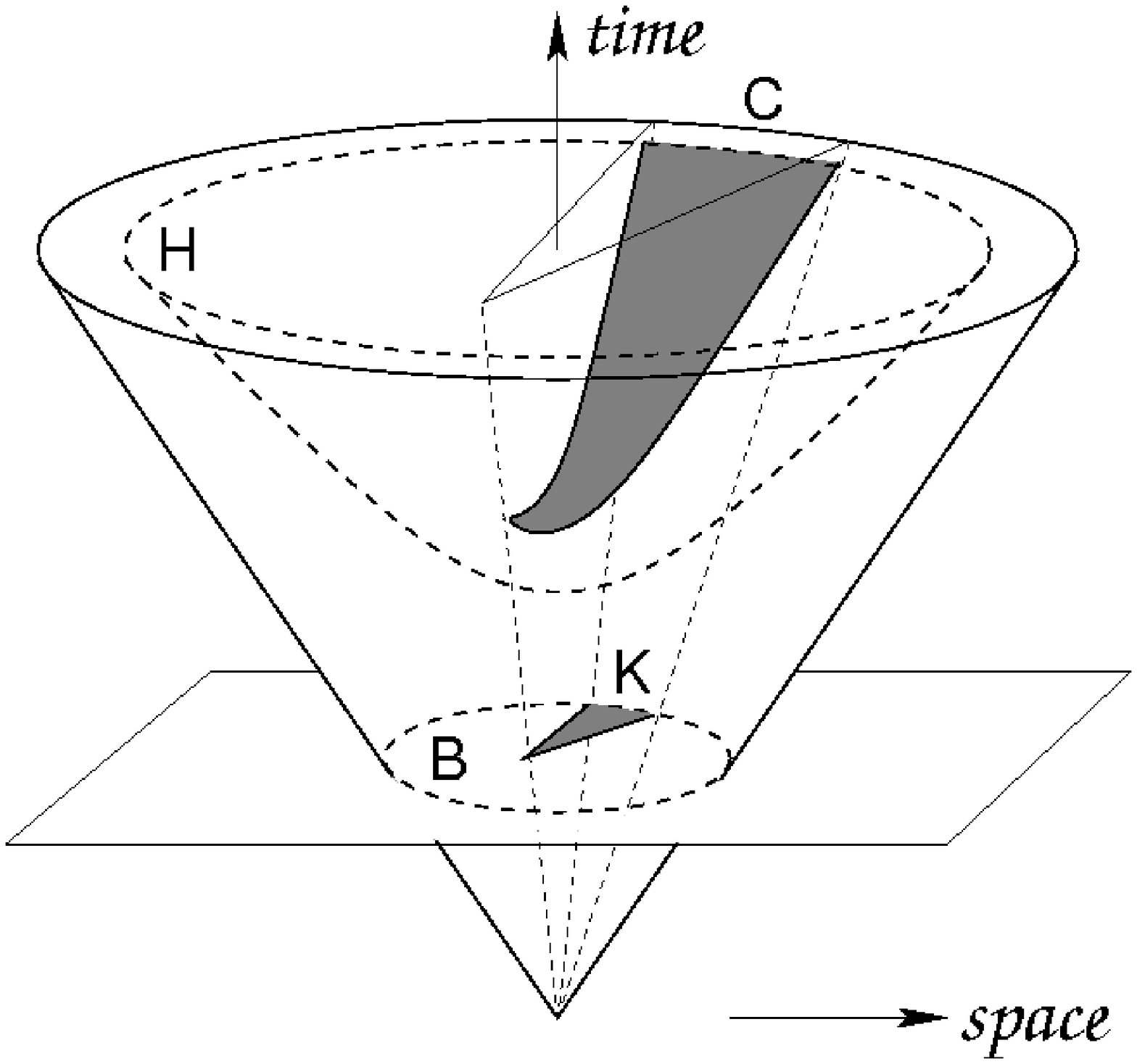,width=90mm}
\caption{Hyperbolic cone $\HC$ on the hyperboloid $\HY$ and 
its image $\HK$ in the ball $\HB$} 
\label{fig1} 
\end{figure}

\begin{define} \label{D.1} 
Let $\HY \subset V$ be a fixed hyperboloid. 
A hyperball $\cO \subset V$ is the causal 
completion of a hyperbolic ball $\HO \subset \HY$. 
A hypercone $\cC(\HK) \subset V$ is the causal
completion of a hyperbolic cone $\HC(\HK) \subset \HY$,
where $\HK \subset \HB$ is any pointed open convex cone with 
elliptical base. The family of these hypercones is denoted by
$\cF$.
\end{define}

\noindent Note that the hypercones $\cC(\HK) \in \cF$,
\mbox{$\HK \subset \HB$}, inherit the structure of 
a partially ordered set from the underlying ordinary cones, 
\ie $\cC(\HK_1) \subset \cC(\HK_2)$
iff $\HK_1 \subset \HK_2$. Moreover, the spacelike separation of
hypercones, $\cC(\HK_1) \subset \cC(\HK_2)^c$, is equivalent to 
disjointness of the underlying ordinary cones, 
$\HK_1 \bigcap \HK_2 = \emptyset$. It should also be noted that
any double cone in $V$ is contained in some hyperball. 
These facts greatly simplify the subsequent discussion. 

\subsection*{Topology of hypercones} 

In this subsection we list some topological properties of 
hypercones. We begin by discussing the funnels of hypercones
entering into the extension of morphisms. 

\begin{faktum} \label{A.1} 
Given any hypercone $\cC$ there is a 
decreasing sequence of hypercones (a funnel) $\cC_n \subset \cC$, 
$n \in \NN$ such that $\cC_n{}^c \nearrow V$. 
\end{faktum}

\begin{proof} Let $\cC = \cC(\HK)$ and pick any 
decreasing sequence of cones $\HK_n \subset \HK$, $n \in \NN$ 
such that $\bigcap_n \HK_n = \emptyset$. The resulting sequence 
of hypercones $\cC(\HK_n)$, $n \in \NN$ has the required 
properties.  
\end{proof}

\begin{faktum} \label{A.2} 
Let $\cC_n$,  $n \in \NN$ be an increasing sequence 
of hypercones such that $\cC_n \nearrow V$. The corresponding
(decreasing) sequence of opposite hypercones is a funnel in the 
preceding sense.
\end{faktum}

\begin{proof} 
Let $\cC_n = \cC(\HK_n)$, where $\HK_n \subset \HB$ is an increasing sequence of
cones such that $\HK_n \nearrow \HB$, $n \in \NN$. The respective
opposite cones $\HK_n^o$, $n \in \NN$ form a decreasing sequence and
$\bigcap_n \HK_n^o = \emptyset$. Since the hypercone 
opposite to $\cC_n$ is given by 
$\cC_n^o = \cC(\HK_n^o)$, $n \in \NN$ the latter 
sequence forms a funnel, as claimed. 
\end{proof} 

What is important when analyzing morphisms is to be flexible 
in choosing their localization. The following geometrical fact 
is then significant. 

\begin{faktum} \label{A.3} If $\cO$ is a double cone and 
$\cC$ a hypercone, then there is a hypercone 
$\cC_0 \subset \cO^c \bigcap \cC$. 
\end{faktum}

\begin{proof} 
As a double cone is contained in some hyperball we 
may assume that~$\cO$ is a hyperball.
Let $\HO \subset \HB$ be the base of $\cO$, 
let $\cC = \cC(\HK)$ and let $\HL \subset S^2$ be the
boundary of $\HK$. 
Since $\HO$ is relatively compact in $\HB$
there is a cone $\HK_0 \subset \HK$ whose apex 
is sufficiently close to $\HL$ such that 
$\HK_0 \bigcap \HO = \emptyset$. The hypercone 
$\cC_0 = \cC(\HK_0)$ has the desired properties. 
\end{proof} 

The following result is of a similar nature.

\begin{faktum} \label{A.4} If $\cC$ is a hypercone
and $\cO \subset \cC^c$ a hyperball then there is 
a hypercone $\cC_0$ with $\cO \subset \cC_0 \subset \cC^c$.
\end{faktum}

\begin{proof} 
Let $\cC = \cC(\HK)$ and let $\HO \subset \HB \backslash \HK$ be
the base of the given hyperball. Since $\HO$ and $\HK$ are 
disjoint convex
subsets of $\HB$ there is a hyperplane separating them. 
Pick a point $\ba_0$ on this hyperplane which does not
touch $\HO$ 
and consider the cone $\HK_0 \subset \HB$ generated by all rays with apex 
$\ba_0$ which pass through the points of $\HO$. Then 
$\HK_0 \supset \HO$ and $\HK_0 \bigcap \HK = \emptyset$. 
The hypercone $\cC_0 = \cC(\HK_0)$ has the desired properties. \end{proof}

To analyze morphisms one has to interpolate
between different hyperbolic cones by passing alternately to suitable 
smaller and  larger cones. A given subfamily 
\mbox{$\cF_0 \subset \cF$} of hyperbolic cones is said to be 
pathwise connected if for any pair \mbox{$\cC_a, \cC_b \in \cF_0$} 
there is an interpolating sequence  
$\cC_a = \cC_1, \, \cC_2, \dots, \cC_n = \cC_b \in \cF_0$ such that
each successive pair $\cC_m, \cC_{m+1}$ contains a common 
hyperbolic cone in $\cF_0$, $m = 1, \dots, n-1$.
The following results are used in such interpolation arguments.

\begin{faktum} \label{A.5} 
The family $\cF$ of all hypercones is pathwise connected. 
\end{faktum} 
\begin{proof} The family of cones $\HK \subset \HB$ is obviously 
pathwise connected. 
Since the corresponding family of hypercones $\cC(\HK)$ is order
isomorphic to $\HK$, the result holds. 
\end{proof}

\begin{faktum} \label{A.6}
The family $\cF(\cC^c) \subset \cF$
of hypercones localized
in the spacelike complement of a given hypercone $\cC$ 
is pathwise connected. 
\end{faktum} 
\begin{proof} Let $\cC = \cC(\HK)$ for given $\HK \subset \HB$. 
The family $\cF(\cC^c)$ consists of all hypercones 
$\cC(\HK^\perp)$, where $\HK^\perp \subset \HB \backslash {\HK}$ 
is any pointed open convex cone with elliptical base. 
Let \mbox{$\HL \subset S^2$} be the boundary of $\HK$.
Given any pair of cones 
$\HK_1^\perp, \HK_2^\perp \subset \HB \backslash {\HK}$, 
their respective boundaries $\HL_1, \HL_2$ are contained
in $S^2 \backslash \HL$. Since $S^2 \backslash \HL$ is connected 
it is then clear that $\HK_1^\perp, \HK_2^\perp$
can be connected by an alternating path of cones  
whose apices stay sufficiently close to $S^2 \backslash \HL$
therefore not meeting $\HK$. Hence the family 
$\HK^\perp \subset \HB \backslash \HK$ of cones is pathwise 
connected and with it 
the corresponding family of hypercones $\cC(\HK^\perp) \in \cF(\cC^c)$.  
\end{proof}

Given any two hypercones $\cC_a, \cC_b \in \cF$ there may be no 
hypercone in their common spacelike complement 
$\cC_a{}^{\! c} \bigcap \cC_b{}^{\! c}$. (This is most easily seen in the
Beltrami--Klein model.) Yet, by proceeding to smaller cones
this difficulty can be overcome.

\begin{faktum} \label{A.7}
Let  $\cC_a, \cC_b \in \cF$ be arbitrary hypercones. Then there is
a hypercone $\cC_0 \subset \cC_a$ making the family of hypercones
$\cF(\cC_0{}^{\! c}) \bigcap \cF(\cC_b{}^{\! c})$ pathwise connected.
\end{faktum}
\begin{proof} Let $\cC_a = \cC(\HK_a)$, $\cC_b = \cC(\HK_b)$ for
given $\HK_a, \HK_b \subset \HB$ and let $\HL_a, \HL_b \subset S^2$
be the (convex) boundaries of $\HK_a, \HK_b$. 
If $\HL_a \bigcap \HL_b = \emptyset$ there is a cone 
$\HK_0 \subset \HK_a$ whose apex lies sufficiently close to 
the boundary $\HL_a$ so that $\HK_0 \bigcap \HK_b = \emptyset$. 
Let $\HL_0 \subset S^2$ be the boundary of $\HK_0$. 
Then $S^2 \backslash (\HL_0 \bigcup \HL_b)$ is connected and,
as in the preceding argument, it follows that the family of cones 
$\HK^\perp \subset \HB \backslash (\HK_0 \bigcup \HK_b)$ is pathwise
connected. 
If $\HL_a \bigcap \HL_b \neq \emptyset$ there is a cone 
$\HK_0$ with boundary
$\HL_0 \subset \HL_a \bigcap \HL_b$ and apex sufficiently close  
to the boundary with $\HK_0 \subset \HK_a \bigcap \HK_b$. 
Hence 
$\cF(\cC_0{}^{\! c}) \bigcap \cF(\cC_b{}^{\! c}) = \cF(\cC_b{}^{\! c}) $ 
and the result follows from the preceding result.
\end{proof}

It is crucial in  the statistics analysis that
any given pair of spacelike separated hypercones 
has hypercones in its spacelike complement.

\begin{faktum} \label{A.8}
Let $\cC_a, \cC_b$ be spacelike separated hypercones,
$\cC_a \subset \cC_b{}^{\! c}$. Then there is a hypercone 
$\cC_0 \subset \cC_a{}^{\! c} \bigcap \cC_b{}^{\! c}$. 
\end{faktum}

\begin{proof} Let $\cC_a = \cC(\HK_a)$, $\cC_b = \cC(\HK_b)$
where $\HK_a \bigcap \HK_b = \emptyset$ and let 
$\HL_a, \HL_b \subset S^2$ be the boundaries of 
$\HK_a$ and $\HK_b$, respectively.
Since $\HL_a, \HL_b$ are disjoint convex sets there is 
an open convex set $\HL_0 \subset S^2 \backslash (\HL_a  \bigcup \HL_b)$
and choosing an apex $\ba_0$ sufficiently close to~$\HL_0$ one obtains
a cone $\HK_0 \subset \HB \backslash (\HK_a \bigcup \HK_b)$. 
Hence $\cC_0 = \cC(\HK_0)$ meets the requirements. 
\end{proof}

\subsection*{Hypercones based on different hyperboloids}

Having clarified the structure of hypercones based on a fixed hyperboloid, 
we analyze the relations between hypercones based on 
different hyperboloids. We put
\mbox{$\HY_\tau = \{ x \in \FL : x_0 = \sqrt{\bx^2 + \tau^2} \}$}, 
$\tau > 0$ and denote the corresponding distance function by $d_\tau$. 
Points, hyperbolic balls and hyperbolic 
cones on these manifolds are denoted 
by $\ha_\tau$, $\HO_\tau$ and $\HC_\tau$ and 
their causal completions (the hyperballs and
hypercones) by $\cO_\tau$ and $\cC_\tau$, respectively.  
Points on different 
hyperboloids are identified by scaling, \ie given 
$\ha_\tau \in \HY_\tau$ we let $\ha_\sigma = (\sigma / \tau) \,
\ha_\tau \in \HY_\sigma$. The hyperbolic balls 
and cones on different hyperboloids are identified in this way.
The identifying map commutes with Lorentz transformations and preserves 
the causal relations between corresponding 
hypercones based on different hyperboloids.
It will again be convenient to identify the hyperbolic cones
$\HC_\tau \subset \HY_\tau$ with their canonical images $\HK \subset \HB$
in the Beltrami--Klein model, $\HC_\tau = \HC_\tau(\HK)$. 
Their respective causal completions are denoted by $\cC_\tau(\HK)$. 
The distinguished families of hypercones $\cF_\tau$ based on
different hyperboloids $\HY_\tau$ are identified in this way.

Given any $\ha_\sigma \in \HY_\sigma$, its 
future and past causal shadows on $\HY_\tau$ are given by 
$\HO_\tau(\ha_\sigma) \doteq 
(\ha_\sigma \pm \Trans) \bigcap \HY_\tau$ if $ \pm(\tau - \sigma) \geq 0$, 
respectively. Since the distance function $d_\tau$ is invariant
under Lorentz transformations, a convenient description of
this shadow follows from a straightforward computation: 
$\HO_\tau(\ha_\sigma) = \{ \ha \in \HY_\tau :
d_\tau(\ha, \ha_\tau) \leq \tau c_{\sigma, \tau} \}$, where
$c_{\sigma, \tau} \doteq \cosh^{-1}((\sigma^2 + \tau^2)/2 \sigma \tau)$.  
Thus the causal shadow on $\HY_\tau$ 
of any given point $\ha_\sigma \in \HY_\sigma$
is a (closed) 
hyperbolic ball about $\ha_\tau = (\tau/\sigma) \ha_\sigma \in \HY_\tau$
whose radius depends only on $\sigma$ and $\tau$. 
Incidentally, this result shows that 
the causal completion of a hyperbolic ball about some point $\ha_\tau$
is an ordinary double 
cone in $V$ whose vertices lie on the timelike ray
$\RR_+ \ha_\tau$.\smallskip

It follows from these remarks that the causal shadow 
of a hyperbolic ray on $\HY_\sigma$  
is the union of hyperbolic balls of fixed diameter
centred on the points of the 
hyperbolic ray on~$\HY_\tau$ which is the image of 
the given ray by the above identification map. Similarly,
the causal shadow of a hyperbolic cone on  $\HY_\sigma$
is the union of hyperbolic balls of fixed diameter
centred on the points of the corresponding hyperbolic 
cone on $\HY_\tau$. Proceeding to the Beltrami--Klein
model, the resulting region in $\HB$ is 
the union of hyperbolic balls with fixed radius 
centred about the points of an ordinary cone
$\HK \subset \HB$. The following result is an easy consequence.

\begin{faktum} \label{A.9}
Let $\sigma, \tau > 0$ and 
$\cC_\sigma \in \cF_\sigma$. Then there is a
$\cC_\tau \in \cF_\tau$ with
$\cC_\sigma \subset \cC_\tau$.
\end{faktum}

\begin{proof} Let $\cC_\sigma = \cC_\sigma(\HK)$ for $\HK
  \subset \HB$  and let 
$\widehat{\HK} = 
\{ \bu \in \HB : d_\tau(\bu,\bv) \leq \tau c_{\sigma, \tau}, \  
\bv \in \HK \}$   be the region corresponding to its shadow on 
$\HY_\tau$. Furthermore, let $\HK^o$ be the cone opposite to $\HK$.
It follows from the triangle inequality for the metric 
$d_\tau$ that all points in $\HK^o$ sufficiently 
close to the boundary $S^2$ of~$\HB$ (and therefore being arbitrarily 
far from $\HK$) lie in the complement
of~$\widehat{\HK}$. Hence $\widehat{\HK}$ is contained in 
the interior of a spherical cap of $\HB$ with non--trivial
complement. Picking a point $\ba_0$ in this complement as
apex and connecting it to all points of the cap by straight
lines yields a pointed convex cone $\HK_0 \subset \HB$ with 
spherical base containing $\widehat{\HK}$. The corresponding 
hypercone $\cC_\tau = \cC_\tau(\HK_0)$ based on $\HY_\tau$ has
the stated property.
\end{proof}

There is also a converse to this result. 

\begin{faktum} \label{A.10} 
Let $\sigma, \tau > 0$ and let
$\cC_\sigma \in \cF_\sigma$. Then there is a  $\cC_\tau \in \cF_\tau$ with
$\cC_\tau \subset \cC_\sigma$.
\end{faktum}

\begin{proof} Let $\cC_\sigma = \cC_\sigma(\HK)$ for $\HK
  \subset \HB$. Picking any cone $\HK_0 \subset \HK$ whose 
hyperbolic distance from the boundary of $\HK$ is larger than
$\tau c_{\sigma, \tau}$ ($\HK_0$ must have
a sufficiently small opening compared to 
$\HK$ and an apex  
sufficiently close to the boundary $S^2$ of~$\HB$)
it follows that the causal shadow 
$\widehat{\HK}_0 = \{ \bu \in \HB : d_\tau(\bu,\bv) \leq \tau c_{\sigma, \tau}, \  
\bv \in \HK_0 \}$  
of $\HK_0$ is contained in $\HK$ since $\HK$ is convex. The corresponding 
hypercone $\cC_\tau = \cC_\tau(\HK_0)$ based on $\HY_\tau$ 
fulfils the requirements. \end{proof}

The preceding results imply that, as is needed in the main text,  
 the structure of hypercone localized 
morphisms does not depend on the choice of a hyperboloid $\HY_\tau
\subset V$ and the corresponding family $\cF_\tau$ of hypercones. 

\subsection*{Spacetime transformations of hypercones}

We now study spacetime transformations
of hypercones. To this end we again fix a hyperboloid $\HY \subset V$
and consider the corresponding distinguished family $\cF$ of
hypercones based on it. The following results are relevant to 
the discussion of covariance properties of morphisms.

\begin{faktum} \label{A.11} 
Let $\cC$ be a hypercone. Then 
there are an open set of directions $\bl \in S^2$ and
sequences of boosts $\Lambda_n(\bl)$ in these directions, 
\ie $\Lambda_n(\bl) \, (1,\bl) = e^{n} \, (1,\bl)$, such that   
$\Lambda_n(\bl) \, \cC \subset \cC$, $n \in \NN$. 
Moreover, given any double cone $\cO$,  
$\Lambda_n(\bl) \, \cC \subset \cO^c$ for sufficiently large $n$. 
\end{faktum}

\begin{proof} Let $\cC = \cC(\HK)$ for given $\HK \subset \HB$. As
the set of hypercones is stable under Lorentz transformations we may 
suppose the apex of~$\HK$ is the centre
of $\HB$.
Let $\HL \subset S^2$ be the boundary of $\HK$ and let 
$\bl \in \HL$. There are boosts $\Lambda_n(\bl)$ with
\mbox{$\Lambda_n(\bl) \, (1, \bl) = e^{n} \, (1, \bl)$}, $n \in \NN$,  
inducing the action 
$$
\bv \mapsto {\frac{1}{\mbox{ch}(n) + \mbox{sh}(n) \, \bv \bl}} \,
\big( (\mbox{sh}(n) + \mbox{ch}(n) \, \bv \bl) \, \bl + \bv^\perp \big) \, ,
$$  
on $\HB$, where $\bv^\perp$ is the component of $\bv$ orthogonal
to $\bl$. Thus the apex $\bv = 0$ 
of~$\HK$ is shifted in the direction of $\bl$ 
into the cone~$\HK$ and the points of $\HL$ are moved along 
great circles towards $\bl \in \HL$. Since~$\HL$ is convex,  
the boosts induce contractions of~$\HL$ and consequently 
map $\HK$ into itself. Hence 
$\Lambda_n(\bl) \, \cC \subset \cC$, $n \in \NN$. It now 
suffices to consider hyperballs $\cO$.
Let $\HO \subset \HB$ be the base of $\cO$. Since the boosts 
move $\HK$ towards the boundary $\HL$ and~$\HO$ 
is relatively compact, the two regions become disjoint 
for sufficiently large $n \in \NN$. 
But this implies  
\mbox{$\Lambda_n(\bl) \, \cC \subset \cO^c$}, completing the proof.  
\end{proof}

The next result is of a similar nature.

\begin{faktum} \label{A.12} 
Given a hypercone $\cC$, there is a hypercone $\cC_0$ 
with $\Lambda \, \cC \subset \cC_0$ for all Lorentz transformations
$\Lambda$ in some open neighbourhood $\cN^L \subset \Lore$ of the identity. 
\end{faktum} 

\begin{proof}
Let $\cC = \cC(\HK)$ for given $\HK \subset \HB$. The action 
of arbitrary Lorentz transformations $\Lambda \in \Lore$ 
induced on $\HB$ is given by (choosing coordinates properly)
$$
\bv_i \mapsto  
{\frac{\Lambda_{i 0} + \sum_k  \Lambda_{ik} \bv_k}{\Lambda_{00} + 
\sum_k \Lambda_{0k} \bv_k}} \, , \quad i = 1,2,3 \, .
$$
Since $|\Lambda_{00} + \sum_k  \Lambda_{0k} \bv_k| \geq
\big(\sqrt{1 + \sum_k \Lambda_{0k}^2} - \sqrt{\sum_k \Lambda_{0k}^2}
\,\big) > 0$
this induced action is norm continuous on $\HB$ in the 
Euclidean topology. Hence, choosing $\Lambda$ in some sufficiently
small neighbourhood $\cN^L \subset \Lore$ of the identity, the
corresponding transformed cones $\HK_\Lambda$ 
all lie in some sufficiently large
cone $\HK_0 \subset \HB$. The hypercone $\cC_0 = \cC(\HK_0)$ therefore has
the required property.  
\end{proof}

We now study the action of time translations on hypercones. 
The translated hypercones are no longer based on any of the   
hyperboloids foliating $V$, but 
they are still contained in sufficiently large hypercones 
based on the given $\HY$. 

\begin{faktum} \label{A.13} 
Let $\cC$ be a hypercone and let $\cB^T \subset \Trans$ be a 
bounded set of translations. Then there is a
hypercone $\cC_0$ with $\cC + \cB^T \subset \cC_0$. 
\end{faktum}

\begin{proof}
We first suppose $\cC$ to be the 
causal completion of a special type of hyperbolic cone $\HC \subset \HY$ 
made up of hyperbolic rays given, choosing coordinates suitably, 
by 
\mbox{$u \mapsto a(u) \doteq 
(u \tau + v(u), v(u) \bl)$, $0 < u \leq 1$,} with 
$v(u) = \tau (1 - u^2)/2u$ and $\bl \in \HL$, 
where $\HL \subset S^2$ is any convex set whose closure 
is contained in a hemisphere. 
The opposite cone is made up of the rays 
$u^\prime \mapsto a^\prime(u^\prime) \doteq 
(u^\prime \tau + v(u^\prime), v(u^\prime) \bl^\prime)$, $0 < u^\prime
\leq 1$, with $\bl^\prime \in - \HL$. Let 
$t = (\mathrm{t}, \mathbf{0})$, $\mathrm{t} \geq 0$,  be a time translation. 
A straightforward computation gives 
\begin{equation*}
\begin{split}
& (a(u) + t - a^\prime(u^\prime))^2 \\
& = \mathrm{t}^2 + 2 \tau^2 + 2 \mathrm{t} (u \tau + v (u)) 
- 2 (\mathrm{t} + u \tau + v(u)) (u^\prime \tau + v (u^\prime)) 
+ 2 v(u) v(u^\prime) \, \bl \bl^\prime \, . 
\end{split} 
\end{equation*}
Thus, if $(u^\prime \tau + v(u^\prime)) > (\tau + \mathrm{t})$,
the points $(a(u) + t)$ and $a^\prime(u^\prime)$ are spacelike
separated, $(a(u) + t - a^\prime(u^\prime))^2 < 0$, for 
any $0 < u \leq 1$, $\bl \in \HL$ and
$\bl^\prime \in - \HL$; note that $\bl \bl^\prime < 0$. 
Phrased differently, the causal shadow
of $\HC + t$ on $\HY$ is disjoint from all spacetime points in the 
cone opposite to $\HC$ for times larger than $(\tau + \mathrm{t})$. 
Proceeding to the Beltrami--Klein model this implies 
that the image of this shadow in $\HB$ is contained
in the interior of a spherical cap of $\HB$ with 
non--trivial complement. As in the proof of 
\ref{A.9}, the shadow fits into a cone
$\HK_0 \subset \HB$, hence $\cC + t \subset \cC(\HK_0) = \cC_0$. 

Now let $\cC$ be any hypercone. A suitable Lorentz 
transformation $\Lambda$ shifts
its apex to the point $(\tau, \mathbf{0})$ and
$\Lambda \cC$ is then 
of the special type considered in the preceding step. The given set of 
translations is mapped to $\Lambda \, \cB^T$. Since 
$\cB^T$ is bounded, there is a time translation $t$ as above
with $\Lambda \, \cB^T \subset (t - \Trans) \bigcap \Trans$.  
Since the points 
$\HC + (t - \Trans) \bigcap \Trans$ on the hyperboloid $\HY$ 
have the same causal shadow as $\HC + t$,  
the hypercone $\cC_0$ constructed in the preceding step
satisfies $\Lambda (\cC + \cB^T) \subset \cC_0$. Hence 
$\Lambda^{-1} \cC_0$ is as required. 
\end{proof}
\end{appendix} 

\subsection*{\Large Acknowledgements} 
This research was begun in collaboration with Sergio Doplicher. 
We express our thanks to him for numerous discussions, 
many challenging remarks and for his continuing interest 
in this work. DB would also like to thank the Universit\`a di Roma
for hospitality and financial support at various stages of this work. 
JER would like to thank the Fachbereich Mathematik der Universit\" at 
G\" ottingen for their hospitality  during the closing stages of this 
research.

\end{document}